\documentclass[10pt,a4paper]{amsart}
\usepackage{geometry}
\usepackage{graphicx}
\usepackage{amssymb}
\usepackage{pgfplots}
\usepackage{amsmath, mathrsfs, stmaryrd, color,slashed,bbm}
\newtheorem{theorem}{Theorem}
\usepackage{color}
\usepackage[title]{appendix}
\newtheorem{remark}[theorem]{Remark}

\newtheorem{definition}[theorem]{Definition}
\newtheorem{proposition}[theorem]{Proposition}
\newtheorem{corollary}[theorem]{Corollary}
\newtheorem{lemma}[theorem]{Lemma}

\newcommand{\GG}[1]{}
\newcommand{\LW}{\mathsf{LW}}
\newcommand{\HG}{\mathsf{HG}}
\newcommand{\s}{\slashed}

\newcommand{\nablas}{\slashed{\nabla}}

\newcommand{\Div}{\textup{div}}

\newcommand{\even}{\textup{even}}
\newcommand{\RW}{\textup{RW}}
\newcommand{\done}{\slashed{D}_1}
\newcommand{\dtwo}{\slashed{D}_2}
\newcommand{\donest}{\slashed{D}^*_1}
\newcommand{\dtwost}{\slashed{D}^*_2}
\newcommand{\h}{\mathsf{H}}
\newcommand{\W}{\mathsf{W}}
\newcommand{\cblue}{\color{black}}
\newcommand{\cred}{\color{black}}
\newcommand{\cpur}{\color{black}}

\newcommand{\f}{\mathcal{D}}
\newcommand{\K}{\mathcal{K}}

\newcommand{\D}{\left(1-\frac{2M}{r}\right)}

\numberwithin{theorem}{section}
\numberwithin{equation}{section}

\begin{document}
\title[Schwarzschild linear stability in harmonic gauge]{The linear stability of the Schwarzschild spacetime in the harmonic gauge: even part}

\author[Pei-Ken Hung]{Pei-Ken Hung}
\address{Pei-Ken Hung, Massachusetts Institute of Technology,
Department of Mathematics,
182 Memorial Drive,
Cambridge, MA 02139, USA}
\email{pkhung@mit.edu}

\begin{abstract}
In this paper we study the even part of the linear stability of the Schwarzschild spacetime as a continuation of \cite{Hung-thesis}. By taking the harmonic gauge, we prove that the energy decays at a rate $\tau^{-2+}$ for the solution of the linearized Einstein equation after subtracting its spherically symmetric part. We further show that the spherically symmetric part converges to a linear combination of two special solutions. One is the gauge-fixed mass change solution \cite{Hafner-Hintz-Vasy}. The other is the deformation tensor of a stationary one form, which solves the tensorial wave equation. As a key ingredient, we prove that the solutions of the tensorial wave equation converge to this stationary one form up to a scalar multiplication. 
\end{abstract}
\maketitle

\section{Introduction}

Einstein's general theory of relativity describes how the metric on the spacetime, a 4-dimensional Lorenzian manifold, interacts with the matter fields. When there is no matter field, the theory reduces to the study of the vacuum Einstein equation, which is equivalent to the Ricci flat equation on the metric $g$:

\begin{align}\label{VEE}
Ric(g)=0.
\end{align}

Choquet-Bruhat and Choquet-Bruhat-Geroch \cite{Choquet-Bruhat, Choquet-Bruhat-Geroch} formulated equation \eqref{VEE} as a Cauchy problem with the initial data being a triple $(\Sigma,g^3_{ij},k_{ij})$. The triple consists of a 3-dimensional manifold $\Sigma$, a Remannian metric $g^3_{ij}$ and a symmetric two tensor $k_{ij}$. If and only if the constraint equations (Gauss and Codazzi equations) are satisfied, there exist a vacuum spacetime $(\mathcal{M},g)$, a 4-dimensional manifold equipped with a Lorentzian metric $g$ solving \eqref{VEE}, and an embedding $i:\Sigma\to\mathcal{M}$ with $g^3$ and $k$ being the pulled back induced metric and second fundamental form respectively. This can be viewed as the local (in time) existence theorem for \eqref{VEE}.

With the local existence of \eqref{VEE}, it is of great attraction to study the long-time behavior of the solution. A major open problem in this direction is the stability conjecture of the Kerr spacetimes $(\mathcal{M}_{M,a},g_{M,a})$, which is a family of stationary solutions of \eqref{VEE}. It is believed that the Kerr family is stable as solutions of \eqref{VEE}: for initial data $(\Sigma,g^3_{ij},k_{ij})$ close to one from a Kerr spacetime, the solution $g$ has ``long-time existence" and converges to a member of the Kerr family. We refer readers to \cite[section 5.6]{Dafermos-Rodnianski_note} for a more precise statement of the conjecture.

One needs to understand the linearized equation of $\eqref{VEE}$ before studying the stability problem. In the Kerr background, the linearized equation is a linear equation on symmetric two tensors $h$, which reads
\begin{equation}\label{linear_gravity}
\begin{split}
\frac{d}{ds} Ric (g_{M,a}+sh)\bigg|_{s=0}=0.
\end{split}
\end{equation}
There are two types of special solutions of \eqref{linear_gravity}. The first is a 4-dimensional vector space which comes from the perturbation within the Kerr family. The second, being infinite dimensional, consists of deformation tensors which correspond to infinitesimal diffeomorphisms. As a consequence of the Kerr stability conjecture, the solution $h$ is believed to decay to zero up to these special solutions, which is known as the linear stability problem of Kerr.

In this paper, we study the linear stability of the \textbf{Schwarzschild spacetime}. To deal with the infinite dimensional deformation tensors, we impose the harmonic gauge condition:

\begin{align}\label{HG}
\Gamma_b[h]:=\nabla^a\left(h_{ab}-\frac{1}{2}({tr} h)g_{ab}\right)=0.
\end{align}
The harmonic gauge is the linearization of the harmonic map gauge used in \cite{Choquet-Bruhat, Choquet-Bruhat-Geroch}. Lindblad-Rodnianski \cite{Lindblad-Rodnianski-05, Lindblad-Rodnianski-10} and Hintz-Vasy \cite{Hintz-Vasy-1} also adapted the harmonic map gauge in proving the stability of the Minkowski spacetime. These accomplishments in the harmonic map gauge is the main motivation to study the harmonic gauge in the Schwarzschild background. Under the harmonic gauge condition \eqref{HG}, the linearized equation (\ref{linear_gravity}) is equivalent to the Lichnerowicz d'Alembertian equation:
\begin{equation}\label{Main_equation}
\begin{split}
\Box h_{ab}+2R_{a\ b}^{\ c\ d}h_{cd}=0,
\end{split}
\end{equation}
where $R_{a\ b}^{\ c\ d}$ is the (index raised) Remannian curvature tensor of $g$. The equation \eqref{Main_equation} is a wave equation on symmetric two tensors and hence has well-posed Cauchy problem. In particular, we have long-time existence of solutions for regular initial data. Moreover, the gauge condition (\ref{HG}) is preserved: for any solution $h$ of the Lichnerowicz d'Alembertian equation \eqref{Main_equation}, the gauge one from $\Gamma_b[h]$ satisfies the tensorial wave equation: 
\begin{align*}
\Box \Gamma_a[h]=0.
\end{align*}
Therefore $\Gamma_b[h]$ vanishes identically provided $\Gamma_b[h]$ and its normal derivative vanish initially. We remark that there are sitll infinite dimensional deformation tensors satisfying \eqref{Main_equation}. Since a deformation tensor ${}^W\pi$ is always a solution of \eqref{linear_gravity}, it solves \eqref{Main_equation} provided \eqref{HG} holds, which is equivalent to the tensorial wave equation on the potential one form:
\begin{align}\label{wave_equation_vector}
\Box W_a=0.
\end{align}

We investigate the equations \eqref{HG}, \eqref{Main_equation} for \textbf{even} symmetric two tensors and the equation \eqref{wave_equation_vector} for \textbf{even} one forms as the odd part was studied in \cite{Hung-thesis}. See subsection \ref{subsec:evenodd} for the even/odd decomposition. Denote by $W_{\ell=0}$ or $h_{\ell=0}$ the spherically symmetric part of a one form $W$ or a symmetric two tensor $h$ respectively. We show that solutions of \eqref{HG} and \eqref{Main_equation} or \eqref{wave_equation_vector} decay to zero after subtracting $W_{\ell=0}$ or $h_{\ell=0}$. See Theorem \ref{thm:ell_geq_2} and \ref{thm:ell_geq_1} for the precise statements of the following results.
\begin{theorem}
Let $W=W_adx^a$ be a solution of \eqref{wave_equation_vector} with initial data falling off fast enough. Then the energy of $W-W_{\ell=0}$ decays at a rate $\tau^{-2+}$.
\end{theorem}
\begin{theorem}
Let $h=h_{ab}dx^adx^b$ be a solution of \eqref{HG} and \eqref{Main_equation} with initial data falling off fast enough.. Then the energy of $h-h_{\ell=0}$ decays at a rate $\tau^{-2+}$.
\end{theorem}
The equation \eqref{wave_equation_vector} admits a spherically symmetric and stationary solution $W^*$ with finite initial energy \cite{Hafner-Hintz-Vasy}. See \eqref{def:W*} in section \ref{sec:ell=0} for the explicit form of $W^*$. We show that spherically symmetric solutions of \eqref{wave_equation_vector} converge to a multiple of $W^*$. See Theorem \ref{thm:vector,l=0} for the precise statement of the following result.
\begin{theorem}
Let $W=W_adx^a$ be a spherically symmetric solution of \eqref{wave_equation_vector} with initial data falling off fast enough. Then $W$ converges to $cW^*$ for some constant $c$.
\end{theorem}
The equations \eqref{HG} and \eqref{Main_equation} have two spherically symmetric and non-decaying solutions. One solution $K^*$ is the gauge-fixed mass perturbation \cite{Hafner-Hintz-Vasy}, which grows linearly. See subsection \ref{subsec:sad} for the explicit form of $K^*$. The other is the deformation tensor of $W^*$, which is stationary. We show that under certain conditions, spherically symmetric solutions of \eqref{HG} and \eqref{Main_equation} converge to a linear combination of these two. See Theorem \ref{thm:sad}  for the precise statement of the following result.
\begin{theorem}
Let $h=h_{ab}dx^adx^b$ be a spherically symmetric solution of \eqref{HG} and \eqref{Main_equation} with initial data falling off fast enough. Then $h$ converges to $c_1 K^*+c_2 ({}^{W^*}\pi)$ for some constants $c_1$ and $c_2$.
\end{theorem}
{\color{white}.}\\
\textbf{Related Work.}\\

There have been many progresses towards the Kerr stability conjecture. The stability of the Minkowski spacetime was proved in the monumental work of Christodoulou and Klainerman \cite{Christodoulou-Klainerman}. See also \cite{Klainerman-Nicolo, Bieri, Lindblad-Rodnianski-05, Lindblad-Rodnianski-10, Hintz-Vasy-1} for various approaches. Recently, the stability of the Schwarzschild spacetime was established by Klainerman and Szeftel \cite{Klainerman-Szeftel} for axial symmetric polarized perturbations. In the positive cosmological constant setting, the stability of Kerr-de Sitter with small angular momentum was proved by Hintz and Vasy \cite{Hintz-Vasy-dS}.
 

The study of equation \eqref{linear_gravity} on the Schwarzschild background was initiated by Regge and Wheeler \cite{Regge-Wheeler}. The authors performed the even/odd decomposition and derived the Regge-Wheeler equation for the odd solutions of (\ref{linear_gravity}). For even solutions, there is a similar equation discovered by Zerilli \cite{Zerilli}. Bardeen and Press \cite{Bardeen-Press} adapted the Newman-Penrose formalism to study equation (\ref{linear_gravity}). This approach was extended to Kerr spacetimes by Teukolsky \cite{Teukolsky}, showing that the extreme Weyl curvature components satisfy the Teukolsky equations. In the Schwarzschild spacetime, the transformation theory of Wald \cite{Wald} and Chandrasekhar \cite{Chandrasekhar} relates Regge-Wheeler-Zerilli-Moncrief system to the Teukolsky equations. See also \cite{Aksteiner-Andersson-Backdahl-Shah-Whiting} for further refinement of the transformation theory. These works accumulated to the proof of mode stability for Kerr by Whiting \cite{Whiting}.

A major progress which goes beyond mode stability is the work of quantitative linear stability of Schwarzschild by Dafermos, Holzegel and Rodnianski \cite{Dafermos-Holzegel-Rodnianski}. The authors proved the boundedness and decay estimates for the Regge-Wheeler equation and then for Teukolsky equations through transformation theory. The metric perturbation was reconstructed in the double null gauge. Keller, Wang and the author \cite{Hung-Keller-Wang} worked in the mixed Regge-Wheeler/Chandrasekhar gauge and shown $t^{-1/2}$ decay of the metric coefficient based on Regge-Wheeler/Zerilli equations. Johnson further \cite{Johnson1,Johnson_2} proved $t^{-1}$ decay through Regge-Wheeler/Zerilli equations with an insightful chosen generalized wave gauge in which the metric perturbation is related to Regge-Wheeler/Zerilli quantities by pseudo-differential operators. See also \cite{Hung-thesis} for the $t^{-1+}$ decay of the odd part in harmonic gauge. For the linearized Einstein-Maxwell equations, Giorgi \cite{Elena-combine, Elena-spin2,Elena-spin1,Elena-linearstability} obtained the boundedness and decay estimates for the Teukolsky system in Reissner-Nordstr\"om spacetimes with small charge and proved linear stability under the gravitational-electromagnetic perturbations.

Recently, there are huge breakthroughs for the linear stability of \textbf{Kerr} spacetimes. Andersson, B\"ackdahl, Blue and Ma \cite{Andersson-Backdahl-Blue-Ma} established linear stability of Kerr with small angular momentum with pointwise $t^{-3/2+}$ decay in the outgoing radiation gauge. Hafner, Hintz and Vasy \cite{Hafner-Hintz-Vasy} gave a detailed description of the metric perturbation in Kerr spacetimes with small angular momentum under the wave gauge of the present paper. The authors identified 7-dimensional stationary and additional 4-dimensional linear growth solutions of \eqref{HG} and \eqref{linear_gravity}; these solutions consist of 4-dimensional Kerr family perturbations and 7-dimensional deformation tensors. Under weaker initial fall-off condition than the present paper, The authors further established $t^{-1-}$ decay upto these 11-dimensional space.\\

\textbf{Outline.} In section \ref{sec:schwarzschild} we introduce the Schwarzschild spacetime and relevant notations. Section \ref{sec:lemma} contains simple lemmas that we use repeatedly. In section \ref{sec:4} we  apply the vector field method to \eqref{Main_equation} near the horizon and null infinity. Then a gauge transformation from harmonic gauge to Regge-Wheeler gauge is performed. The transformation satisfies \eqref{wave_equation_vector} with source terms on the right hand side. In section \ref{sec:W0} and section \ref{sec:W12} we start to analyze equation \eqref{wave_equation_vector} supported on $\ell\geq 1$, which leads to the main results for $\ell\geq 1$ in section \ref{sec:proof}. The case $\ell=0$ is considered in section \ref{sec:ell=0}.

\section*{Acknowledgments}
The author is grateful to Simon Brendle for suggesting this problem and for his initial contribution. The author also thanks Sergiu Klainerman and Mu-Tao Wang for their encouragement. The author further thanks T\"ubingen University where part of this work was carried out.
\section{Schwarzschild spacetime and notations}\label{sec:schwarzschild}

\subsection{Schwarzschild coordinate}\label{subsec:Schwarzschild}

In this subsection we introduce the Schwarzschild coordinate and the vector fields that we need later. Let $M>0$ be a fixed constant. The Schwarzschild metric with mass $M$ can be written as

\begin{equation*}
g_M=-\D dt^2+\D^{-1}dr^2+r^2(d\theta^2+\sin^2\theta d\phi^2).
\end{equation*}
The range of the $(t,r,\theta,\phi)$ coordinate is $t\in \mathbb{R}$, $r>2M$ and $(\theta,\phi)\in \mathbb{S}^2$. To see the scaling of energies clearly, we define $s:=r/M$. We suppress the dependence of $M$ in $g_M$ and denote it by $g$. Let $\nabla$ be the the Levi-Civita connection of $g$. We use $x^a,x^b$ to denote a spacetime coordinate, $x^A,x^B$ to denote the spherical coordinate on $\mathbb{S}^2$ and $x^\alpha,x^\beta$ to denote the quotient coordinate on $\mathbb{R}_t\times (2M,\infty)_r$. Let $\s{g}_{AB}dx^Adx^B:= r^2(d\theta^2+\sin^2\theta d\phi^2)$ be the induced metric on the orbit spheres and $\nablas$ be the Levi-Civita connection of $\s{g}$. For any scalar function $\psi$, we denote 

\begin{align*}
|\nablas\psi|^2:=&\s{g}^{AB}\nablas_A\psi \nablas_B\psi=\frac{1}{r^2}\left( \left|\frac{\partial\psi}{\partial\theta}\right|^2+\sin^{-2}\theta \left|\frac{\partial \psi}{\partial\phi}\right|^2 \right).
\end{align*}
The $(t,r,\theta,\phi)$ coordinate system has coordinate singularity at $r=2M$. To remove the singularity and to include the future horizon, it is convenient to work with the $(v,R,\theta,\phi)$ coordinate, where 

\begin{equation}\label{coord_horizon}
v:=t+r+2M\log \left(\frac{r}{2M}-1\right),\ R:=r.
\end{equation}
The metric takes the form
\begin{equation*}
g=-\left(1-\frac{2M}{R}\right) dv^2+2dvdR+R^2 (d\theta^2+\sin^2\theta d\phi^2),
\end{equation*}
which is smooth upto $R=2M, v\in\mathbb{R}$. We denote by $\mathcal{M}$ the Schwarzschild exterior including the future horizon $\mathcal{H}^+:=\{R=2M, v\in\mathbb{R}\}$ as

\begin{align*}
\mathcal{M}:=\bigg( \mathbb{R}_t\times (2M\infty)_r \sqcup \mathbb{R}_v\times [2M,\infty)_R/\sim  \bigg)\times \mathbb{S}^2.
\end{align*} 
Here we identify $(t,r)$ and $(v,R)$ coordinates through \eqref{coord_horizon}.\\

We record here the wave operator in various coordinate systems defined above. Let $\Box$ be the d'Alembertion operator with respect to the Schwarzschild metric $g$ as

\begin{align*}
\Box \psi:=g^{ab}\nabla_a\nabla_b \psi.
\end{align*}
In particular, in the $(t,r,\theta,\phi)$ coordinate we have

\begin{align*}
\Box=-\D^{-1}\frac{\partial^2}{\partial r^2}+\D\frac{\partial^2}{\partial r^2}+\frac{2}{r}\left(1-\frac{M}{r}\right)\frac{\partial}{\partial r}+\s{\Delta},
\end{align*}
where

\begin{align*}
\s{\Delta}=\frac{1}{r^2}\left( \frac{\partial^2}{\partial \theta^2}+\cot\theta\frac{\partial}{\partial\theta}+\sin^{-2}\theta\frac{\partial^2}{\partial\phi^2} \right)
\end{align*}
is the Laplacian operator for $\s{g}$. In the $\{v,R,\theta,\phi\}$ coordinate, we have

\begin{equation}\label{Box_horizon}
\Box =\D\frac{\partial^2}{\partial R^2}+2\frac{\partial^2}{\partial v\partial R}+\frac{2}{r}\left(1-\frac{M}{r}\right)\frac{\partial}{\partial R}+\frac{2}{r}\frac{\partial}{\partial v}+\s{\Delta}.
\end{equation}
More generally, let $g(r)$ be a function and define the new variables

\begin{align*}
\tau:=t+r+2M\log \left(\frac{r}{2M}-1\right)-g(r),\ \rho:=r.
\end{align*}
In this coordinate system, the the d'Alembertion operator reads

\begin{align}\label{Box decom}
\Box=\D\frac{\partial^2}{\partial \rho^2}+\frac{2}{r}\left(1-\frac{M}{r}\right)\frac{\partial}{\partial \rho}+Z\frac{\partial}{\partial \tau}+\s{\Delta},
\end{align}
where

\begin{equation*}
\begin{split}
Z=&\left( -2 \frac{dg}{dr}+\D \left(\frac{dg}{dr}\right)^2 \right)\frac{\partial}{\partial \tau}+\left( 2-2\D \frac{dg}{dr} \right)\frac{\partial}{\partial \rho}\\
&+\left(-\D \frac{d^2g}{dr^2}-\frac{2}{r}\left(1-\frac{M}{r}\right)\frac{dg}{dr}+\frac{2}{r}\right).
\end{split}
\end{equation*}

As for any mass parameter $M>0$ the Schwarzschild metric is vacuum, the infinitesimal mass change, denoted by $K$, solves \eqref{linear_gravity}. Explicitly,

\begin{align}\label{def:K}
K=\frac{1}{R}dv^2.
\end{align}
Similarly, for each $a$, the Kerr metric $g_{M,a}$ is vacuum and infinitesimal angular momentum change solve \eqref{linear_gravity}. Denote them as
\begin{align}\label{def:Km}
K_m=\frac{1}{R} \s{\epsilon}_A^{\ B}\nablas_B Y_{m1}\bigg((dv-dR) dx^A+dx^A (dv-dR)\bigg),\ m=-1,0,1.
\end{align}
Here $Y_{1m}$ are the spherical harmonic functions supported on $\ell=1$ and $\s{\epsilon}$ is the Levi-Civita tensor of $\s{g}_{AB}$. It can be verified directly that $K_m$ satisfies the harmonic gauge \eqref{HG}: $\Gamma[K_m]=0$. However, $K$ fails to be in the harmonic gauge and

\begin{align*}
\Gamma[K] =\frac{1}{R^2}dv.
\end{align*}
See subsection \ref{subsec:sad} for more discussion.\\

We start to define the vector fields. The isometry group of $g_M$ is 4-dimensional, including one time translation and 3-dimensional rotation Killing vector fields. We fix the notation for these Killing vector fields as

\begin{align*}
T:=&\frac{\partial}{\partial t},\ \Omega_1:=\frac{\partial}{\partial \phi},\\
\Omega_2:=&\cos\phi\frac{\partial}{\partial \theta}-\cot\theta\sin\phi\frac{\partial}{\partial \phi},\\
\Omega_3:=&\sin\phi\frac{\partial}{\partial \theta}+\cot\theta\cos\phi\frac{\partial}{\partial \phi}.
\end{align*}
Define the outgoing and incoming null vectors $L$ and $\underline{L}$ as

\begin{align*}
L:=\frac{\partial}{\partial t}+\D\frac{\partial}{\partial r},\ \underline{L}:=\frac{\partial}{\partial t}-\D\frac{\partial}{\partial r}.
\end{align*}
Note that in the $(v,R,\theta,\phi)$ coordinate, $\underline{L}=-\D\partial_R$ which vanishes along $\mathcal{H}^+$. To get non-zero incoming null vector, we define

\begin{align*}
\underline{L}':=&\D^{-1}\underline{L}.
\end{align*}
We are ready to define the red-shift vector $Y(\sigma)$ as in \cite{Dafermos-Rodnianski_redshift}. Fix  $r_{rs}^+:=5M/2$ and a cut-off function 

\begin{align}\label{cutoff_rs}
\eta_{rs}(r)=\left\{\begin{array}{cc}
1 & r\in [2M, 9M/4],\\
0 & r\in [5M/2,\infty).
\end{array}
\right.
\end{align}
For any $\sigma>0$, let

\begin{align}\label{vector_rs}
Y(\sigma):=\eta_{rs}(r)\cdot \left( \frac{\sigma}{2}(r-2M)T+(2+\sigma(r-2M))\underline{L}' \right).
\end{align}
Note that $Y(\sigma)$ is casual in $r\geq 2M$ and is supported in $r\in [2M,r_{rs}^+]$. We fix the notation for the following collections of vector fields:

\begin{equation}\label{def:collection}
\begin{split}
\mathcal{K}:=& \{MT,\Omega_1,\Omega_2,\Omega_3\},\\
\partial:=&\{T,\frac{1}{r}\Omega_1,\frac{1}{r}\Omega_2,\frac{1}{r}\Omega_3,\underline{L}'\},\\
\f:=&\{MT,\Omega_1,\Omega_2,\Omega_3,rL,M\underline{L}' \}.
\end{split}
\end{equation}
The collection $\K$ consists Killing vectors and we will use it, together with $\f$, to express higher order estimates of wave equations. See Proposition \ref{cor:source_high_order} in Appendix \ref{sec:wave equation}. 



\subsection{foliation of spacetime}

In this subsection we define the spacelike/null hypersurface along which we measure the solutions. Let $R_{\textup{null}}\geq 10 M$ to be determined later. It is the maximum among $R_Z$, $\ R_{\h}$ and $R_{\W}$ determined in subsections \ref{subsec:decomp_h}, \ref{subsec:psi_estimate} and \ref{subsec:rp} respectively and we keep the dependence of $R_{\textup{null}}$ in estimates until its value is fixed. Let $g(r)$ be a piecewise smooth function  defined in $[2M,\infty)$ which satisfies the following conditions: In $[2M,R_{\textup{null}}]$, $g(r)$ is smooth with

\begin{align*}
\frac{dg}{dr} \left(-2+\D \frac{dg}{dr}\right)<0.
\end{align*}
In $[5M/2,R_{\textup{null}}]$, 

\begin{align*}
g(r)=r+2M\log\left(\frac{r}{2M}-1\right).
\end{align*}
In $(R_{\textup{null}},\infty)$, 

\begin{align*}
g(r)=2\left(r+2M\log\left(\frac{r}{2M}-1\right)\right) -\left(R_{\textup{null}}+2M\log\left(\frac{R_{\textup{null}}}{2M}-1\right)\right).
\end{align*}
Let $\Sigma_\tau:=\{t+r+2M\log (r/2M-1)-g(r)=\tau\}$ for any $\tau\in\mathbb{R}$. The above requirements ensure that $\Sigma_\tau$ is spacelike in $r\in [2M,R_{\textup{null}}]$, is null in $r\in (R_{\textup{null}},\infty)$ and intersects with $\mathcal{H}^+$ transversally. Let $n$ be the unit future vector of $\Sigma_\tau$ in $[2M,R_{\textup{null}}]$ and fix $n=L$ in $(R_{\textup{null}},\infty)$.  For any $\tau_2\geq \tau_1$, let $D(\tau_1,\tau_2)$ be the region bounded by $\Sigma_{\tau_1}$ and $\Sigma_{\tau_2}.$ We introduce notations for $\Sigma_\tau$ and $D(\tau_1,\tau_2)$ in $r\leq R_{\textup{null}}$ or $r\geq R_{\textup{null}}$ respectively.

\begin{align*}
\Sigma_{\tau}'':=&\Sigma_{\tau}\cap\{2M \leq r\leq R_{\textup{null}}\},\\
\Sigma_{\tau}':=&\Sigma_{\tau}\cap\{R_{\textup{null}}\leq r\},\\
D''(\tau_1,\tau_2):=&D(\tau_1,\tau_2)\cap\{2M \leq r\leq R_{\textup{null}}\},\\
D'(\tau_1,\tau_2):=&D(\tau_1,\tau_2)\cap\{R_{\textup{null}}\leq r\}.
\end{align*}
For any $\tau\in\mathbb{R}$ and $r_0\geq 2M$, we denote the orbit sphere as

\begin{equation*}
\mathbb{S}^2(\tau,r_0):=\Sigma_\tau\cap \{r=r_0\}.
\end{equation*}
See figure \ref{fig:Penrose} for the Penrose diagram of $\mathcal{M}$. \\

Let the spacetime volume form and the volume form of unit sphere be denoted by

\begin{align*}
dvol&:=r^2\sin\theta dt\wedge dr\wedge d\theta\wedge d\phi,\\
dvol_{\mathbb{S}^2}&:=\sin\theta d\theta\wedge d\phi.
\end{align*}
Denote by $dvol_3$ the three form of $\Sigma_\tau$ corresponding to $n$. In other words,

\begin{equation*}
-n_{\flat} \wedge dvol_3=dvol,
\end{equation*}
where $(n_{\flat})_a=g_{ab}n^b$.  

\begin{figure}[h]\label{fig:Penrose}
\includegraphics[scale=0.15]{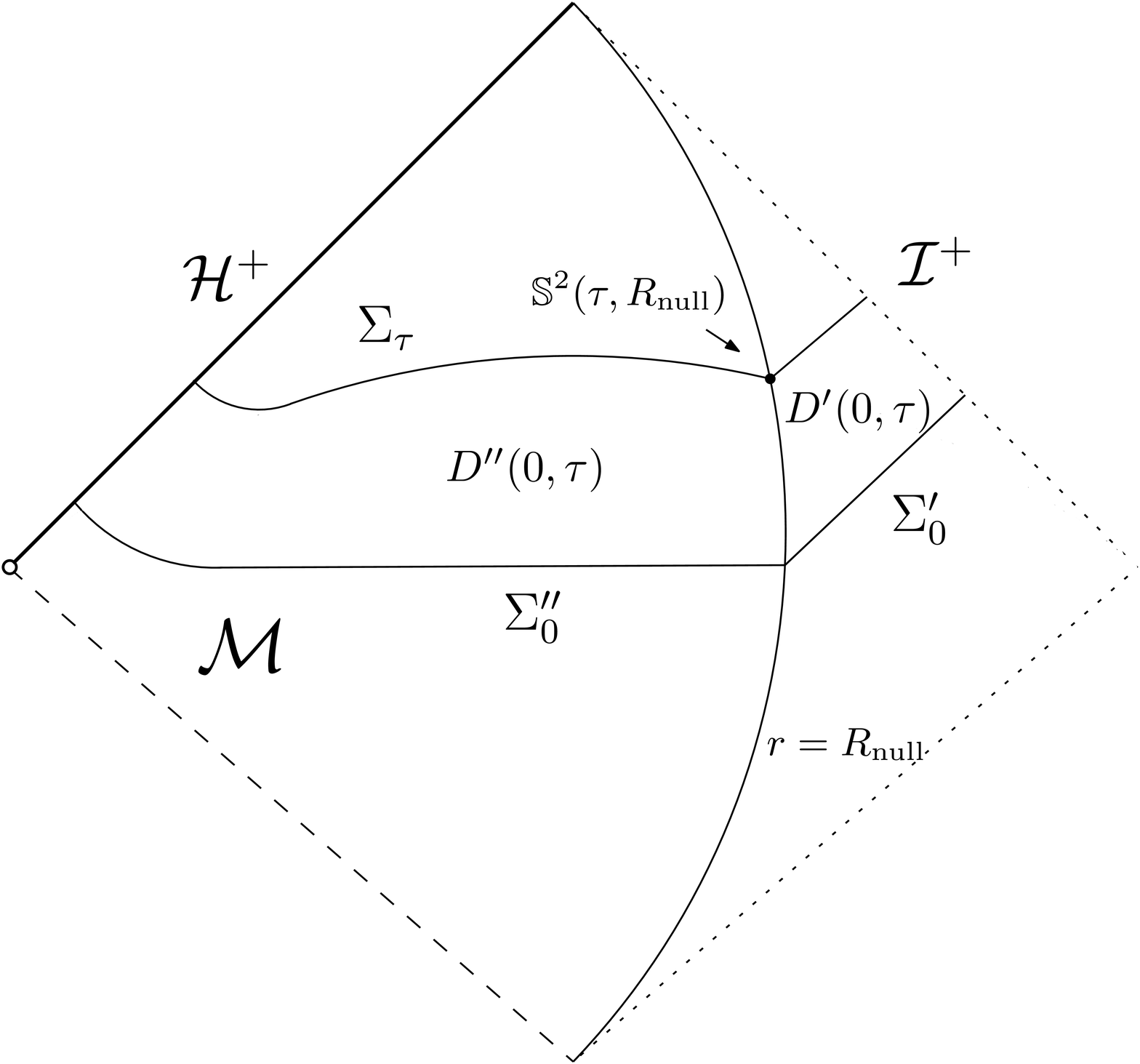}
\caption{Penrose diagram}
\end{figure}

\subsection{vector bundles $\mathcal{L}(-1)$ and $\mathcal{L}(-2)$}

Besides scalar functions, we also work with spherical one forms and spherical symmetric traceless two tensors. Let $\mathcal{L}(-1)\subset T^*\mathcal{M}$ be the subbundle of spherical one forms. Locally, a section $\xi$ of $\mathcal{L}(-1)$ can be written as

\begin{align*}
\xi=\xi_A dx^A.
\end{align*} 
Similarly, we denote by $\mathcal{L}(-2)\subset T^*\mathcal{M}\otimes_s T^*\mathcal{M}$ the subbundle of  spherical symmetric traceless two tensors. Locally, a section $\Xi$ of $\mathcal{L}(-2)$ can be written as

\begin{align*}
\Xi=\Xi_{AB} dx^A dx^B,\ \s{g}^{AB}\Xi_{AB}=0.
\end{align*}
The connections on $\mathcal{L}(-1)$ and $\mathcal{L}(-2)$ induced by the Levi-Civita connection are denoted by ${}^1\nabla$ and ${}^2\nabla$ respectively. Equipped with the induced metric, $\mathcal{L}(-1)$ and $\mathcal{L}(-2)$ are Riemannian vector bundles. We denote by ${}^1\Box$ and ${}^2\Box$ the d'Alembertion operator for ${}^1\nabla$ and ${}^2\nabla$. We omit the superscript in ${}^1 \nabla$ and ${}^2\nabla$ when it doesn't cause confusion.\\

Denote by $\mathcal{L}(0)$ the trivial line bundle $\mathbb{R}\times\mathcal{M}$ with sections being scalar functions. Define the operators $\done: \mathcal{L}(-1) \longrightarrow \mathcal{L}(0)$ and $ \dtwo:\mathcal{L}(-2) \longrightarrow \mathcal{L}(-1)$ as:

\begin{align*}
\done \xi:=&\nablas^A\xi_A,\\
(\dtwo \Xi)_A:=&\nablas^B\Xi_{AB}.
\end{align*}
Their adjoints are

\begin{align*}
(\donest \psi)_A:=&-\nablas_A \psi,\\
(\dtwost \xi)_{AB}:=&-\frac{1}{2}(\nablas_A\xi_B+\nablas_B\xi_A-\nablas^C\xi_C \slashed{g}_{AB}).
\end{align*}
We record here the simple commutation relation between $\done, \dtwo,\donest,\dtwost$ and ${}^1\nabla, {}^2\nabla$:

\begin{equation}\label{comm_r}
\begin{split}
\donest \nabla_r&= {}^1\nabla_r\donest+r^{-1}\donest,\ \done {}^1\nabla_r=\nabla_r\done +r^{-1}\done,\\
\dtwost {}^1\nabla_r&= {}^2\nabla_r\dtwost+r^{-1}\dtwost,\ \dtwo {}^2\nabla_r={}^1\nabla_r\dtwo +r^{-1}\done,
\end{split}
\end{equation}

\begin{equation}\label{comm_t}
\begin{split}
\donest \nabla_t&= {}^1\nabla_t\donest ,\ \done {}^1\nabla_t=\nabla_t\done,\\
\dtwost {}^1\nabla_t&= {}^2\nabla_t\dtwost ,\ \dtwo {}^2\nabla_t={}^1\nabla_t\dtwo .
\end{split}
\end{equation}


Let $W_a$ be a one form. we decompose $W_a$ as

\begin{align*}
W_adx^a=W_0dt+\D^{-1} W_1 dr+W_{2,A}dx^A.
\end{align*}   
Here $W_0$ and $W_1$ are scalar functions and $W_2$ is a section of $\mathcal{L}(-1)$. They are smooth up to the horizon provided $W_a$ is. We record here the components of the deformation tensor ${}^W\pi_{ab}=\nabla_a W_b+\nabla_b W_a$ in terms of $W_0$, $W_1$ and $W_2$.

\begin{equation}\label{pi}
\begin{split}
({}^W \pi)_{tt}=&2\nabla_t W_0-\frac{2M}{r^2}W_1,\\
({}^W \pi)_{rr}=&2\D^{-1}\nabla_r W_1-\frac{2M}{r^2}\D^{-2}W_1,\\
({}^W \pi)_{tr}=&\nabla_r W_0+\D^{-1}\nabla_t W_1-\frac{2M}{r^2}\D^{-1}W_0,\\
({}^W \pi)_{tA}=&-\donest W_0+{}^1\nabla_t W_2,\\
({}^W \pi)_{rA}=&-\D^{-1}\donest W_1+{}^1\nabla_r W_2-\frac{1}{r}W_2,\\
\s{tr} ({}^W \pi) =&2\done W_2+\frac{4}{r}W_1,\\
{}^W \hat{\pi}_{AB}=&-2\dtwost W_2.
\end{split}
\end{equation}
Here $\s{tr}{}^W\pi$ is the trace of ${}^W\pi_{AB}$ with respect to $\s{g}^{AB}$ and ${}^W\hat{\pi}_{AB}$ is the traceless part of ${}^W\pi_{AB}$.\\

Let $\Phi$ be a scalar function, a section of $\mathcal{L}(-1)$, $\mathcal{L}(-2)$ or more generally a section of products of these vector bundles, the stress-energy tensor of $\Phi$ is defined as

\begin{align*}
T_{ab}[\Phi]:=\nabla_a\Phi\cdot\nabla_b\Phi-\frac{1}{2}(\nabla^c\Phi\cdot\nabla_c\Phi)g_{ab}.
\end{align*}
Here the $\cdot$ stands for the contraction using the bundle metric. The stress-energy tensor satisfies the energy condition that $T_{ab}[\Phi]X^a Y^b\geq 0$ for any future causal vectors $X$ and $Y$. Moreover, for any vector field $X$ \textbf{orthogonal} to $\mathbb{S}^2(\tau,r)$, one has

\begin{align}\label{div}
\nabla^a (T_{ab}[\Phi]X^b)=\nabla_X\Phi\cdot \Box\Phi+T_{ab}[\Phi]\nabla^aX^b.
\end{align}
As $\Phi$ is a scalar function, \eqref{div} holds without the requirement $X$ being perpendicular to $\mathbb{S}^2(\tau,r)$. As $\Phi$ is a section of $\mathcal{L}(-1)$, $\mathcal{L}(-2)$ or their products, \eqref{div} follows from the fact that the curvature two forms of $\mathcal{L}(-1)$ and $\mathcal{L}(-2)$ are supported on the tangent plane of $\mathbb{S}^2(\tau,r)$, which can be checked through direct computation. An alternative way to see this is to consider $\mathcal{L}(-1)$ and $\mathcal{L}(-2)$ as pull back bundles from $\mathbb{S}^2$. We compute here the divergence of the red-shfit current. Let $Y(\sigma)$ be the red-shift vector defined in \eqref{vector_rs}. For any constant $c$, one has

\begin{equation}\label{g_redshift}
\begin{split}
&\nabla^a \left(T_{ab}[\Phi]Y^b(\sigma)-\frac{c}{r^2}|\Phi|^2 Y_a(\sigma) \right)\bigg|_{r=2M}\\
=&\frac{\sigma}{2}|\nabla_v\Phi|^2+\frac{1}{2M}|\nabla_R\Phi|^2+\frac{2}{M}\nabla_R\Phi\cdot\nabla_v\Phi+\frac{\sigma}{2}|\nablas\Phi|^2+\frac{c\sigma}{8M^2}|\Phi|^2\\
+&\nabla_Y\Phi\cdot \left(\Box\Phi-\frac{c}{r^2}\Phi\right)\bigg|_{r=2M}.
\end{split}
\end{equation}
Here we use the coordinate vector fields $\partial_v$ and $\partial_R$ in the coordinate \eqref{coord_horizon}. See \cite{Dafermos-Rodnianski_redshift} for the original computation.\\

For a one form $J_a$, the divergence theorem in $D(\tau_1,\tau_2)$ implies

\begin{equation}\label{div_thm}
\begin{split}
\int_{\Sigma_{\tau_1}} J\cdot n\ dvol_3=&\int_{\Sigma_{\tau_2}} J\cdot n\ dvol_3+\int_{\mathcal{H}^+(\tau_1,\tau_2) } J\cdot L\ dvol_{\mathcal{H}^+}\\
+&\int_{\mathcal{I}^{+}(\tau_1,\tau_2)} J\cdot \underline{L}\ dvol_{\mathcal{I}^+}+\int_{D(\tau_1,\tau_2)} \Div J\ dvol.
\end{split}
\end{equation}
Here $\mathcal{H}^+(\tau_1,\tau_2):=\mathcal{H}^+\cap\{\tau_1\leq \tau\leq \tau_2 \}$ is part of the boundary of $D(\tau_1,\tau_2)$ and $dvol_{\mathcal{H}^+}$ is the three form corresponding to the normal vector $L$. $\mathcal{I}^+(\tau_1,\tau_2)$ is the boundary of $D(\tau_1,\tau_2)$ at null infinity (See figure \ref{fig:Penrose}) and the integral is understood as

\begin{align*}
\int_{\mathcal{I}^{+}(\tau_1,\tau_2)} J\cdot \underline{L}\ dvol_{\mathcal{I}^+}:=\lim_{v_0\to\infty} \int_{\{\tau_1\leq \tau\leq \tau_2, v=v_0\}} J\cdot \underline{L}\ dvol_{\underline{L}},
\end{align*}
and $dvol_{\underline{L}}$ is the three form corresponding to $\underline{L}$. 
The vector field method in estimating wave equation is choosing a suitable multiplier $X$ and applying the divergence theorem to $J_a=T_{ab}[\Phi]X^b$ to control the behavior of $\Phi$ at later time $\Sigma_{\tau_2}$ by the previous data along $\Sigma_{\tau_1}$. See Appendix \ref{sec:wave equation} for a brief overview.\\ 

Recall that $\K,\f$ and $\partial$ are collection of vector fields defined in \eqref{def:collection}. For $\mathcal{A}=\f$ or $\mathcal{A}=\partial$ and $j\in\mathbb{N}_0$, we define

\begin{align*}
|\mathcal{A}^j\Phi|^2:=&\sum_{X_1,X_2,\cdots,X_j\in\mathcal{A}}|\nabla_{X_1}\nabla_{X_2}\cdots\nabla_{X_j}\Phi|^2,\\
|\mathcal{A}^{\leq j}\Phi|^2:=&\sum_{i=1}^j|\mathcal{A}^i\Phi|^2.
\end{align*}
For the Killing vector fields $\K$, we instead use Lie derivative and abuse the notation as

\begin{align*}
|\mathcal{K}^j\Phi|^2:=&\sum_{X_1,X_2,\cdots,X_j\in\mathcal{K}}|\mathcal{L}_{X_1}\mathcal{L}_{X_2}\cdots\mathcal{L}_{X_j}\Phi|^2,\\
|\mathcal{K}^{\leq j}\Phi|^2:=&\sum_{i=1}^j|\mathcal{K}^i\Phi|^2.
\end{align*}

\subsection{Energy norms}

Let $\Phi$ be a smooth function, section of $\mathcal{L}(-1)$, $\mathcal{L}(-2)$ or their products,  we define inductively for any $j\in\mathbb{N}$,

\begin{align*}
{\Phi}^{(0)}&:= s\cdot \Phi,\\
{\Phi}^{(j)}&:=\D^{-1} M s^2\cdot\nabla_L {\Phi}^{(j-1)}.
\end{align*}
Consider the energy norms:

\begin{align}
F[\Phi](\tau):=&\int_{\Sigma_\tau''} |\partial\Phi|^2+M^{-2}s^{-2}|\Phi|^2 dvol_3+\int_{\Sigma_\tau'} |\nabla_L \Phi|^2+|\nablas \Phi|^2+M^{-2}s^{-2}|\Phi|^2 dvol_3, \label{def:Fenergy}\\
B[\Phi](\tau):=&\int_{\Sigma_\tau} s^{-2}\left( \left(1-3s^{-1}\right)^2|\partial\Phi|^2+(1-2s^{-1})^2|\nabla_r\Phi|^2+M^{-2}s^{-2}|\Phi|^2 \right)  dvol_3, \label{def:Benergy}
\end{align}

\begin{equation}\label{def:FTenergy}
\begin{split}
F^T[\Phi](\tau):=&\int_{\Sigma_\tau''} (1-2s^{-1})|\nabla_{\underline{L}'}\Phi|^2+|\nabla_{L}\Phi|^2+|\nablas\Phi|^2+M^{-2}s^{-2}|\Phi|^2 dvol_3\\
+&\int_{\Sigma_\tau'} |\nabla_L \Phi|^2+|\nablas \Phi|^2+M^{-2}s^{-2}|\Phi|^2 dvol_3,
\end{split}
\end{equation}

\begin{align}\label{def:Bbarenergy}
\bar{B} [\Phi](\tau):=&\int_{\Sigma_\tau} s^{-2}\left( |\partial\Phi|^2+M^{-2}s^{-2}|\Phi|^2 \right)  dvol_3,
\end{align}

\begin{align}
E^{p,(j)}_L[\Phi](\tau):=&\int_{\Sigma_\tau'} s^{p-2}|\nabla_L {\Phi}^{(j)}|^2 dvol_3, \label{def:rpL_energy}\\
E^{p,(j)}_{\nablas}[\Phi](\tau):=&\int_{\Sigma_\tau'} s^{p-2}|\nablas {\Phi}^{(j)}|^2 dvol_3, \label{def:rps_energy}\\
E^{p,(j)}_{L,\nablas}[\Phi](\tau):=&E^{p,(j)}_L[\Phi](\tau)+E^{p,(j)}_{\nablas}[\Phi](\tau).
\end{align}
The $F[\Phi](\tau)$ energy is the $\dot{H}^1\cap\ rL^2$ norm of $\Phi$ along $\Sigma_\tau$ together with the $L^2$ norm of $\nabla_{n}\Phi$. Except the zeroth order term, $F^T[\Phi](\tau)$ energy is the boundary integral along $\Sigma_\tau$ when applying divergence theorem \eqref{div_thm} to the $T$-current $T_{ab}[\Phi]T^b$. Compared to $F[\Phi]$, $F^T[\Phi]$ degenerates at $r=2M$ for the term $|\nabla_{\underline{L}'}\Phi|^2$ because $T$ becomes a null vector. This degeneracy can be removed by considering the red-shift vector $Y(\sigma)$ defined in \eqref{vector_rs}. The $B[\Phi](\tau)$ energy appears in the bulk term in the Morawetz type estimates. It has degeneracy at the photon sphere $r=3M$ because of trapped geodesics. See subsection \ref{subsection:Morawetz}. The non-degenerate version is denoted by $\bar{B}[\Phi]$. The $E^{p,(0)}_{L,\nablas}$ energy appears in the the $r^p$-estimates. See subsection \ref{subsec:rp}.\\

We note that we only consider $E^{p,(j)}_L,\ E^{p,(j)}_{\nablas}$ and $E^{p,(j)}_{L,\nablas}$ for truncated $\Phi$ which is supported in $r\geq R_{\textup{null}}$. Throughout this paper, we take a cut-off function $\eta_{\textup{null}}(r)$ (depending only on $R_{\textup{null}}$) as follows. Let $\phi(s)$ be a cut-off function such that $\phi(s)=1$ for $s\geq 1$ and $\phi(s)=0$ for $s\leq 0$. Define 

\begin{equation}\label{cutoff_rp}
\eta_{\textup{null}}(r)=\phi((r-R_{\textup{null}})/M).
\end{equation}
To simply the notation, we use $\tilde{\Phi}$ to stand for $\eta_{\textup{null}}(r)\Phi$.

\begin{remark}
For a symmetric two tensor $h_{ab}$ or a one form $W_a$, we measure their size by components with respect to vector fields in $\partial$. In particular, the $F$-energy norm of $h_{ab}$ is defined as

\begin{align*}
F[h](\tau):=\sum_{X_1,X_2\in\partial} F[h(X_1,X_2)](\tau).
\end{align*}
Other energy norms for $h_{ab}$ or $W_a$ are defined similarly.
\end{remark}
We define below the norms for the source term. For any $p\in\mathbb{R}$ and $\tau_2\geq \tau_1$, let

\begin{equation*}
\begin{split}
\textup{E}^p_{\textup{source}}[G](\tau_1,\tau_2):=&M\int_{D(\tau_1,\tau_2)}s^{p+1}|G|^2dvol+ M^2\int_{\Sigma''_{\tau_1}}|G|^2dvol_3+M^2\int_{\Sigma''_{\tau_2}}|G|^2dvol_3\\
+&M^3\int_{D''(\tau_1,\tau_2)} \left(1-\frac{3M}{r}\right)^2|\partial G|^2+\D^2|\partial_r G|^2 dvol.
\end{split}
\end{equation*}
The decay of $\Phi$ which satisfies a wave equation with a source term $G$ will be formulate in terms of the following initial norms for $\Phi$ and spacetime bound on $G$. We fix a small constant $0<\delta<1/10$ throughout this paper. For any $k\in\mathbb{N}_0$, let

\begin{equation}\label{initial_norm}
I^{(k)}[\Phi]:=\left( F[\mathcal{K}^{\leq 2k+2}\Phi]+\sum_{\substack{i_1+i_2\leq 2k+1 \\ i_2\leq k}} E_L^{2-\delta,(i_2)}[\mathcal{K}^{i_1}\tilde{\Phi}]\right)(0).
\end{equation}
For any $\bar{\delta}\geq\delta>0$, let

\begin{equation}\label{source_assumption}
\textup{I}_{\textup{source},\bar{\delta}}[G]:=\sup_{p\in [\delta,2-\bar{\delta}]}\sup_{\tau_2\geq \tau_1\geq 0} \left(1+\frac{\tau_1}{M}\right)^{2-p-\bar{\delta}} \textup{E}^p_{\textup{source}}[\K^{\leq 2} G](\tau_1,\tau_2).
\end{equation}
We record in Appendix \ref{sec:wave equation} the decay estimates for wave equation obtained through the vector field method. 

\subsection{even/odd and spherical harmonic decomposition}\label{subsec:evenodd}

A spherical one form $\xi_A$ is said to be {even} if $\s{\epsilon}^{AB} \nablas_A\xi_B=0$ and is said to be odd if $\nablas^A\xi_A=0$, where $\s{\epsilon}$ is the Levi-Civita tensor of $\s{g}$. This is equivalent to the Hodge decomposition on $\mathbb{S}^2$. Similarly, a traceless symmetric two tensor $\Xi_{AB}$ is called even (odd) if $\nablas^A \Xi_{AB}$ is even (odd). All scalar functions are said to be even.\\

Recall that $x^\alpha, x^\beta$ is the coordinate of the quotient space $\mathcal{M}/\mathbb{S}^2$ and $x^A,x^B$ is the coordinate of $\mathbb{S}^2$. For a spacetime one from $W_a$ or a spacetime symmetric two tensor $h_{ab}$, we decompose them as $W_a dx^a=W_\alpha dx^\alpha+W_A dx^A$ and

\begin{align*}
h_{ab}dx^adx^b=h_{\alpha\beta}dx^\alpha dx^\beta+h_{\alpha A}(dx^\alpha dx^A+dx^Adx^\alpha )+\hat{h}_{AB}dx^Adx^B+\frac{1}{2}\s{tr}h\s{g}_{AB}dx^Adx^B.
\end{align*}
We use $\hat{h}_{AB}$ and $\s{tr}h$ to denote the traceless and the trace of $h_{AB}$ respectively. We consider $W_\alpha$, $h_{\alpha\beta}$ and $\s{tr}h$ as scalars, $W_A$ and $h_{\alpha A}$ as spherical one forms and $\hat{h}_{AB}$ as a spherical symmetric two tensor. Then $W_a$ or $h_{ab}$ is called even or odd if its components above are all even or odd. In the case that $W_a$ and $h_{ab}$ are axially symmetric, being even (odd) is the same as being symmetric (anti-symmetric) under the isometry $\phi\mapsto -\phi$. Therefore, the following lemma holds.

\begin{lemma}
Any symmetric two tensor $h_{ab}$ or one form $W_a$ can be decomposed into even and odd parts.

\begin{align*}
h_{ab}=&h_{ab}^{\even}+h_{ab}^{\textup{odd}},\\
W_a=&W_a^{\even}+W_a^{\textup{odd}}.
\end{align*}
Furthermore, $h_{ab}$ is a solution of \eqref{linear_gravity}, \eqref{HG} or \eqref{Main_equation} if and only if both $h_{ab}^{\even}$ and $h_{ab}^{\textup{odd}}$ are solutions separately. Similarly, $W_a$ is a solution of \eqref{wave_equation_vector} if and only if both $W_a^{\even}$ and $W_a^{\textup{odd}}$ are solutions.
\end{lemma}

Let $Y_{m\ell}(\theta,\phi),\ \ell\geq 0, |m|\leq \ell$ be the spherical harmonic functions on the unit sphere. For any smooth function $\psi(t,r,\theta,\phi)$, we can perform spherical harmonic decomposition for each $(t,r)$ and express $\psi$ as

\begin{align*}
\psi(t,r,\theta,\phi)=\sum_{\ell=0}^{\infty}\sum_{|m|\leq \ell} \psi_{m\ell}(t,r)\cdot Y_{m\ell}(\theta,\phi).
\end{align*}
Similarly, for any spherical one form $\xi_A$ , we decompose it as

\begin{align*}
\xi_A=&\sum_{\ell=1}^\infty \sum_{|m|\leq \ell} \xi^{\even}_{m\ell}(t,r)\cdot \nablas_A Y_{m\ell}+\sum_{\ell=1}^\infty \sum_{|m|\leq \ell}\xi^{\textup{odd}}_{m\ell}(t,r)\cdot\s{\epsilon}_A^{\ B} \nablas_B Y_{m\ell}.
\end{align*}
For any spherical symmetric traceless two tensor $\Xi_{AB}$,

\begin{align*}
\Xi_{AB}=&\sum_{\ell=2}^\infty \sum_{|m|\leq \ell} \Xi^{\even}_{m\ell}(t,r)\cdot\hat{\nablas}_{AB} Y_{m\ell}+\sum_{\ell=2}^\infty \sum_{|m|\leq \ell} \Xi^{\textup{odd}}_{m\ell}(t,r)\cdot\frac{1}{2}( \s{\epsilon}_A^{\ C} {\nablas}_{CB} Y_{m\ell}+\s{\epsilon}_B^{\ C} {\nablas}_{CA} Y_{m\ell}).
\end{align*}
In this paper, we focus only on the \textbf{even} tensors $\xi$ and $\Xi$ for which $\xi^{\textup{odd}}_{m\ell}\equiv 0$ and $\Xi^{\textup{odd}}_{m\ell}\equiv 0$. We often work with a fixed mode $\psi_{m\ell}(t,r)Y_{m\ell}(\theta,\phi)$. In this situation, $\s{\Delta}$ is equivalent to multiplying the number $r^{-2}\ell(\ell+1)$. We define $\lambda=\lambda(\ell)$ and $\Lambda=\Lambda(\ell)$ by

\begin{align}\label{def:lambda_ell}
\ell(\ell+1)=2\lambda+2=\Lambda.
\end{align}
Using this notation, we have a simple lemma below.

\begin{lemma}\label{lem:d_mode}

\begin{align*}
\donest \left(\sum_{\ell=0}^\infty\sum_{|m|\leq \ell} \psi_{m\ell}Y_{m\ell}\right)=\sum_{\ell=1}^\infty\sum_{|m|\leq \ell} -\psi_{m\ell} \nablas_A Y_{m\ell},\\
\dtwost \left(\sum_{\ell=1}^\infty \sum_{|m|\leq \ell} \xi_{m\ell} \nablas_A Y_{m\ell}\right)=\sum_{\ell=2}^\infty\sum_{|m|\leq \ell} -\xi_{m\ell} \hat{\nablas}_{AB} Y_{m\ell},
\end{align*}

and

\begin{align*}
\done \left(\sum_{\ell=1}^\infty\sum_{|m|\leq \ell} \xi_{m\ell} \nablas_A Y_{m\ell}\right)= \sum_{\ell=1}^\infty\sum_{|m|\leq \ell} -\Lambda r^{-2} \xi_{m\ell}  Y_{m\ell},\\
\dtwo \left(\sum_{\ell=2}^\infty\sum_{|m|\leq \ell}\Xi_{m\ell} \hat{\nablas}_{AB} Y_{m\ell}\right)=\sum_{\ell=2}^\infty\sum_{|m|\leq \ell} -\lambda r^{-2} \Xi_{m\ell} \nablas_{A} Y_{m\ell}.
\end{align*}

\end{lemma}
We denote the projection for a scalar function $\psi$ as

\begin{align*}
\psi_{\ell\geq \ell_0}:=&\sum_{\ell\geq \ell_0}\sum_{|m|\leq \ell} \psi_{m\ell}(t,r)Y_{m\ell}(\theta,\phi),\\
\psi_{\ell_0}:=&\sum_{|m|\leq \ell_0} \psi_{m\ell_0}(t,r)Y_{m\ell_0}(\theta,\phi).
\end{align*}
The projection of a section in $\mathcal{L}(-1)$ or $\mathcal{L}(-2)$ is defined similarly. We say $\psi$ is supported on $\ell\geq \ell_0$ if $\psi=\psi_{\ell\geq \ell_0}$.\\

We adapt the notation $A\leq_s B$ if for any $\tau\in\mathbb{R}$ and any $r\geq 2M$,

\begin{align*}
\int_{\mathbb{S}^2 (\tau,r)}A dvol_{\mathbb{S}^2}\leq \int_{\mathbb{S}^2(\tau,r)}B dvol_{\mathbb{S}^2}.
\end{align*}
For example, if $\psi=\psi_{\ell\geq \ell_0}$, we have $r^{-2}\ell_0(\ell_0+1) |\psi|^2\leq_s |\nablas \psi|^2$ by Poincar\'e inequality. Also We use the notation $A\lesssim B$ for two non-negative quantities $A$ and $B$ to indicate that there is a constant $C$ such that $A\leq CB$. The constant $C$ may depend on $\delta$, the small constant we fix throughout this paper, and the defining function of $\Sigma_0$ in $r\in [2M,3M]$.  We also use $A\approx B$ to denote the relation $A\lesssim B$ and $B\lesssim A$. If the constant $C$ depends on other quantities, we indicate the dependence using subscript. For instance, the notation $\lesssim_{R_{\textup{null}}}$ emphasizes that the constant depends on $R_{\textup{null}}$.

\subsection{main theorems for $\ell\geq 1$}

Now we are ready to state the main theorems for $\ell\geq 1$. The results for $\ell=0$ have more involved assumptions and we postpone them to section \ref{sec:ell=0}. The first one concerns the even solutions of \eqref{HG} and \eqref{Main_equation} supported on $\ell\geq 1$.

\begin{theorem}\label{thm:ell_geq_2}

Let $h_{ab}$ be an \textbf{even} solution of \eqref{HG} and \eqref{Main_equation}. \\

(1)\ Further assume that $h_{ab}$ is supported on $\ell\geq 2$ and that $\textup{Decay}_{\ell\geq 2}[h]$ defined below is finite. Then for any $p\in [\delta,2-3\delta]$ and $\tau\geq 0$, we have

\begin{align*}
F[h](\tau)+E^{p,(0)}_L[h](\tau)+\int_\tau^\infty \bar{B} [h](\tau')+E^{p-1,(0)}_{L,\nablas}[h](\tau') d\tau'\lesssim \left(1+\frac{\tau}{M}\right)^{-2+p+3\delta} \textup{Decay}_{\ell\geq 2}[h].
\end{align*}
Here

\begin{align*}
\textup{Decay}_{\ell\geq 2}[h]:=I^{(0)}[s^3 (M\partial)^{\leq 2} \K^{\leq 4}h, s^2(M\partial)^{\leq 3}\K^{\leq 1}h ,s\f^{\leq 3}(M\partial)^{\leq 1} \K^{\leq 6}h]+I^{(1)}[s(M\partial)^{\leq 1} \K^{\leq 8}h].
\end{align*}

(2)\ Assume that $h_{ab}$ is supported on $\ell= 1$ and that $\textup{Decay}_{\ell=1}[h]$ defined below is finite. Then for any $p\in [\delta,2-3\delta]$ and $\tau\geq 0$, we have

\begin{align*}
F[h](\tau)+E^{p,(0)}_L[h](\tau)+\int_\tau^\infty F[h](\tau')+E^{p-1,(0)}_{L,\nablas}[h](\tau') d\tau'\lesssim \left(1+\frac{\tau}{M}\right)^{-2+p+3\delta} \textup{Decay}_{\ell=1}[h].
\end{align*}
Here

\begin{align*}
\textup{Decay}_{\ell=1}[h]:=I^{(0)}[s^4(M\partial)^{\leq 2}\K^{\leq 4}h,\f^{\leq 1}\K^{\leq 6} h ]+I^{(1)}[\K^{\leq 6} h].
\end{align*}

\end{theorem}

The second theorem concerns the even solutions of $\eqref{wave_equation_vector}$.

\begin{theorem}\label{thm:ell_geq_1}

Let $W_{a}$ be an \textbf{even} solution of \eqref{wave_equation_vector}. Further assume that $W_a$ is supported on $\ell\geq 1$ and that $\textup{Decay}_{\ell\geq 1}[W]$ defined below is finite. Then for any $p\in [\delta,2-3\delta]$ and $\tau\geq 0$, we have

\begin{align*}
F[W](\tau)+E^{p,(0)}_L[W](\tau)+\int_\tau^\infty B[W](\tau')+E^{p-1,(0)}_{L,\nablas}[W](\tau') d\tau'\lesssim \left(1+\frac{\tau}{M}\right)^{-2+p+3\delta} \textup{Decay}_{\ell\geq 1}[W].
\end{align*}
Here

\begin{align*}
\textup{Decay}_{\ell\geq 1}[W]:=&I^{(0)}[s(M\partial)^{\leq 1}\K^{\leq 3}W,\f^{\leq 1}(M\partial)^{\leq 1}\K^{\leq 5}W]+I^{(1)}[(M\partial)^{\leq 1} \K^{\leq 5}W].
\end{align*}
\end{theorem}





\section{Basic Lemmas}\label{sec:lemma}

We prove in this section some useful lemmas that we need. Once we have an estimate for the solution $\psi$ of a wave equation, for instance Proposition \eqref{cor:free_price}, we can commute the equation with Killing vectors in $\K$ to control $\K\psi$. However $\K$ doesn't generate the full tangent space of $\mathcal{M}$. Fortunately, as long as $r>2M$, the missing direction can be recovered using $\Box\psi$, which is the content of the lemma below. 

\begin{lemma}
Let $\psi$ be a smooth function. Then for any $k\geq 1$ and $\bar{r}>2M$, we have for $r\geq \bar{r}$,

\begin{equation}\label{killing_to_partial}
\left|(M\partial)^{\leq k} u\right|^2\lesssim_{k,\bar{r}} \left(|u|^2+|M\partial\mathcal{K}^{\leq k-1} u|^2+M^{4}|(M\partial)^{\leq k-2}\Box u|^2 \right).
\end{equation}

\end{lemma}

\begin{proof}
As $k=1$, \eqref{killing_to_partial} holds trivially. Now we prove \eqref{killing_to_partial} by induction. Assume \eqref{killing_to_partial} holds for $k-1$. By the commutation relation in $\partial$,
\begin{align*}
|(M\partial)^{\leq k}\psi|^2\lesssim  |(M\partial)^{\leq k-1}\cdot  MT \psi|^2+\frac{M^2}{r^2}|(M\partial)^{\leq k-1}\cdot \Omega \psi|^2+|(M\underline{L}')^{k}\psi|^2+|(M\partial)^{\leq k-1}\psi|^2 .
\end{align*}
Applying the induction hypothesis to $MT u,$ we have
\begin{align*}
\left|(M\partial)^{\leq k-1}\cdot MT\psi\right|^2\lesssim&  |M\partial\mathcal{K}^{\leq k-2}\cdot M T\psi|^2+M^4|(M\partial)^{\leq k-3}\Box\cdot MT\psi|^2+|MT \psi|^2 \\
                                      \lesssim & |M\partial\mathcal{K}^{\leq k-1} \psi|^2+M^4|(M\partial)^{\leq k-2}\Box \psi|^2.
\end{align*}
Similarly,
\begin{align*}
\frac{M^2}{r^2}\left|(M\partial)^{\leq k-1}\cdot \Omega \psi\right|^2\lesssim &\frac{M^2}{r^2} \left( |M\partial\mathcal{K}^{\leq k-2}\cdot \Omega \psi|^2+M^4 |(M\partial)^{\leq k-3}\Box\cdot  \Omega \psi|^2+|\Omega \psi|^2 \right)\\
                                      \lesssim&  |M\partial\mathcal{K}^{\leq k-1} \psi|^2+M^4|(M\partial)^{\leq k-2}\Box \psi|^2.
\end{align*}
To deal with $(M\underline{L}')^{k}\psi$, we note that from \eqref{Box_horizon},

\begin{equation*}
(\underline{L}')^2=\D^{-1}\left( \Box +2T\underline{L}'+\frac{2}{r}\left(1-\frac{M}{r}\right)\underline{L}'-\frac{2}{r}T-\frac{1}{r^2}\sum_{i=1}^3\Omega_i^2 \right).
\end{equation*}
Therefore as $r\geq \bar{r}$,

\begin{align*}
|(M\underline{L}')^k \psi|^2\lesssim_{\bar{r}} &  M^4|(M\underline{L}')^{k-2}\Box \psi|^2+|(M\underline{L}')^{ k-1} MT \psi|^2+\frac{M^4}{r^4}|(M\partial)^{\leq k-2}\Omega^2 \psi|^2+|(M\partial)^{\leq k-1}\psi|^2.
\end{align*}
The first term appears on the right hand side of \eqref{killing_to_partial}, the second and the third terms are already estimated above and the last term can be controlled by induction hypothesis. Adding the above finishes the proof.
\end{proof}

In using the divergence theorem \eqref{div_thm}, it's often to only have $\dot{H}^1$ type norm without the zeroth order term and we rely on the Hardy inequality to recover such base term.

\begin{lemma}\label{lem:Hardy}
For any positive number $q>0$ and any Lipschitz function $\psi$, we have

\begin{equation}\label{Hardy}
\begin{split}
&2Mq^{-1}\sup_{r\geq R_{\textup{null}}} \int_{\mathbb{S}^2(\tau,r)} s^{-q}|\psi|^2 dvol_{\mathbb{S}^2}+M^{-2}\int_{\Sigma_\tau'}s^{-q-3}|\psi|^2dvol_3\\
\leq  &4Mq^{-1}\int_{\mathbb{S}^2(\tau,R_{\textup{null}})} s^{-q}|\psi|^2 dvol_{\mathbb{S}^2}+8q^{-2}\int_{\Sigma'_\tau} s^{-q-1}|\nabla_L \psi|^2 dvol_3.
\end{split}
\end{equation}
For any negative number $q<0$, we have

\begin{equation}\label{Hardy_reverse}
\begin{split}
&2M|q|^{-1}\int_{\mathbb{S}^2(\tau,R_{\textup{null}})} s^{-q}|\psi|^2 dvol_{\mathbb{S}^2}+M^{-2}\int_{\Sigma_\tau'}s^{-q-3}|\psi|^2dvol_3\\
\leq  &4M|q|^{-1} \liminf_{r\to\infty} \int_{\mathbb{S}^2(\tau,r)} s^{-q}|\psi|^2 dvol_{\mathbb{S}^2}+8|q|^{-2}\int_{\Sigma'_\tau} s^{-q-1}|\nabla_L \psi|^2 dvol_3.
\end{split}
\end{equation}

\end{lemma}

\begin{proof}
For any Lipchitz function $f(r)$, $a<b$ and $q\neq 0$ we have

\begin{align*}
\int_a^b r^{-q-1}f(r)^2 dr=q^{-1}r^{-q}f(r)^2\bigg|^{r=a}_{r=b}+q^{-1}\int_a^b 2r^{-q}f(r) f'(r)dr.
\end{align*}
As $q>0$, by Cauchy-Schwarz, 

\begin{align*}
2q^{-1}b^{-q}f(b)^2+\int_a^b r^{-q-1}f(r)^2 dr\leq 2q^{-1}a^{-q}f(a)^2+4q^{-2}\int_a^b r^{-q+1}(f'(r))^2 dr.
\end{align*}
Applying the above inequality with $f(r)=\psi(\tau,r,\theta,\phi)$, $a=R_{\textup{null}}$, $b\in (R_{\textup{null}}, \infty)$ and integrating along $\mathbb{S}^2$, we have

\begin{align*}
 &2q^{-1}\sup_{r\geq R_{\textup{null}}} \int_{\mathbb{S}^2(\tau,r)} r^{-q}|\psi|^2 dvol_{\mathbb{S}^2}+\int_{R_{\textup{null}}}^\infty \int_{\mathbb{S}^2(\tau,r)} r^{-q-3}|\psi|^2\cdot  r^2dvol_{\mathbb{S} ^2} dr  \\
                          \leq &2q^{-1}\int_{\mathbb{S}^2(\tau,R_{\textup{null}})} r^{-q}|\psi|^2 dvol_{\mathbb{S}^2} +4q^{-2}\int_{R_{\textup{null}}}^\infty \int_{\mathbb{S}^2(\tau,r)} r^{-q-1}|\nabla_L  \psi|^2\cdot \D^{-2} r^2dvol_{\mathbb{S}^2} dr.
\end{align*}
From $R_\textup{null}\geq 10M$,

\begin{align*}
r^2dvol_{\mathbb{S}^2}dr=&\D dvol_3\geq \frac{4}{5}dvol_3,\\
\D^{-2} r^2dvol_{\mathbb{S}^2}dr=&\D^{-1} dvol_3 \leq \frac{5}{4}  dvol_3.
\end{align*}
Then \eqref{Hardy} follows as $5/4,25/16<2$. Similar argument yields \eqref{Hardy_reverse}.

\end{proof}
This lemma has several useful corollaries.

\begin{corollary}\label{cor:E^p_base}

Let $\psi$ be a smooth function and $\tilde{\psi}=\eta_{\textup{null}}(r)\psi$ with the cut-off function $\eta_{\textup{null}}(r)$ defined in \eqref{cutoff_rp}. For any $p < 1$,

\begin{equation*}
\begin{split}
E^{p,(0)}[\tilde{\psi}](\tau)\geq &\frac{(1-p)}{4}\sup_{r\geq R_{\textup{null}}}\int_{\mathbb{S}^2(\tau,r)} M s^{p-1}|\tilde{\psi}^{(0)}|^2  dvol_{\mathbb{S}^2} +\frac{(1-p)^2}{8}\int_{\Sigma_\tau'}M^{-2}s^{p-4}|\tilde{\psi}^{(0)}|^2 dvol.
\end{split}
\end{equation*}

\end{corollary}

\begin{proof}

This follows by taking $q=1-p>0$ and $\psi=\tilde{\psi}^{(0)}$ together with $\tilde{\psi}\bigg|_{r=R_{\textup{null}}}=0$.

\end{proof}

\begin{corollary}\label{cor:F_base}

Let $\psi$ be a smooth function and $\tilde{\psi}=\eta_{\textup{null}}(r)\psi$ with the cut-off function $\eta_{\textup{null}}(r)$ defined in \eqref{cutoff_rp}. Suppose $E^{p,(0)}[\tilde{\psi}](\tau)<\infty$ for some $p>0$ and $\tau\in\mathbb{R}$, then

\begin{align*}
\int_{\Sigma'_\tau} |\nabla_L\psi|^2 dvol_3\geq \frac{1}{8}\int_{\Sigma'_\tau} M^{-2}s^{-2}|\psi|^2 dvol_3.
\end{align*}

\end{corollary}

\begin{proof}
By \eqref{Hardy_reverse} with $q=-1$, it's sufficient to check 

\begin{align*}
\liminf_{r\to\infty} \int_{\mathbb{S}^2(\tau,r)} s|\psi|^2 dovl_{\mathbb{S}^2}=0,
\end{align*}
which follows from the assumption $E^{p,(0)}[\tilde{\psi}](\tau)<\infty$ and Corollary \ref{cor:E^p_base}.
\end{proof}

\begin{corollary}

Let $\psi$ be a smooth function and $\tilde{\psi}=\eta_{\textup{null}}(r)\psi$ with the cut-off function $\eta_{\textup{null}}(r)$ defined in \eqref{cutoff_rp}. We have 

\begin{equation}\label{energy_by_rp}
\begin{split}
F[\psi](\tau)\lesssim_{R_{\textup{null}}} &B[\mathcal{K}^{\leq 1}\psi](\tau)+E^{0,(0)}_{L,\nablas}[\tilde{\psi}](\tau),\\
F[\psi](\tau)\lesssim_{R_{\textup{null}}} &\bar{B}[\psi](\tau)+E^{0,(0)}_{L,\nablas}[\tilde{\psi}](\tau),\\
\end{split}
\end{equation}

\end{corollary}

\begin{proof}

Clearly as long as $r$ is bounded, the integrand on the left hand side is controlled by the one on the right hand side. Also, from Corollary \ref{cor:E^p_base} with $p=0$, we can control the base term from $E^{0,(0)}[\tilde{\psi}](\tau)$.  The only term on the left which is not controlled for large $r$ is $|\nabla_L\psi|^2$, which can be controlled by $M^{-2}s^{-2}|\psi|^2$ and $s^{-2}|\nabla_L (s\psi)|^2$.

\end{proof}

\section{linearized Gravity under Harmonic Gauge}\label{sec:4}

Let $h_{ab}$ be an \textbf{even} tensor in the Schwarzschild spacetime satisfying \eqref{HG} and \eqref{Main_equation}. We start by decomposing $h$ into components and use \eqref{Main_equation} to obtain red-shift and $r^p$ estimate near horizon and null infinity respectively. The result is Proposition \ref{pro:metric_far}. Next, in sebsection \ref{subsec:to_RW} the mode $h_{\ell\geq 2}$ is considered. We perform a gauge transformation to the Regge-Wheeler gauge $h^{\RW}$, which is governed by the Zerilli quantity $\psi_Z$ satisfying the Zerilli equation \eqref{equ:psi}. The mode $h_{\ell}=1$ equals a deformation tensor $-{}^W\pi$ and the relationship between $h_{\ell=1}$ and $W$ is discussed in subsection \ref{subsec:ell=1}.

\subsection{decomposition of $h$}\label{subsec:decomp_h}

The main result of this subsection is:

\begin{proposition}\label{pro:metric_far}
Let $h_{ab}$ be an \textbf{even} solution of \eqref{Main_equation}. There exist $R_{\h}\geq 10M$ and $r_{rs,\h}\in (2M,r_{rs}^+)$ such that the following statement holds. Suppose $R_{\textup{null}}\geq R_{\h}$ then for any $p\in [\delta,2-\delta]$ and $\tau_2\geq \tau_1$,

\begin{equation}\label{ineq:metric_far}
\begin{split}
&F[h](\tau_2)+E^{p,(0)}_L[\tilde{h}](\tau_2)+M^{-1}\int_{\tau_1}^{\tau_2} \bar{B}[h](\tau)+E^{p-1,(0)}_{L,\nablas}[\tilde{h}](\tau) d\tau\\
\lesssim_{R_{\textup{null}}} &F[h](\tau_1)+E^{p,(0)}_L[\tilde{h}](\tau_1)+M^{-3} \int_{D(\tau_1,\tau_2)\cap \{r_{rs,\h}\leq r\leq R_{\textup{nul}}+M\}} |(M\partial)^{\leq 1}h|^2  dvol.
\end{split}
\end{equation}
Here $\tilde{h}=\eta_{\textup{null}}h$ and $\eta_{\textup{null}}(r)$ is the cut-off function defined in \eqref{cutoff_rp}.
\end{proposition}

This will be proved using standard $r^p$ and red-shift estimate of Dafermos-Rodnianski to a suitable decomposition of $h$. From simplicity, we only prove the case $h=h_{\ell\geq 2}$. The lower modes $\ell=1$  and $\ell=0$ can be proved similarly. See the remark at the end of this subsection.\\

We split the \textbf{even} tensor $h_{ab}$ into seven components as following. There are four scalar functions

\begin{align*}
\h_1=&\frac{1}{2}\left( h_{tt}+\D^2h_{rr}\right)+\frac{M}{r}\slashed{tr}h,\\
\h_2=&\frac{1}{2}\slashed{tr}h,\\
\h_3=&\frac{1}{2}\left(\D^{-1} h_{tt}-\D h_{rr}\right),\\
\h_6=&\D h_{tr}+\frac{M}{r}\slashed{tr}h,\\
\end{align*}
two spherical 1-forms

\begin{align*}
\h_4=&\D h_{rA}dx^A,\\
\h_7=&h_{tA}dx^A,
\end{align*}
and one spherical symmetric traceless 2-tensor

\begin{align*}
\h_5=\frac{1}{\sqrt{2}}\hat{h}_{AB}dx^Adx^B.
\end{align*}

\begin{proposition}

Suppose $h_{ab}$ satisfies \eqref{HG} and \eqref{Main_equation}, then $\h=( \h_1,\dots, \h_7)$ satisfies

\begin{equation}\label{equation_far}
\Box \h=A\h+B\h+U\h.
\end{equation}
Here $A$ is a seven by seven matrix:

\begin{align*}
A=&\frac{2}{r^2}\left[
\begin{array}{ccccccc}
\left(1-6s^{-1}\right) & -1+8s^{-1}-10s^{-2} & -1+10s^{-1}-20s^{-2} & 0 & 0 & 0 & 0\\
-1 & 1-2s^{-1} & 1-4s^{-1} & 0 & 0 & 0 & 0\\
-1 & 1-2s^{-1} & 1-4s^{-1} & 0 & 0 & 0 & 0\\
0  &  0  & 0 & 5/2-9s^{-1} & 0 & 0 & 0\\
0  &  0  & 0 & 0 & 1 & 0 & 0\\
-4s^{-1}  &  4s^{-1}-6s^{-2}  & 6s^{-1}-16s^{-2} & 0 & 0 & 1-2s^{-1} & 0\\
0  &  0  & 0 & -2s^{-1} & 0 & 0 & 1/2-3s^{-1}
\end{array}
\right].
\end{align*}
$B$ is a first order angular operator:

\begin{align*}
B=&\frac{2}{r}\left[
\begin{array}{ccccccc}
0  &  0  & 0 & (1-6s^{-1})\done & 0 & 0 & 0\\
0  &  0  & 0 & -\done & 0 & 0 & 0\\
0  &  0  & 0 & -\done & 0 & 0 & 0\\
\donest  &  -\donest  & (-1+3s^{-1})\donest & 0 & \sqrt{2}(1-3s^{-1})\dtwo & 0 & 0\\
0  &  0  & 0 & \sqrt{2}\dtwost & 0 & 0 & 0\\
0  &  0  & 0 & -4s^{-1}\done & 0 & 0 & (1-s)\done\\
0  &  -2s^{-1}\donest  & s^{-1}\donest & 0 & -\sqrt{2}s^{-1}\dtwo & \donest & 0
\end{array}
\right].
\end{align*}
And $U$ is a first order derivative along $L$:

\begin{align*}
U=&\frac{2M}{r^3}\D^{-1}\left[
\begin{array}{ccccccc}
0  &  0  & 0 & 0 & 0 & 0 & 0\\
0  &  0  & 0 & 0 & 0 & 0 & 0\\
0  &  0  & 0 & 0 & 0 & 0 & 0\\
0  &  0  & 0 & 0 & 0 & 0 & 0\\
0  &  0  & 0 & 0 & 0 & 0 & 0\\
-2 \nabla_L\cdot r  &  0  & 0 & 0 & 0 & 2 \nabla_L\cdot r & 0\\
0  &  0  & 0 & - \nabla_L\cdot r & 0 & 0 &  \nabla_L\cdot r
\end{array}
\right].
\end{align*}

\end{proposition}

The derivation of this equation can be found in Appendix \ref{sec:equation_metirc}. The leading terms in $A$ and $B$ are self-adjoint and we denote them by $A_{\textup{main}}$ and $B_{\textup{main}}$ respectively. 

\begin{align*}
A_{\textup{main}}=&\frac{2}{r^2}\left[
\begin{array}{ccccccc}
1 & -1 & -1 & 0 & 0 & 0 & 0\\
-1 & 1 & 1 & 0 & 0 & 0 & 0\\
-1 & 1 & 1 & 0 & 0 & 0 & 0\\
0  &  0  & 0 & 5/2 & 0 & 0 & 0\\
0  &  0  & 0 & 0 & 1 & 0 & 0\\
0 & 0 & 0 & 0 & 0 & 1 & 0\\
0 & 0 & 0 & 0 & 0 & 0 & 1/2
\end{array}
\right].
\end{align*}

\begin{align*}
B_{\textup{main}}=&\frac{2}{r}\left[
\begin{array}{ccccccc}
0 & 0 & 0 & \done & 0 & 0 & 0\\
0 & 0 & 0 & -\done & 0 & 0 & 0\\
0 & 0 & 0 & -\done & 0 & 0 & 0\\
\donest  &  -\donest  & -\donest & 0 & \sqrt{2}\dtwo & 0 & 0\\
0  &  0  & 0 & \sqrt{2}\dtwost & 0 & 0 & 0\\
0 & 0 & 0 & 0 & 0 & 0 & \done\\
0 & 0 & 0 & 0 & 0 & \done & 0
\end{array}
\right].
\end{align*}
Denote the remaining terms as $A_{\textup{sub}}=s(A-A_{\textup{main}})$ and $B_{\textup{sub}}=s(B-B_{\textup{main}})$. They decay faster in $r$ as

\begin{equation}\label{metric_low}
\begin{split}
|A_{\textup{sub}}\h|^2\lesssim &M^{-4}s^{-4} |\h|^2,\\
|B_{\textup{sub}}\h|^2\lesssim &M^{-2}s^{-2}|\nablas\h|^2,\\
|U\h|^2\lesssim &M^{-2}s^{-6}|\nabla_L(s\h)|^2.
\end{split}
\end{equation}

\begin{lemma}\label{lem:metric_coefficient_positivity}

Suppose $\h=\h_{\ell\geq 2}$, then

\begin{equation*}
|\nablas \h|^2+\h\cdot (A_{\textup{main}}+B_{\textup{main}})\cdot\h\geq_s 0.
\end{equation*}

Further assume $\h=\h_{\ell\geq 3}$, then we have

\begin{equation*}
|\nablas \h|^2+\h\cdot (A_{\textup{main}}+B_{\textup{main}})\cdot\h\gtrsim_s |\nablas \h|^2.
\end{equation*}

\end{lemma}

\begin{proof}

Through spherical harmonic decomposition, it's sufficient to prove the lemma for $\h$ supported on one spherical harmonic mode $Y_{m\ell}$. In this situation, angular operators become multiplication. From Lemma \ref{lem:d_mode}, in terms of the basis $Y_{m\ell}$, $r\Lambda^{-1/2}\nablas_A Y_{m\ell}$ and $r^2\lambda^{-1/2}\Lambda^{-1/2} \hat{\nablas}_{AB} Y_{m\ell}$, $B_{\textup{main}}$ is of the form

\begin{align*}
B_{\textup{main}}=&\frac{2}{r^2}\left[
\begin{array}{ccccccc}
0 & 0 & 0 & -\Lambda^{-1/2} & 0 & 0 & 0\\
0 & 0 & 0 & \Lambda^{-1/2} & 0 & 0 & 0\\
0 & 0 & 0 & \Lambda^{-1/2} & 0 & 0 & 0\\
-\Lambda^{-1/2}  &  \Lambda^{-1/2}  & \Lambda^{-1/2} & 0 & -(2\lambda)^{-1/2} & 0 & 0\\
0  &  0  & 0 & -(2\lambda)^{-1/2} & 0 & 0 & 0\\
0 & 0 & 0 & 0 & 0 & 0 & -\Lambda^{-1/2}\\
0 & 0 & 0 & 0 & 0 & -\Lambda^{-1/2} & 0
\end{array}
\right],
\end{align*}
Here $\Lambda$ and $\lambda$ are defined in \eqref{def:lambda_ell}. Similarly, $|\nablas\h|^2=_s\h\cdot \mathcal{O}\cdot\h$, where

\begin{align*}
\mathcal{O}=&\frac{1}{r^2}\left[
\begin{array}{ccccccc}
\Lambda & 0 & 0 & 0 & 0 & 0 & 0\\
0 & \Lambda & 0 & 0 & 0 & 0 & 0\\
0 & 0 & \Lambda & 0 & 0 & 0 & 0\\
0 & 0 & 0 & \Lambda-1 & 0 & 0 & 0\\
0 & 0 & 0 & 0 & \Lambda-4 & 0 & 0\\
0 & 0 & 0 & 0 & 0 & \Lambda & 0\\
0 & 0 & 0 & 0 & 0 & 0 & \Lambda-1
\end{array}
\right].
\end{align*}
The eigenvalues of $r^2(A_{\textup{main}}+B_{\textup{main}}+\mathcal{O})$ are

\begin{align*}
\{ \Lambda,\ (\ell-1)(\ell-2),\ (\ell+2)(\ell+3),\ \ell(\ell-1),\ (\ell+1)(\ell+2) \},
\end{align*}
where $\Lambda$ has multiplicity $3$. Hence $A_{\textup{main}}+B_{\textup{main}}+\mathcal{O}$ is non-negative definite for $\ell\geq 2$ and is comparable to $\mathcal{O}$ as $\ell\geq 3$.

\end{proof}

\begin{lemma}

Let $\h=\h_{\ell\geq 2}$ be a solution of \eqref{equation_far}. There exists $R_{\h}\geq 10M$ such that if $R_{\textup{null}}\geq R_{\h}$, for any $p\in[\delta,2-\delta]$ and $\tau_2\geq\tau_1\geq 0$ we have

\begin{equation}\label{metric_rp}
\begin{split}
&E^{p,(0)}_L[\tilde{\h}](\tau_2)+M^{-1}\int_{\tau_1}^{\tau_2}E^{p-1,(0)}_{L,\nablas}[\tilde{\h}](\tau)d\tau+M^{-1}\int_{D(\tau_1,\tau_2)} s^{-1-\delta}|\nabla_{\underline{L}}\tilde{\h}|^2\ dvol\\
\lesssim_{R_{\textup{null}}} &F[h](\tau_1)+E^{p,(0)}_L[\tilde{\h}](\tau_1)+M^{-3}\int_{D(\tau_1,\tau_2)\cap\{R_{\textup{null}}\leq r\leq R_{\textup{null}}+M\}} |(M\partial)^{\leq 1}\h|^2 dvol.
\end{split}
\end{equation}
Here $\tilde{\h}=\eta_{\textup{null}}\h$ and $\eta_{\textup{null}}(r)$ is the cut-off function defined in \eqref{cutoff_rp}.

\end{lemma}

\begin{proof}

Consider the current

\begin{align*}
J^p_{\h,a}:=s^{p-2}\D^{-1}\left( T_{ab}[s\tilde{\h}]L^b-\frac{1}{2} (s\tilde{\h})\cdot (A_{\textup{main}}+B_{\textup{main}}+2M/r^3)\cdot (s\tilde{\h})L_a \right).
\end{align*}

Clearly,

\begin{align*}
\int_{\Sigma_\tau'} J^p_{\h}\cdot L\  dvol_3\approx E^{p,(0)}_L[\tilde{\h}](\tau).
\end{align*}

We will apply the divergence theorem \ref{div_thm} to $J^p_{\h}$ in the region $D(\tau_1,\tau_2)$. To make sure the contribution along null infinity is non-negative , we compute

\begin{align*}
J^p_{\h}\cdot \underline{L}=s^{p-2} \left( |\nablas (s\tilde{H})|^2+(s\tilde{H})\cdot (A_{\textup{main}}+B_{\textup{main}}+2M/r^3)\cdot (s\tilde{H}) \right),
\end{align*}
which is non-negative (after integrated along $\mathbb{S}^2$) from Lemma \ref{lem:metric_coefficient_positivity}. We start to estimate $\Div J^p_{\h}$. By direct calculation,

\begin{equation}\label{div_H_rp}
\begin{split}
\Div J^p_{\h}=&M^{-1}s^{p-3}\left(\frac{p}{2}\D -\frac{M}{r}\right)\D^{-2} |\nabla_L(s\tilde{\h})|^2+M^{-1}s^{p-3}\left(1-\frac{p}{2}\right)|\nablas (s\tilde{\h})|^2\\
- &M^{-1}\frac{1}{2}s^{p-3} (s\tilde{\h})\cdot \left(r\left[\nabla_r, (A_{\textup{main}}+B_{\textup{main}}+2M/r^3)\right]+p(A_{\textup{main}}+B_{\textup{main}}+2M/r^3) \right)\cdot (s\tilde{\h})\\
+&s^{p-1}\D^{-1}\nabla_L(s\tilde{\h}) \cdot \Box \tilde{\h}\\
-&\frac{1}{2} s^{p-1}\D^{-1}\nabla_L(s\tilde{\h}) \cdot (A_{\textup{main}}+B_{\textup{main}}) \cdot \tilde{\h}\\
-&\frac{1}{2} s^{p-1}\D^{-1}  \tilde{\h} \cdot (A_{\textup{main}}+B_{\textup{main}})\cdot\nabla_L(s\tilde{\h}).
\end{split}
\end{equation}
From $[\nabla_r,r^2(A_{\textup{main}}+B_{\textup{main}})]=0$, we have

\begin{align*}
&r\left[\nabla_r, (A_{\textup{main}}+B_{\textup{main}}+2M/r^3)\right]+p(A_{\textup{main}}+B_{\textup{main}}+2M/r^3)\\
=&-(2-p)(A_{\textup{main}}+B_{\textup{main}})-\frac{(6-2p)M}{r^3}.
\end{align*}
Therefore by Lemma \ref{lem:metric_coefficient_positivity}, after integrated along $\mathbb{S}^2$,

\begin{align*}
M^{-1}s^{p-3}\left(1-\frac{p}{2}\right)|\nablas (s\tilde{\h})|^2-M^{-1}\frac{1}{2}s^{p-3}& (s\tilde{\h})\cdot \bigg(r\left[\nabla_r, (A_{\textup{main}}+B_{\textup{main}}+2M/r^3)\right]\\
+&p((A_{\textup{main}}+B_{\textup{main}}+2M/r^3)) \bigg)\cdot (s\tilde{\h})
\end{align*}
is non-negative and is bounded below by $M^{-1}s^{p-3}|\nablas (s\tilde{\h})|^2$ if $\h=\h_{\ell\geq 3}$. For $\h=\h_{\ell=2}$, we can still get $M^{-1}s^{p-3}|\nablas (s\tilde{\h})|^2$ through Corollary \ref{cor:E^p_base} as

\begin{align*}
M^{-1}E^{p-1,(0)}_L[\tilde{\h}](\tau)\geq &\frac{(2-p)^2}{8}M^{-3}\int_{\Sigma_\tau'} s^{p-4}|s\tilde{\h}|^2 svol\gtrsim  M^{-1}E^{p-1,(0)}_{\nablas}[\tilde{\h}](\tau).
\end{align*}
Since $A_{\textup{main}}$ and $B_{\textup{main}}$ are self-adjoint, the last three terms in \eqref{div_H_rp} can be combined as

\begin{align*}
&s^{p-1}\D^{-1}\nabla_L(s\tilde{\h})\cdot (\Box-A_{\textup{main}}-B_{\textup{main}})\cdot\tilde{\h}\\
=&s^{p-1}\D^{-1}\nabla_L(s\tilde{\h})\cdot \left( s^{-1}(A_{\textup{sub}}+B_{\textup{sub}})\tilde{\h}+\eta_{\textup{null}}\cdot U\h+{\cpur \Box\eta_{\textup{null}} \cdot \h+2\nabla\eta_{\textup{null}}\cdot\nabla\h} \right).
\end{align*}
The $A_{\textup{sub}}$ and $B_{\textup{sub}}$ terms decay faster in $r$ and can be estimated by \eqref{metric_low} and Cauchy-Schwarz:

\begin{align*}
&s^{p-1}\D^{-1}\nabla_L(s\tilde{\h})\cdot s^{-1}(A_{\textup{sub}}+B_{\textup{sub}})\tilde{\h}\\
=&s^{p-3}\D^{-1}\nabla_L(s\tilde{\h})\cdot (A_{\textup{sub}}+B_{\textup{sub}})(s\tilde{\h})\\
\lesssim &\epsilon\cdot M^{-1}s^{p-3}|\nabla_L (s\tilde{\h})|^2+{\cblue \frac{1}{\epsilon}\cdot M^{-3}s^{p-7}|s\tilde{\h}|^2+\frac{1}{\epsilon}\cdot M^{-1}s^{p-5}|\nablas (s\tilde{\h})|^2},
\end{align*}
where $\epsilon$ is any positive real. Similarly,

\begin{align*}
&s^{p-1}\D^{-1}\nabla_L(s\tilde{\h})\cdot \eta_{\textup{null}}(r) U\h\\
\lesssim & \epsilon \cdot M^{-1}s^{p-3}|\nabla_L(s\tilde{\h})|^2+{\cblue \frac{1}{\epsilon} \cdot M^{-1}s^{p-5}|\nabla_L(s\tilde{\h})|^2}+{\cpur \frac{1}{\epsilon} \cdot M^{-1}s^{p-5}|\nabla_L\eta_{\textup{null}}|^2|s\h|^2}.
\end{align*}
We take $\epsilon>0$ small enough to absorb the term $\epsilon \cdot M^{-1} s^{p-3}|\nabla_L(s\tilde{\h})|^2$, and then pick $R_{\h}$ large enough to absorb the terms $\epsilon^{-1} \cdot M^{-1}s^{p-5}|\nabla_L(s\tilde{\h})|^2, \epsilon^{-1} \cdot M^{-1}s^{p-5}|\nablas (s\tilde{\h})|^2$ and $ \epsilon^{-1} \cdot M^{-3}s^{p-7}| s\tilde{\h}|^2$ above. Therefore we have for some constant $C>0$,

\begin{align*}
\int_{D(\tau_1,\tau_2)} \Div J^p_{\h}\ dvol+&{C\cpur M \int_{D(\tau_1,\tau_2)} s^{p+1}\left(|\Box\eta_{\textup{null}} \cdot \h|^2+|\nabla\eta_{\textup{null}} \cdot\nabla\h|^2+M^{-2}s^{-4}|\nabla_L\eta_{\textup{null}} |^2 |\h|^2 \right) dvol}\\
\gtrsim &M^{-1}\int_{D(\tau_1,\tau_2)} s^{p-3}|\nabla_L(s\tilde{\h})|^2+s^{p-3}|\nablas(s\tilde{\h})|^2  dvol.
\end{align*}
Applying the divergence theorem \eqref{div_thm} yields 

\begin{equation*}
\begin{split}
&E^{p,(0)}_L[\tilde{\h}](\tau_2)+M^{-1}\int_{\tau_1}^{\tau_2}E^{p-1,(0)}_{L,\nablas}[\tilde{\h}](\tau)d\tau\\
\lesssim_{R_{\textup{null}}} & E^{p,(0)}_L[\tilde{\h}](\tau_1)+M^{-3}\int_{D(\tau_1,\tau_2)\cap\{R_{\textup{null}}\leq r\leq R_{\textup{null}}+M\}} |(M\partial)^{\leq 1}\h|^2 dvol.
\end{split}
\end{equation*}

To add $s^{-1-\delta}|\nabla_{\underline{L}}\tilde{\h}|^2$ term, we consider

\begin{align}
(J_{\h,\delta})_a:=s^{-\delta}\left(T_{ab}[\tilde{\h}]-\frac{1}{2}\tilde{\h}\cdot (A_{\textup{main}}+B_{\textup{main}})\cdot \tilde{\h} g_{ab}  \right)T^b.
\end{align}
From Lemma \ref{lem:metric_coefficient_positivity}, $J_{\h,\delta}\cdot n_\tau$ and $J_{\h,\delta}\cdot \underline{L}$ are non-negative after integrated along $\mathbb{S}^2$. Moreover, 

\begin{align*}
0\leq \int_{\Sigma_\tau} J_{\h,\delta}\cdot n dvol_3\lesssim F[\h](\tau).
\end{align*}
The divergence of $J_{\h,\delta}$ is

\begin{align*}
\Div J_{\h,\delta}=_s&\frac{\delta}{4}M^{-1}s^{-1-\delta}|\nabla_{\underline{L}}\tilde{\h}|^2-{\cpur \frac{\delta}{4}M^{-1}s^{-1-\delta}|\nabla_{L}\tilde{\h}|^2}\\
+&{\cpur \eta_{\textup{null}}(r)  s^{-\delta}\nabla_t\h\cdot (\Box\eta_{\textup{null}}(r)\cdot\h+2\nabla\eta_{\textup{null}}(r)\cdot\nabla\h)}\\
+&\eta_{\textup{null}}(r) s^{-\delta} \nabla_t\tilde{\h} \cdot ( s^{-1}A_{\textup{sub}}+s^{-1}B_{\textup{sub}}+U)\h.
\end{align*}
The second line will be grouped into $M^{-3}|(M\partial)^{\leq 1}\h|^2\cdot\chi_{[R_{\textup{null}},R_{\textup{null}}+M]}$. The last line, by Cauchy-Schwarz, can be bounded as

\begin{align*}
&\big|\eta_{\textup{null}}(r) s^{-\delta} \nabla_t\tilde{\h} \cdot ( s^{-1}A_{\textup{sub}}+s^{-1}B_{\textup{sub}}+U)\h\big|\\\lesssim_{R_{\textup{null}}} &\epsilon\cdot M^{-1}s^{-1-\delta}(|\nabla_L\tilde{\h}|^2+|\nabla_{\underline{L}}\tilde{\h}|^2)\\
+&\frac{1}{\epsilon}\cdot M^{-1}s^{-5-\delta}(|\nabla_L (s\tilde{\h})|^2+|\nablas (s\tilde{\h})|^2+M^{-2}s^{-2}|s\tilde{\h}|^2)\\
+&\frac{1}{\epsilon}\cdot M^{-3}|(M\partial)^{\leq 1}\h|^2\cdot\chi_{[R_{\textup{null}},R_{\textup{null}}+M]}.
\end{align*}
By taking $\epsilon>0$ small, we have for some $C>0$,

\begin{align*}
\Div J_{\h,\delta}\geq &\frac{1}{C}\cdot M^{-1}s^{-1-\delta} |\nabla_{\underline{L}}\tilde{\h}|^2\\
-&C\cdot M^{-1}s^{-3-\delta} (|\nabla_L (s\tilde{\h})|^2+|\nablas (s\tilde{\h})|^2+M^{-2}s^{-2}|s\tilde{\h}|^2)\\
-&C\cdot M^{-3}|(M\partial)^{\leq 1}\h|^2\cdot\chi_{[R_{\textup{null}},R_{\textup{null}}+M]}.
\end{align*}
Because for any $p\in [\delta,2-\delta]$, $p-3\geq -3+\delta>-3-\delta$, the second line can be absorbed into the integrand of $E^{p-1,(0)}_{L,\nablas}[\tilde{\h}]$ through further enlarging the value of $R_{\h}$. Then the result follows by applying divergence theorem \eqref{div_thm} to $J^p_{\h}+\epsilon J_{\h,\delta}$ with $\epsilon>0$ small enough.

\end{proof}

We now turn to the red-shift estimate near the horizon $\mathcal{H}^+$. To apply the red-shift estimate, the only thing we need to check is that the coefficient of $\nabla_{\underline{L}'} \h$ in $\Box\h$ is non-negative (actually, negative part bounded by the surface gravity of Schwarzschild is enough.) The above decomposition is actually not regular on the horizon as they don't generate the full vector space $T^*\mathcal{M}\otimes_s T^*\mathcal{M}$. We instead consider $\h_6' :=\D^{-2}(\h_6-\h_1)$ and $\h_7' :=\D^{-1}(\h_7-\h_4)$. They satisfy the equation

\begin{align*}
\Box \h_6'= &\frac{4M}{r^2}\nabla_{\underline{L}'} \h_6'+\frac{2}{r}\done \h_7'+\frac{2}{r^2}\h_2+\frac{2}{r^2}\h_3+\frac{2}{r^2}\h_6',\\
\Box \h_7'= &\frac{2M}{r^2}\nabla_{\underline{L}'} \h_7'-\frac{2}{r}\donest\h_2+\frac{2}{r}\donest\h_3+\frac{2\sqrt{2}}{r}\dtwo\h_5+\frac{2}{r}\D\donest\h_6'\\
-&\frac{4}{r^2}\h_4+\frac{1}{r^2}\left(1-\frac{4M}{r}\right)\h_7'.
\end{align*}
Let $\h'=(\h_1,\h_2,\h_3,\h_4,\h_5,\h_6',\h_7')$ and consider the red-shift current

\begin{align}
J_{\h',a}:=T_{ab}[\h']\cdot (Y(\sigma)+\eta_{rs} T)^b-\frac{1}{2r^2}|\h '|^2\cdot(Y(\sigma)+\eta_{rs} T)_a,
\end{align}
where $\sigma$ is a positive number to be determined, $\eta_{rs}$ is the red-shift cut-off function and $Y(\sigma)$ is the red-shift vector defined in \eqref{cutoff_rs} and \eqref{vector_rs}. From \eqref{g_redshift}, on the horizon $r=2M$ we have

\begin{align*}
\Div J_{\h'}\bigg|_{r=2M}=&\frac{\sigma}{2}|\nabla_v\h'|^2+\frac{1}{2M}|\nabla_R\h'|^2+\frac{2}{M}\nabla_R\h'\cdot\nabla_v\h'+\frac{\sigma}{2}|\slashed{\nabla}\h'|^2+\frac{\sigma}{8M^2}|\h'|^2\\
+&\nabla_{T+Y}\h'\cdot (\Box\h'-\frac{1}{r^2}\h').
\end{align*}
Since the coefficient of $\nabla_{\underline{L}'}\h$ in $\Box\h'$ is non-negative and $Y(\sigma)=2\underline{L}'$ on the horizon, we have from Cauchy-Schwarz for some constant $C>0$,

\begin{align*}
\nabla_{Y+T}\h'\cdot(\Box\h'-\frac{1}{r^2}\h')+ \frac{1}{4M}|\nabla_R\h'|^2 \bigg|_{r=2M} \geq -C\bigg( M^{-1}|\nabla_v\h'|^2+M^{-1}|\nablas\h'|^2+M^{-3}|\h'|^2\bigg),
\end{align*}
By choosing $\sigma>>M^{-1}$, we obtain

\begin{align*}
\Div J_{\h'}\bigg|_{r=2M}\gtrsim M^{-3}|(M\partial)^{\leq 1}\h'|^2.
\end{align*}
Because of continuity, there exists $r_{rs,\h}\in (2M,r_{rs}^+)$ such that the above estimate holds in $[2M,r_{rs,\h}]$. Therefore through the divergence theorem \eqref{div_thm} of $J_{\h'}$, we obtain

\begin{equation}\label{metric_rs}
\begin{split}
&M^{-2}\int_{\Sigma_{\tau_2}} |(M\partial)^{\leq 1}\h'|^2\cdot\chi_{[2M, r_{rs,\h}]} dvol_3+M^{-3}\int_{D(\tau_1,\tau_2)} |(M\partial)^{\leq 1}\h'|^2\cdot\chi_{[2M, r_{rs,\h}]} dvol\\
\lesssim &M^{-2}\int_{\Sigma_{\tau_1}} |(M\partial)^{\leq 1}\h'|^2\cdot\chi_{[2M, r_{rs}^+]} dvol_3+M^{-3}\int_{D(\tau_1,\tau_2)}|(M\partial)^{\leq 1}\h'|^2\cdot\chi_{[ r_{rs,\h}, r_{rs}^+]} dvol.
\end{split}
\end{equation}


\begin{proof}[proof of Proposition \ref{pro:metric_far}]

From the view of \eqref{metric_rp}, we have controlled $E^{p,(0)}_L[\tilde{h}](\tau_2)$ and $M^{-1}\int_{\tau_1}^{\tau_2} E^{p,(0)}_L[\tilde{h}](\tau) d\tau$. From \eqref{metric_rs}, the part of $F[h](\tau_2)$ and the integral of $M^{-1}\int_{\tau_1}^{\tau_2}\bar{B}[h](\tau)d\tau$ near the horizon $r\in [2M, r_{rs,\h}]$ is also bounded by the right hand side of \eqref{ineq:metric_far}. Furthermore, from Corollary \ref{cor:E^p_base}, 
\begin{align*}
&\int_{\tau_1}^{\tau_2}E^{p-1,(0)}_{L,\nablas}[\tilde{\h}](\tau)d\tau+\int_{D(\tau_1,\tau_2)} s^{-1-\delta}|\nabla_{\underline{L}}\tilde{\h}|^2 dvol\\
\approx &\int_{D(\tau_1,\tau_2)} s^{p-1}(|\nabla_L\tilde{\h}|^2+|\nablas\tilde{\h}|^2+M^{-2}s^{-2}|\tilde{\h}|^2)+s^{-1-\delta}|\nabla_{\underline{L}}\tilde{\h}|^2 dvol,
\end{align*}
which bounds the integrand of $\bar{B}[h](\tau)$ in $r\geq R_{\textup{null}}+M$. As the integrand of $\bar{B}[h](\tau)$ in $r\in [r_{rs,\h},R_{\textup{null}}+M]$ can be bounded by the right hand side of \eqref{ineq:metric_far}, the only term on the left hand side of \eqref{ineq:metric_far} which is not controlled yet is
\begin{align*}
&\int_{\Sigma_{\tau_2}''\cap\{r\geq r_{rs,\h} \}} |\partial h|^2+M^{-2}s^{-2}|h|^2 dvol_3+\int_{\Sigma_{\tau_2}'} |\nabla_L h|^2+|\nablas h|^2+M^{-2}s^{-2}|h|^2 dvol_3.
\end{align*}
This can be bounded by $F^T[\h]$ defined in \eqref{def:FTenergy}. We can add this term by considering the current

\begin{align}
(J_{\h,T})_a:=\left( T_{ab}[\h]-\frac{1}{2r^2} |\h|^2 g_{ab}\right) T^b.
\end{align}
Clearly, 

\begin{align*}
\int_{\Sigma_{\tau_2}} J_{\h,T}\cdot n dvol_3\approx F^T[\h].  
\end{align*}
We estimate its divergence as

\begin{align*}
\left|\Div J_{\h,T}\right|\lesssim& |\nabla_t\h|\left( \frac{1}{M^2s}|M\partial\h|+\frac{1}{M^2s^2}|\h|\right)\\
\lesssim &M^{-3}|(M\partial)^{\leq 1}\h|^2\cdot \chi_{[2M,R_{\textup{null}}+M]}+\bigg(M^{-1} s^{-1-\delta}|\nabla_t \h|^2\\
+&M^{-1} s^{-3+\delta} (|\nabla_L (s\h)|^2+|\nablas (s\h)|^2+M^{-2}s^{-2} |s\h|^2)\bigg)\cdot \chi_{[R_{\textup{null}}+M,\infty)}.
\end{align*}
The second term already appears in the the bulk term of \eqref{metric_rp}. Then we apply the divergence theorem \eqref{div_thm} to $\epsilon J_{\h,T}$ with $\epsilon>0$ small enough such that $\left|\Div J_{\h,T}\right|$ can be absorbed into \eqref{metric_rs} and \eqref{metric_rp} in $r\in [2M,r_{rs,\h}]$ and in $ r\in[R_{\textup{null}}+M,\infty)$ respectively. Adding this estimate together with \eqref{metric_rs} and \eqref{metric_rp} yields the result.

\end{proof}

\begin{remark}

The equation for $\h_{\ell=1}$ is the same as the $\ell\geq 2$ case except that there is no $\h_5$. In particular, the leading term of $A_{\ell=1}$ and $B_{\ell=1}$ are

\begin{align*}
\frac{2}{r^2}\left[
\begin{array}{cccccc}
1 & -1 & -1 & 0 &  0 & 0\\
-1 & 1 & 1 & 0 &  0 & 0\\
-1 & 1 & 1 & 0 &  0 & 0\\
0  &  0  & 0 & 5/2 &  0 & 0\\
0 & 0 & 0 & 0 &  1 & 0\\
0 & 0 & 0 & 0 &  0 & 1/2
\end{array}
\right].
\end{align*}

\begin{align*}
&\frac{2}{r}\left[
\begin{array}{ccccccc}
0 & 0 & 0 & \done &  0 & 0\\
0 & 0 & 0 & -\done &  0 & 0\\
0 & 0 & 0 & -\done &  0 & 0\\
\donest  &  -\donest  & -\donest & 0  & 0 & 0\\
0 & 0 & 0 & 0 & 0 & \done\\
0 & 0 & 0 & 0 & \done & 0
\end{array}
\right].
\end{align*}
We verify that in this case the eigenvalue of $r^2(A_{\ell=1}+B_{\ell=1}+\mathcal{O}_{\ell=1})$ is $\{0,2,2,2,6,12\}$. In particular, $A_{\ell=1}+B_{\ell=1}+\mathcal{O}_{\ell=1}$ is semi-positive definite. Therefore, \eqref{metric_rp} and \eqref{metric_rs} and then Proposition \ref{pro:metric_far} can be proved by the same argument.\\
Similarly, as $\ell=0$ there's no $\h_4$ and $\h_7$. $B_{\ell=0}=0$ and the leading term of $A_{\ell=0}$ is

\begin{align*}
&\frac{2}{r^2}\left[
\begin{array}{cccccc}
1 & -1 & -1  &  0 \\
-1 & 1 & 1  &  0 \\
-1 & 1 & 1  &  0 \\
0 & 0 & 0  &  1 
\end{array}
\right],
\end{align*}
which is still semi-positive definite with eigenvalues $\{0,0,1,6\}$.

\end{remark}

\subsection{relation to Regge-Wheeler gauge}\label{subsec:to_RW}

Throughout this subsection we assume $h=h_{\ell\geq 2}$. An \textbf{even} solution of \eqref{linear_gravity}, denoted by $h^{\RW}$, is said to be in the Regge-Wheeler gauge if it is supported on $\ell\geq 2$ and

\begin{equation}
\hat{h}^{\RW}_{AB}=0,\ h^{\RW}_{tA}=h^{\RW}_{rA}=0.
\end{equation}
For any \textbf{even} solution of \eqref{linear_gravity} supported on $\ell\geq 2$, denoted by $h_{ab}$, there exists a unique \textbf{even} vector field $W$, which is also supported on $\ell\geq 2$, such that 

\begin{align*}
h+{}^W\pi=h^{\RW}.
\end{align*}
We decompose $W$ into components as

\begin{align*}
W=W_0dt+W_1 \D^{-1} dr+W_{2,A}dx^A,
\end{align*}
where $W_0$ and $W_1$ are scalars and $W_2$ is an \textbf{even} spherical one form. As $\donest$ and $\dtwost$ are invertible for $\ell\geq 2$, They are uniquely determined through \eqref{comm_r}, \eqref{comm_t} and \eqref{pi}  as

\begin{align}
\dtwost\donest W_0=&\dtwost h_{tA}+\frac{1}{2} {}^2\nabla_t \hat{h}_{AB}, \label{def:W0} \\
\dtwost\donest W_1=&\D\dtwost h_{rA}+\frac{1}{2}\D {}^2\nabla_r \hat{h}_{AB}, \label{def:W1}\\
\dtwost W_2=&\frac{1}{2}\hat{h}_{AB}.\label{def:W2}
\end{align}

\begin{remark}
For the mode $\ell=1$, any even solution $h_{\ell=1}$ of \eqref{linear_gravity} equals $-{}^W\pi$ for a unique even vector field $W$ supported on $\ell=1$. This decomposition was done in the work of Zerilli \cite{Zerilli_2}. See also \cite{Martel-Poisson}. In this case $W$ satisfies the equation \eqref{wave_equation_vector} provided $h_{\ell=1}$ solves \eqref{HG}. However, the relation between $h$ and $W$ is \textbf{not} given by \eqref{def:W0}, \eqref{def:W1} and \eqref{def:W2}. We will treat this case separately in subsection \ref{subsec:ell=1}. For the mode $\ell=0$, any solution $h_{\ell=0}$ of \eqref{linear_gravity} is a linear combination of the mass perturbation $K$ (defined in \eqref{def:K}) and a deformation tensor ${}^W\pi$. We will discuss this case in section \ref{sec:ell=0}.
\end{remark}
Define the operator $\s{\Delta}_Z$ as in \cite{Johnson1,Johnson_2} by

\begin{align}
\s{\Delta}_Z:=\s{\Delta}+\frac{2}{r^2}-\frac{6M}{r^3}.
\end{align}
The Zerilli quantity is defined as

\begin{equation}\label{def:psi}
\begin{split}
\psi_Z:=&-2r^{-2}\D\s{\Delta}^{-1}\s{\Delta}_Z^{-1} \partial_r\s{tr}h-r^{-1}\s{\Delta}^{-1}\s{tr}h+4r^{-3}\D^2\s{\Delta}^{-1}\s{\Delta}_Z^{-1}h_{rr}\\
&-4r^{-2}\D\s{\Delta}_Z^{-1}(\donest)^{-1}h_{rA}+r^{-1}(\donest)^{-1}(\dtwost)^{-1}\hat{h}_{AB}.
\end{split}
\end{equation}
Zerilli \cite{Zerilli} first defined this quantity and derived the so called Zerilli equation, a wave equation with potential depending on $\ell$. Precisely, as $h_{ab}$ satisfies linearized Einstein equation \eqref{linear_gravity}, $r^{-1}\psi_Z$ satisfies. 

\begin{equation}\label{equ:psi}
\Box(r^{-1}\psi_{Z,\ell})-V_{Z,\ell}(r^{-1}\psi_{Z,\ell})=0,
\end{equation}
where

\begin{align*}
V_{Z,\ell}=-\frac{2M}{r^3}\frac{(2\lambda+3)(2\lambda+3M/r)}{(\lambda+3M/r)^2},
\end{align*}
and $\lambda(\ell)$ is defined as \eqref{def:lambda_ell}. Moreover, $\psi_Z$ is a gauge-invariant quantity as for any vector field $W=W_{\ell\geq 2}$, $\psi_Z[{}^W \pi]=0$. Furthermore, $h^{\RW}$ can be expressed in terms of $\psi_Z$ as




\begin{equation*}
g^{\alpha\beta}h^{\RW}_{\alpha\beta}=0,\ \hat{h}^{\RW}_{\alpha\beta}=\hat{\nabla}^2_{\alpha\beta} (r\psi_Z)+\frac{12M}{r^3}\s{\Delta}_Z^{-1}{\nabla}_\alpha r\hat{\otimes}_s {\nabla}_\beta\psi_Z,
\end{equation*}
and

\begin{equation*}
\begin{split}
\slashed{tr}h^{\RW}=-r\s{\Delta}\psi_Z+2\D\partial_r\psi_Z+\frac{12M}{r^4}\D\s{\Delta}_Z^{-1}\psi_Z.
\end{split}
\end{equation*}
These identities was derived in \cite{Zerilli_2}. See also \cite[Theorem 7.2.3]{Johnson1} for a modern treatment. In particular, 

\begin{align}\label{RW_by_psi}
|(M\partial)^{\leq m} h^{\RW}|\lesssim_m |s^2(M\partial)^{\leq m+2} r^{-1}\psi_Z|.
\end{align}


Suppose $\psi_Z=0$, then $h^{\RW}=0$ and $h=-{}^W\pi$. In this case \eqref{HG} is equivalent to tensorial wave equation \eqref{wave_equation_vector}. As $\psi_Z\neq 0$, the Zerilli quantity $\psi_Z$ becomes the source term in \eqref{wave_equation_vector}. In the following we rewrite \eqref{wave_equation_vector} in terms of three components $W_0, W_1$ and $W_2$. We fix an integer $\ell\geq 2$ and focus on this mode. To simplify the notation, we drop subscript $\ell$ and denote $\psi_{Z,\ell},W_{0,\ell}$, etc. simply by $\psi_Z,W_{0}$, etc.\\

Define the quantities

\begin{align}
S_W:=&-2\D^{-1}\nabla_t W_0+2\nabla_r W_1+2\done W_2+\frac{4}{r}W_1, \label{def:SW}\\
P_{\even}:=&-r\nabla_r W_0+\D^{-1}\cdot r\nabla_t W_1-W_0, \label{def:P}
\end{align}
and

\begin{equation}\label{def:Q}
\begin{split}
Q_{\even}:=&\donest W_1+\D {}^1\nabla_r W_2+\D \frac{1}{r}W_2-\frac{r}{2}\donest S_W-\frac{r^3}{3}  \donest  \slashed{\Delta}_Z( r^{-1}\psi_Z).
\end{split}
\end{equation}
Here $S_W$ is the trace of the deformation tensor ${}^W\pi$ and $P_{\even}$ as well as $Q_{\even}$ are components of the two form $dW$ provided $S_W$ and $\psi_Z$ are zero. Let $S:=trh$ be the trace of $h_{ab}$ and we note that $S=trh^{\RW}-S_W$.

\begin{proposition}

The quantities $W_0,W_1,W_2,S,P_{\even}$ and $Q_{\even}$ (to be precise, their projection on the mode $\ell$) satisfy the wave equations:

\begin{equation}\label{equ:W_0}
\Box W_0+\frac{2M}{r^3}W_0=-\frac{2M}{r^3}P_{\even}+\frac{2\lambda r^2-6\lambda Mr-6M^2}{r^2(\lambda r+3M)}\nabla_t\psi_Z +2\D\nabla_t\nabla_r\psi_Z.
\end{equation}

\begin{equation}\label{equ:W_1}
\begin{split}
&\Box W_1-\frac{2}{r^2}\left(1-\frac{4M}{r}\right)W_1-\frac{2}{r}\left(1-\frac{3M}{r}\right)\done W_2=-\frac{M}{r^2}S -\frac{6M(r-2M)^2}{r^3(\lambda r+3M)}\nabla_r\psi_Z+2\nabla_t\nabla_t\psi_Z\\ 
-&\frac{2M}{r^4(\lambda r+3M)^2}\bigg(\lambda(\lambda-3)(\lambda+1)r^3+3\lambda(4\lambda+1)M r^2+9(3\lambda-1)M^2 r+36M^3\bigg) \psi_Z.
\end{split}
\end{equation}

\begin{equation}\label{equ:W_2}
{}^1\Box W_2-\frac{1}{r^2}\D W_2-\frac{2}{r}\donest W_1=0.
\end{equation}

\begin{equation}\label{equ:S}
\Box S=0.
\end{equation}

\begin{equation}\label{equ:P}
\begin{split}
\Box P_{\even}=&-\nabla_t  S -6\D\nabla_t\nabla_r\psi_Z-\frac{2\lambda ((\lambda+3)r-3M  ) }{r(\lambda r+3M)}\nabla_t\psi_Z.
\end{split}
\end{equation}

\begin{equation}\label{equ:Q}
{}^1\Box Q_{\even}-\frac{1}{r^2}\D Q_{\even}=0.
\end{equation}

\end{proposition}

The derivation of these equations was done in \cite{Berndtson} and we present the computations in Appendix \ref{sec:equation_vector} for completeness. Our plan is to estimate $r^{-1}\psi_Z$ and $S$ first. Then it gives bound on the right hand side of \eqref{equ:P}, the equation of $P_{\even}$. However, the right hand side of \eqref{equ:P} has bad weight in $r$ and we can't directly apply vector field methods. To resolve this issue, we consider the substitution:

\begin{align}
\bar{P}_{\even}:=&P_{\even}+\frac{3}{2}r\nabla_t\psi_Z- \frac{u+r}{2r}(2\lambda+3)\cdot \psi_Z-\frac{u+r}{2}S,\\
\bar{W_0}:=&W_0-\frac{r}{2}\nabla_t\psi_Z+\frac{(u+r)}{2r}\psi_Z,\\
\bar{W_1}:=&W_1-\frac{r}{2}\D\nabla_r\psi_Z+\frac{3M}{2r}\D\left(\lambda+\frac{3M}{r}\right)^{-1}\psi_Z,\\
\bar{W_2}:=&W_2-\frac{r}{2}\donest\psi_Z.
\end{align}
Here $u$ is defined as

\begin{align*}
u:=t-r-2M\log\left(\frac{r}{2M}-1\right)+R_{\textup{null}}+2M\log\left(\frac{R_{\textup{null}}}{2M}-1\right)
\end{align*}
such that $u=\tau$ for $r\geq R_{\textup{null}}$. From \eqref{equ:psi}, \eqref{equ:S}, these quantities satisfy the following wave equations.

\begin{equation}\label{equ:bP}
\begin{split}
\Box \bar{P}_{\even}=&\D^{-1}\cdot\frac{2M}{r}\nabla_L S+\frac{M}{r^2}S+\D^{-1}\cdot (2\lambda+3)\frac{2M}{r}\nabla_L (r^{-1}\psi_Z)\\
+&(2\lambda+3)\cdot\frac{9M^2}{(\lambda r+3M)^2}\nabla_t(r^{-1}\psi_Z)-\frac{2\lambda+3}{2}\left( (u+r)V_Z-\frac{2M}{r^2} \right)(r^{-1}\psi_Z).
\end{split}
\end{equation}

\begin{equation}\label{equ:bW0}
\begin{split}
&\Box \bar{W}_0+\frac{2M}{r^3}\bar{W}_0=-\frac{2M}{r^3}\bar{P}_{\even}-\left(1+\frac{u}{r}\right)\frac{M}{r^2}S\\
-&\frac{3M^2}{(\lambda r+3M)^2}\cdot(2\lambda+3)\nabla_t (r^{-1}\psi_Z)-\frac{2M}{r}\D^{-1}\nabla_L(r^{-1}\psi_Z)\\
+&\left(\left( -(2\lambda+3)\frac{M}{r^2}+\frac{r}{2}V_Z \right)+\frac{u}{r}\left( -(2\lambda+2)\frac{M}{r^2}+\frac{r}{2}V_Z \right)\right)(r^{-1}\psi_Z).
\end{split}
\end{equation}

\begin{equation}\label{equ:bW1}
\begin{split}
&\Box\bar{W}_1-\frac{2}{r^2}\left(1-\frac{4M}{r}\right)\bar{W}_1-\frac{2}{r}\left(1-\frac{3M}{r}\right)\done \bar{W}_2
=-\frac{M}{r^2}S-\frac{\lambda(\lambda+1)(2\lambda+3)M}{(\lambda r+3M)^2}(r^{-1}\psi_Z).
\end{split}
\end{equation}

\begin{equation}\label{equ:bW2}
\begin{split}
{}^1\Box \bar{W}_2-\D \bar{W}_2-\frac{2}{r^2}\donest \bar{W}_1=\frac{\lambda(2\lambda+3)M r}{(\lambda r+3M)^2}\donest (r^{-1}\psi_Z).
\end{split}
\end{equation}

We used the above substitution to eliminate the terms with bad weights in $r$ under the price that weight in $u$ appears for $r^{-1}\psi_Z$ and $S$. In order to deal with this, we will need the rapid decay result of Angelopoulos-Aretakis-Gajic \cite{Angelopoulos-Aretakis-Gajic1}, recorded in Appendix \ref{sec:wave equation} as Proposition \ref{cor:free_price}. Moreover, $u$ is not regular on the horizon $\mathcal{H}^+$, so we perform an interpolation. Let $b>0$ with $bM\geq 2R_{\textup{null}}$ to be a large number. Its value will be determined in subsection \ref{subsec:estimate_W12} and we keep the dependence of $b$ in estimates until the value is fixed. Let $\eta_b(r)$ be a cutoff function with

\begin{align*}
\eta_b (r)=\left\{ \begin{array}{cc}
0 & r\in [2M, bM],\\
1 & r\in [2bM,\infty).
\end{array}
\right.
\end{align*}
Define 

\begin{align}
\hat{P}_{\even}:=&P_{\even}+\eta_b\cdot\left(\frac{3}{2}r\nabla_t\psi_Z- \frac{u+r}{2r}(2\lambda+3)\cdot \psi_Z-\frac{u+r}{2}S\right),\label{def:hP} \\
\hat{W}_0:=&W_0+\eta_b\cdot\left(-\frac{r}{2}\nabla_t\psi_Z+\frac{(u+r)}{2r}\psi_Z\right),\label{def:hW0} \\
\hat{W}_1:=&W_1+\eta_b\cdot\left(-\frac{r}{2}\D\nabla_r\psi_Z+\frac{3M}{2r}\D\left(\lambda+\frac{3M}{r}\right)^{-1}\psi_Z\right),\label{def:hW1}\\
\hat{W}_2:=&W_2+\eta_b\cdot\left(-\frac{r}{2}\donest\psi_Z\right).\label{def:hW2}
\end{align}
And denote

\begin{align}
G_P:=&\Box \hat{P}_{\even}, \label{equ:hP}\\
G_0:=&\Box \hat{W}_0+\frac{2M}{r}\hat{W}_0, \label{equ:hW0}\\
G_1:=&\Box \hat{W}_1-\frac{2}{r^2}\left(1-\frac{4M}{r}\right)-\frac{2}{r}\left(1-\frac{3M}{r}\right)\done \hat{W}_2, \label{equ:hW1}\\
G_2:=&{}^1\Box \hat{W}_2-\D \hat{W}_2-\frac{2}{r}\donest \hat{W}. \label{equ:hW2}
\end{align}

\begin{lemma}\label{lem:GP_bound}

Fix an integer $m\geq 0$. As $r\in [2M,bM]$, we have

\begin{align*}
|(M\partial)^{\leq m}G_P|\lesssim_{m,b} M^{-1}\left(|(M\partial)^{\leq m+1} S|+|(M\partial)^{\leq m+3}(r^{-1}\psi_Z)|\right).
\end{align*}
As $r\in [bM,2bM]$, we have

\begin{align*}
|(M\partial)^{\leq m}G_P|\lesssim_{m,b} M^{-1}\left( 1+\frac{\tau}{M} \right) \left(|(M\partial)^{\leq m+1}S|+ |(M\partial)^{\leq m+3}(r^{-1}\psi_Z)| \right).
\end{align*}
And as $r\in [2bM,\infty)$,

\begin{align*}
|(M\partial)^{\leq m}G_P|\lesssim_m &s^{-1}|(M\partial)^{\leq m} \nabla_L S|+M^{-1}s^{-2}|(M\partial)^{\leq m}S|\\
              +&s^{-1}|\nabla_L ((M\partial)^{\leq m} \K^{\leq 2} r^{-1}\psi_Z)|+s^{-2}|(M\partial)^{\leq m} \nabla_t (r^{-1}\psi_Z)|\\
              +&\left(1+\frac{\tau}{M}\right)M^{-1}s^{-2	}|(M\partial)^{\leq m} \K^{\leq 2} r^{-1}\psi_Z|.
\end{align*}

\end{lemma}

\begin{proof} 

We will only show the case $m=0$ and $m\geq 1$ can be proved by the same way. In $[2M,bM]$, $G_P=\Box P_{\even}$ and then

\begin{align*}
|\Box P_{\even}|\leq &\left|\nabla_t S\right|+\left|6\D\nabla_t\nabla_r\psi_Z\right|+\left|\frac{2\lambda ((\lambda+3)r-3M  ) }{r(\lambda r+3M)}\nabla_t\psi_Z\right|\\
\lesssim_b &M^{-1}\left(|(M\partial)^{\leq 1} S|+|(M\partial)^{\leq 3}(r^{-1}\psi_Z)|\right).
\end{align*}
Similarly, in $[2bM,\infty)$ $G_P=\Box \bar{P}_{\even}$ and from \eqref{equ:bP}

\begin{align*}
\begin{split}
|\Box \bar{P}_{\even}|\leq &\D^{-1}\cdot\frac{2M}{r}|\nabla_L S|+\frac{M}{r^2}|S|+\D^{-1}\cdot (2\lambda+3)\frac{2M}{r}|\nabla_L (r^{-1}\psi_Z)|\\
+&(2\lambda+3)\cdot\frac{9M^2}{(\lambda r+3M)^2}|\nabla_tr^{-1}\psi_Z|+\frac{2\lambda+3}{2}\left| (\tau+r)V_Z-\frac{2M}{r^2} \right||r^{-1}\psi_Z|\\
\lesssim &s^{-1}|\nabla_L S|+M^{-1}s^{-2}|S|\\
              +&s^{-1}|\nabla_L (\K^{\leq 2} r^{-1}\psi_Z)|+s^{-2}|\nabla_t (r^{-1}\psi_Z)|+\left(1+\frac{\tau}{M}\right)M^{-1}s^{-2}|\K^{\leq 2} r^{-1}\psi_Z|.
\end{split}
\end{align*} 
In $[bM,2bM]$, we compute

\begin{align*}
\Box\hat{P}_{\even}=&\Box P_{\even}+\Box\eta_b\cdot (\bar{P}_{\even}-P_{\even})\\
                   +&2\nabla\eta_b\cdot\nabla(\bar{P}_{\even}-P_{\even})+\eta_b\Box(\bar{P}_{\even}-P_{\even}).
\end{align*}
Since $bM\leq r\leq 2bM$, from \eqref{equ:P} we obtain

\begin{align*}
|\Box P_{\even}|\leq &\left|\nabla_t S\right|+\left|6\D\nabla_t\nabla_r\psi_Z\right|+\left|\frac{2\lambda ((\lambda+3)r-3M  ) }{r(\lambda r+3M)}\nabla_t\psi_Z\right|\\
\lesssim_b &M^{-1}|(M\nabla_t)S|+M^{-1}|(M\partial)^{\leq 3}(r^{-1}\psi_Z)|.
\end{align*}
Also, from the definition of $P_{\even}$ and $\bar{P}_{\even}$

\begin{align*}
|\bar{P}_{\even}-{P}_{\even}|\leq &\frac{3}{2}r^2|\nabla_tr^{-1}\psi_Z|+ \frac{u+r}{2}(2\lambda+3)\cdot |r^{-1}\psi_Z|+\frac{\tau+r}{2}|S|\\
\lesssim_b& M\left( 1+\frac{\tau}{M} \right)|(M\partial)^{\leq 2}r^{-1}\psi_Z|+M\left( 1+\frac{\tau}{M} \right)|S|.
\end{align*}
Similarly,

\begin{align*}
|\partial(\bar{P}_{\even}-{P}_{\even})| \lesssim_b& \left( 1+\frac{\tau}{M} \right)|(M\partial)^{\leq 3}r^{-1}\psi_Z|+\left( 1+\frac{\tau}{M} \right)|(M\partial)^{\leq 1}S|.
\end{align*}
Even though $\Box(\bar{P}_{\even}-{P}_{\even})$ involves one more derivative, since $\psi_Z$ and $S$ satisfy wave equations, we still have

\begin{align*}
|\Box(\bar{P}_{\even}-{P}_{\even})| \lesssim_b& M^{-1}\left( 1+\frac{\tau}{M} \right)|(M\partial)^{\leq 3}r^{-1}\psi_Z|+M^{-1}\left( 1+\frac{\tau}{M} \right)|(M\partial)^{\leq 1}S|.
\end{align*}
Together with $|\Box\eta_b|\lesssim_b M^{-2}, |\partial\eta_b|\lesssim_b M^{-1}$, the assertion follows by putting these inequalities together.

\end{proof}

\begin{lemma}\label{lem:G0_bound}

Fix an integer $m\geq 0$. In $r\in [2M,bM]$,

\begin{align*}
|(M\partial)^{\leq m}G_0|\lesssim_{m,b} M^{-2}|(M\partial)^{\leq m}\hat{P}_{\even}|+M^{-1}|(M\partial)^{\leq m+ 2}(r^{-1}\psi_Z)|.
\end{align*}
In $r\in [bM,2bM]$,

\begin{align*}
|G_0|\lesssim_{m,b} M^{-2}|(M\partial)^{\leq m}\hat{P}_{\even}|+M^{-1}\left(1+\frac{\tau}{M}\right)|(M\partial)^{\leq m}S|+M^{-1}\left(1+\frac{\tau}{M}\right)|(M\partial)^{\leq m+2}(r^{-1}\psi_Z)|.
\end{align*}
In $[2bM,\infty)$,

\begin{align*}
|(M\partial)^{\leq m}G_0|\lesssim_m &M^{-2}s^{-3}|(M\partial)^{\leq m}\hat{P}_{even}|+\left(1+\frac{\tau}{M}\right)M^{-1}s^{-2}|(M\partial)^{\leq m}S|\\
              +&s^{-2}|(M\partial)^{\leq m}\nabla_t ( r^{-1}\psi_Z)|+s^{-1}|(M\partial)^{\leq m}\nabla_L(r^{-1}\psi_Z)|\\
              +&M^{-1}s^{-2}\left(1+\frac{\tau}{M}\right)|(M\partial)^{\leq m}\K^{\leq 2} r^{-1}\psi_Z|.
\end{align*}

\end{lemma}

\begin{proof}

Again, we will only show that case $m=0$. As $r\in [2M,bM]$, $\hat{W}_0=W_0$ and $\hat{P}_{\even}=P_{\even}$. From \eqref{equ:W_0} we have

\begin{align*}
|G_0|\leq&\frac{2M}{r^3}|\hat{P}_{\even}|+\left|\frac{2\lambda r^2-6\lambda Mr-6M^2}{r^2(\lambda r+3M)}\right| |\nabla_t\psi_Z|+2\D|\nabla_t\nabla_r\psi_Z|\\
\lesssim_b &M^{-2}|\hat{P}_{\even}|+M^{-1}|(M\partial)^{\leq 2}(r^{-1}\psi_Z)|.
\end{align*}
As $r\in [2bM,\infty)$, $\hat{W}_0=\bar{W}_0$ and $\hat{P}_{\even}=\bar{P}_{\even}$. From \eqref{equ:bW0} we have

\begin{align*}
|G_0|\leq &\frac{2M}{r^3}|\hat{P}_{\even}|+ \left(1+\frac{\tau}{r}\right)\frac{M}{r^2}|S|\\
+&\frac{3M^2}{(\lambda r+3M)^2}\cdot(2\lambda+3)|\nabla_t (r^{-1}\psi_Z)|+\frac{2M}{r}\D^{-1}|\nabla_L(r^{-1}\psi_Z)|\\
+&\left|\left( -(2\lambda+3)\frac{M}{r^2}+\frac{r}{2}V_Z \right)+\frac{u}{r}\left( -(2\lambda+3)\frac{M}{r^2}+\frac{r}{2}V_Z \right)\right||(r^{-1}\psi_Z)|\\
\lesssim &M^{-2}s^{-3}|\hat{P}_{\textup{even}}|+\left(1+\frac{\tau}{M}\right)M^{-1}s^{-2}|S|\\
              +&s^{-2}|\nabla_t ( r^{-1}\psi_Z)|+s^{-1}|\nabla_L(r^{-1}\psi_Z)|+M^{-1}s^{-2}\left(1+\frac{\tau}{M}\right)|\K^{\leq 2}\cdot r^{-1}\psi_Z|.
\end{align*}
In $r\in [bM,2bM]$, we have

\begin{align*}
G_0=&\left(\Box+\frac{2M}{r^3}\right)\cdot W_0+\Box \eta_b\cdot (\bar{W}_0-W_0)\\
   +&2\nabla\eta_b\cdot\nabla (\bar{W}_0-W_0)+\eta_b \cdot (\Box+2M/r^3)(\bar{W}_0-W_0).
\end{align*}
From \eqref{equ:W_0},

\begin{align*}
\left(\Box+\frac{2M}{r^3}\right)\cdot W_0=&-\frac{2M}{r^3}P_{\even}+\frac{2\lambda r^2-6\lambda Mr-6M^2}{r^2(\lambda r+3M)}\nabla_t\psi_Z+2\D\nabla_t\nabla_r\psi_Z	\\
                                         =&-\frac{2M}{r^3}\hat{P}_{\even}+\frac{2\lambda r^2-6\lambda Mr-6M^2}{r^2(\lambda r+3M)}\nabla_t\psi_Z+2\D\nabla_t\nabla_r\psi_Z\\
                                         &+\eta_b\cdot\frac{2M}{r^3}\left(\frac{3}{2}r\nabla_t\psi_Z- \frac{u+r}{2r}(2\lambda+3)\cdot \psi_Z-\frac{u+r}{2}S\right).
\end{align*}
Therefore

\begin{align*}
\left|\left(\Box+\frac{2M}{r^3}\right)W_0\right| \lesssim_b M^{-2}|\hat{P}_{\even}|+M^{-1}\left(1+\frac{\tau}{M}\right)|S|+M^{-1}\left(1+\frac{\tau}{M}\right)|(M\partial)^{\leq 2}(r^{-1}\psi_Z)|.
\end{align*}
From the definition of $\bar{W}_0$, in $r\in[bM,2bM]$

\begin{align*}
|\bar{W}_0-W_0|\leq &r|\nabla_t\psi_Z|+\frac{(\tau+r)}{2r}|\psi_Z|\lesssim_b M\left(1+\frac{\tau}{M}\right)|(M\partial)^{\leq 1}(r^{-1}\psi_Z)|.
\end{align*}
Similarly, in $r\in[bM,2bM]$

\begin{align*}
M|\Box(\bar{W}_0-W_0)|,|\nabla(\bar{W}_0-W_0)|\lesssim_b &\left(1+\frac{\tau}{M}\right)|(M\partial)^{\leq 2}(r^{-1}\psi_Z)|
\end{align*}
Hence the assertion follows by putting these estimates together.

\end{proof}

Similarly, we have the bound on $G_1$ and $G_2$ below.

\begin{lemma}\label{lem:G12_bound}

In $r\in [2M,2bM]$,

\begin{align*}
|(M\partial)^{\leq m} G_1|\lesssim_{m,b} &M^{-1}|(M\partial)^{\leq m}S|+M^{-1}|(M\partial)^{\leq m+2}(r^{-1}\psi_Z)|,\\
|(M\partial)^{\leq m}G_2|\lesssim_{m,b} &M^{-1}|(M\partial)^{\leq m+2}(r^{-1}\psi_Z)|.
\end{align*}
In $r\in [2bM,\infty)$,

\begin{align*}
|(M\partial)^{\leq m}G_1|\lesssim_m M^{-1}s^{-2}|(M\partial)^{\leq m}S|+M^{-1}s^{-2}|(M\partial)^{\leq m}\K^{\leq 2}\cdot r^{-1}\psi_Z|,\\
|(M\partial)^{\leq m}G_2|\lesssim_m M^{-1}s^{-2}|(M\partial)^{\leq m}S|+M^{-1}s^{-2}|(M\partial)^{\leq m}\K^{\leq 1}\cdot r^{-1}\psi_Z|.
\end{align*}

\end{lemma}

\subsection{The $\ell=1$ mode}\label{subsec:ell=1}

In this subsection, we discuss the case that the \textbf{even} two tensor $h_{ab}$ is supported on $\ell=1$. It is well known that in this mode, any solution of linear gravity \eqref{linear_gravity} is of the form $-{}^W\pi$ \cite{Zerilli_2,Martel-Poisson}. We put a minus sign before ${}^W\pi$ to be consistent with the notation in $\ell\geq 2$. If $h_{ab}$ further satisfies the harmonic gauge \eqref{HG}, then the vector field $W$ is a solution of \eqref{wave_equation_vector}. As in the mode $\ell=1$, there is no spherical symmetric traceless two tensor, equation \eqref{def:W2} holds trivially and doesn't determine $W_2$. In stead, we rely on the equalities below to express $W_a$ in terms of $h_{ab}$. From \eqref{pi} and the commutation relation \eqref{comm_r}, \eqref{comm_t}, one derives that on any fixed mode $\ell\geq 1$

\begin{equation}\label{def:W0_ell=1}
\begin{split}
&\frac{r^3}{12M}\nabla_t\s{tr}h-\frac{r^3}{3\Lambda M} \done h_{tA}-\frac{r^2}{6M}\D h_{tr}\\
+&\frac{r^4}{6\Lambda M}\D \nabla_t\done h_{rA}-\frac{r^4}{6\Lambda M}\nabla_r \D\done h_{tA}\\
=&-W_0-\frac{r^3}{6M}(1-2\Lambda^{-1})\nabla_t \done W_2,
\end{split}
\end{equation}

\begin{equation}\label{def:W1_ell=1}
\begin{split}
&\frac{r^3}{12M}\D\nabla_r \s{tr}h-\frac{r^3}{3\Lambda M}\D  \done h_{rA}-\frac{r^2}{6M}\D^2 h_{rr}\\
=& -W_1-\frac{r^3}{6M}\D (1-2\Lambda^{-1})\cdot \nabla_r \done W_2,
\end{split}
\end{equation}

\begin{equation}\label{def:W2_ell=1}
\begin{split}
&-\frac{r^4}{12M}\D {}^1\nabla_r\donest \s{tr}h-\frac{r^3}{12M}\left(1-\frac{5M}{r}\right)\donest \s{tr}h\\
+&\frac{r^3}{6M}\D^2\donest h_{rr}+\frac{r^2}{3M}\D h_{rA}\\
=&-W_2-(\Lambda-2)\left(-\frac{r^2}{6M}\D{}^1\nabla_r W_2+\frac{r}{6M}\left(1+\frac{M}{r}\right)W_2 \right).
\end{split}
\end{equation}
Here $\Lambda$ is given by \eqref{def:lambda_ell}. In particular, as $\ell=1$ and $\Lambda=2$, the above three equalities give the expression of $W_a$ in terms of $h_{ab}$.

\section{Analysis of $\hat{W}_0$}\label{sec:W0}

In this and the next sections, we study in solution of the wave equations \eqref{equ:W_0}, \eqref{equ:W_1} and \eqref{equ:W_2} with $\psi_Z$ satisfying the wave equation \eqref{equ:psi}. We note that by putting $\psi_Z\equiv 0$, these three equations is the \textbf{even} part of the equation \eqref{wave_equation_vector}. We further restrict ourselve in the case $W=W_{\ell\geq 1}$ in this and the next sections. The case $W_{\ell=0}$ will be discussed in the section \ref{sec:ell=0}.\\

We will prove the decay estimate for $\hat{W}_0$, Proposition \ref{pro:X0}, based on the vector field method in Appendix \ref{sec:wave equation}. We start with $S$ and $r^{-1}\psi_Z$, which satisfy wave equations without source term, \eqref{equ:psi} and \eqref{equ:S}. The bound on $S$ and $r^{-1}\psi_Z$ then  is translated to the one for $G_P$, the source term in \eqref{equ:hP}. After that, we have the bound on $G_0$, the source term in \eqref{equ:hW0}, and then get the estimate for $\hat{W}_0$. The main result is the following:

\begin{proposition}\label{pro:X0} 

Let $W_0$, $r^{-1}\psi_Z$, $P_{\even}$ and $S$ be solutions of \eqref{equ:W_0}, \eqref{equ:psi}, \eqref{equ:P} and \eqref{equ:S} respectively and are supported on $\ell\geq 1$. We further assume $\psi_{Z,\ell=1}=0$. Let $\hat{W}_0$ and $\hat{P}_{\even}$ be defined as in \eqref{def:hW0} and \eqref{def:hP}. Suppose $\textup{Decay}[\hat{W}_0]$ defined below is finite, then for any $p\in [\delta,2-3\delta]$ and $\tau \geq 0$ we have,

\begin{equation}
\begin{split}
&F[\hat{W}_0](\tau)+E^{p,(0)}_L[\hat{W}_0](\tau)+M^{-1}\int_{\tau}^{\infty} B[\hat{W}_0](\tau')+E^{p-1,(0)}_{L,\nablas}[\hat{W}_0](\tau')d\tau'\\
&\lesssim_{b,R_{\textup{nul}}} \left(1+\frac{\tau}{M}\right)^{-2+p+3\delta}\cdot \textup{Decay}[\hat{W}_0].
\end{split}
\end{equation}
Here

\begin{equation}
\begin{split}
\textup{Decay}[\hat{W}_0]:=&I^{(0)}[\hat{W}_0,\K^{\leq 2}\hat{P}_{\textup{even}}]+M^2 I^{(0)}[\f^{\leq 1}\K^{\leq 4}S,\f^{\leq 3}\K^{\leq 4}S]+M^2 I^{(1)}[\K^{\leq 4}S,\K^{\leq 6}r^{-1}\psi_Z].
\end{split}
\end{equation}

\end{proposition}

\subsection{Analysis of $S$ and $r^{-1}\psi_Z$}\label{subsec:psi_estimate}

We begin with $S=trh$, which satisfies the wave equation \eqref{equ:S}. From remark \ref{rem:0} we have

\begin{proposition}\label{pro:S}
Let $S$ be an solution of \eqref{equ:S}. Suppose $I^{(1)}[S]$ is finite, then for any $p\in [\delta,2-\delta]$ and $\tau\geq 0$, we have

\begin{equation}
\begin{split}
&\left( F[S]+E^{p,(0)}_L[\tilde{S}]\right)(\tau)+M^{-1}\int_{\tau}^{\infty}\left( B[S]+E^{p-1,(0)}_L[ \tilde{S}]\right)(\tau ')d\tau '\\
\lesssim &\left( 1+\frac{\tau}{M} \right)^{-4+p+3\delta}I^{(1)}[S].
\end{split}
\end{equation}
Here $\tilde{S}=\eta_{\textup{null}}S$ and $\eta_{\textup{null}}$ is the cut-off function define in \eqref{cutoff_rp}. The initial norm $I^{(1)}[\cdot]$ is defined in \eqref{initial_norm}.

\end{proposition}




\begin{lemma}
There exists $R_Z$ large enough such that for any $R_{\textup{null}}\geq R_Z$ and $\ell\geq 2$, the potential $V_{Z,\ell}$ for the Zerilli equation \eqref{equ:psi} is a member of $\in\mathcal{V}(\delta,R_{\textup{null}},\ell,1)$ defined in \ref{def:potential2}.  Moreover, the constants in  \eqref{potential_Mor} and \eqref{potential_pk_4} for $k=1$ can be chosen independent of $\ell$.
\end{lemma}

Therefore Proposition \ref{cor:source_high_order} and Proposition \ref{cor:free_price} with $k=1$ apply to $r^{-1}\psi_{Z,\ell}$ with constants being uniform in $\ell\geq 2$.

\begin{proof}
For $R_{\textup{null}}$ large enough, $V_{Z,\ell}\in\mathcal{V}(\delta,R_{\textup{null}},2)$ is proved in \cite[Theorem 2]{Johnson1} and \cite[Theorem 16]{Hung-Keller-Wang}. Here we check $V_{Z,\ell}$ further satisfies \eqref{potential_pk_1}, \eqref{potential_pk_3} and \eqref{potential_pk_4}. Recall that

\begin{align*}
V_{Z,\ell}^{(1)}=V_{Z,\ell}-\frac{2}{r^2}+\frac{14M}{r^3}.
\end{align*}
Through direct computation,

\begin{align*}
&\left( (1-p/2)(2\lambda+2)-\frac{1}{2}r^3\left( \frac{d V_{Z,\ell}^{(1)}}{dr}+\frac{p}{r}V_{Z,\ell}^{(1)} \right) \right)\times\left(1+\frac{3M}{\lambda r}\right)^3\\
=& (2-p)(\lambda-1)+s^{-1}\cdot Q_1[\lambda^{-1},s^{-1},p],
\end{align*}
where $Q_1[\lambda^{-1},s^{-1},p]$ is a polynomial in $\lambda^{-1},s^{-1},p$. Hence as $R_{Z}$ is large enough, one has for any $p\in [\delta,2-\delta]$, $\lambda\geq 2$ and $r\geq R_{\text{null}}$

\begin{align*}
(2-p)(\lambda-1)+s^{-1}\cdot Q_1[\lambda^{-1},s^{-1},p]\geq \frac{1}{2}(2-p)(\lambda-1),
\end{align*}
which verifies \eqref{potential_pk_1}. For \eqref{potential_pk_3}, one checks that

\begin{align*}
\lim_{r\to\infty}r^2V^{(1)}_{Z,\ell}=-2.
\end{align*}
Finally, \eqref{potential_pk_4} clearly holds for each $V_{Z,\ell}$. The uniformness comes from 

\begin{align*}
\lim_{\ell\to\infty} V_{Z,\ell}=-\frac{8M}{r^3}.
\end{align*}

\end{proof}

Therefore from Proposition \ref{cor:free_price} we obtain

\begin{proposition}\label{pro:psi}
Let $\psi_Z=\psi_{Z,\ell\geq 2}$ be a solution of \eqref{equ:psi}. Suppose that $I^{(1)}[r^{-1}\psi_Z]$ is finite, then for any $p\in [\delta,2-\delta]$ and $\tau\geq 0$, we have

\begin{equation}
\begin{split}
&\left( F[ r^{-1}\psi_Z ]+E^{p,(0)}_L[r^{-1}\tilde{\psi}_Z]\right)(\tau)+M^{-1}\int_{\tau}^{\infty}\left( B[r^{-1}\psi_Z]+E^{p-1,(0)}_L[r^{-1}\tilde{\psi}_Z]\right)(\tau')d\tau'\\
\lesssim &\left( 1+\frac{\tau}{M} \right)^{-4+p+3\delta}I^{(1)}[r^{-1}\psi_Z],
\end{split}
\end{equation}
Here $\tilde{\psi}_Z=\eta_{\textup{null}} \psi_Z$ and $\eta_{\textup{null}}$ is the cut-off function define in \eqref{cutoff_rp}. The initial norm $I^{(1)}[\cdot]$ is defined in \eqref{initial_norm}.

\end{proposition}

\subsection{Analysis of $\hat{P}_{\textup{even}}$}

We move to equation \eqref{equ:hP} and estimate $\hat{P}_{\even}$. To begin we translate the estimate of $S$ and $r^{-1}\psi_Z$ above to the one for $G_P$ through Lemma \ref{lem:GP_bound}.

\begin{lemma}\label{lem:GP_decay}

\begin{align*}
 \textup{I}_{\textup{source},3\delta}[G_P]\lesssim_{b,R_{\textup{null}}} &M^2 I^{(0)}[\f^{\leq 1}\K^{\leq 2}S,\f^{\leq 3}\K^{\leq 2}r^{-1}\psi_Z]+M^2 I^{(1)}[\K^{\leq 2}S,\K^{\leq 4}r^{-1}\psi_Z].
\end{align*}
Here $\textup{I}_{\textup{source},3\delta}[\cdot]$ and $I^{(1)}[\cdot]$ are define in \eqref{source_assumption} and \eqref{initial_norm}.

\end{lemma}

\begin{proof}

We first deal with the spacetime integrand $Ms^{p+1}|G_P|^2$. For $r\in [2M,bM]$, from Lemma \ref{lem:GP_bound}, we have

\begin{align*}
Ms^{p+1}|G_P|^2\lesssim_b M^{-1}|(M\partial)^{\leq 1} S|^2+M^{-1}|(M\partial)^{\leq 3} r^{-1}\psi_Z|^2,
\end{align*}
which is bounded by the integrand of $M B[\f^{\leq 1}S]$ and $M B[\f^{\leq 3} r^{-1}\psi_Z]$. Note that we lose one derivative here because of the photon sphere. Then from Proposition \ref{cor:source_high_order},

\begin{align*}
\int_{D(\tau_1,\tau_2)} Ms^{p+1}|G_P|^2\cdot\chi_{[2M, bM]} dvol\lesssim_{b,R_{\textup{null}}} &M^2 F[\f^{\leq 1}S,\f^{\leq 3} r^{-1}\psi_Z](\tau_1)\\
\lesssim_{R_{\textup{null}}} &\left(1+\frac{\tau_1}{M}\right)^{-2+2\delta}\cdot M^2 I^{(0)}[\f^{\leq 1}S,\f^{\leq 3}r^{-1}\psi_Z].
\end{align*}
For $r\in [bM,2bM]$, from Lemma \ref{lem:GP_bound}, we have

\begin{align*}
Ms^{p+1}|G_P|^2\lesssim_b M^{-1}\left(1+\frac{\tau}{M}\right)^2|(M\partial)^{\leq 1} S|^2+M^{-1}\left(1+\frac{\tau}{M}\right)^2|(M\partial)^{\leq 3} r^{-1}\psi_Z|^2,
\end{align*}
which is bounded by the integrand of $\left(1+\frac{\tau}{M}\right)^2\cdot M B[S]$ and $\left(1+\frac{\tau}{M}\right)^2\cdot M B[\f^{\leq 2} r^{-1}\psi_Z]$.  Note that in this region $T$ is strictly timelike and from \eqref{killing_to_partial} we can replace $M B[\f^{\leq 2} r^{-1}\psi_Z]$ by $M B[\K^{\leq 2} r^{-1}\psi_Z]$. From Proposition \ref{pro:S} and \ref{pro:psi}, we have

\begin{align*}
M^{-1}\int_{\tau_1}^{\infty}B[S](\tau)+B[\K^{\leq 2} r^{-1}\psi_Z](\tau)d\tau\lesssim_{R_{\textup{null}}} \left(1+\frac{\tau_1}{M}\right)^{-4+4\delta}I^{(1)}[S,\K^{\leq 2}r^{-1}\psi_Z].
\end{align*}
Therefore

\begin{align*}
\int_{D(\tau_1,\infty)} Ms^{p+1}|G_P|^2\cdot\chi_{[bM,2bM]}\ dvol \lesssim_{b,R_{\textup{null}}} \left(1+\frac{\tau_1}{M}\right)^{-2+4\delta}\cdot M^2 I^{(1)}[S,\K^{\leq 2}r^{-1}\psi_Z].
\end{align*}
For $r\in [2bM,\infty]$, from Lemma \ref{lem:G0_bound} we have

\begin{align*}
Ms^{p+1}|G_P|^2\lesssim &M s^{p-1}|\nabla_L S|^2+M^{-1}s^{p-3}|S|^2+M s^{p-1}|\nabla_L \K^{\leq 2} r^{-1}\psi_Z|^2\\
+&M s^{p-3}|\nabla_t  r^{-1}\psi_Z|^2+\left(1+\frac{\tau}{M}\right)^2M^{-1}s^{p-3}|\K^{\leq 2}r^{-1}\psi_Z|^2.
\end{align*}
The terms involving $S$ are bounded by the integrand of $M E^{p-1,(0)}_{L,\nablas}[S]$ and the terms involving $r^{-1}\psi_Z$ are bounded by

\begin{align*}
\left(1+\frac{\tau}{M}\right)^2\cdot M E^{p-1,(0)}_{L,\nablas}[\K^{\leq 2} r^{-1}\psi_Z].
\end{align*}
By Proposition \ref{pro:S} and Proposition \ref{pro:psi}, we obtain

\begin{align*}
\int_{D(\tau_1,\tau_2)\cap\{r\geq 2bM\}} Ms^{p+1}|G_P|^2 dvol \lesssim_{R_{\textup{null}}} &\left(1+\frac{\tau_1}{M}\right)^{-2+p+\delta}\cdot M^2 I^{(0)}[S]\\
+&\left(1+\frac{\tau_1}{M}\right)^{-2+p+3\delta}\cdot M^2 I^{(1)}[\K^{\leq 2}r^{-1}\psi_Z].
\end{align*}

Next we estimate the spacetime integrand $\left(1-\frac{3M}{r}\right)^2|\partial G|^2+\D^2 |\partial_r G|^2$. Note that $bM\geq R_{\textup{null}}$. In this region

\begin{align*}
M^3|\partial G_P|^2\lesssim &M|\partial \f^{\leq 1} S|^2+M|\partial \f^{\leq 3} r^{-1}\psi_Z|^2,\\
M^3|\partial_r G_P|^2\lesssim &M|\partial_r \f^{\leq 1} S|^2+M|\partial_r \f^{\leq 3} r^{-1}\psi_Z|^2,
\end{align*}
which are bounded by the integrand of $M B[\f^{\leq 1} S]$ and $M B[\f^{\leq 3} S]$. By Proposition \ref{cor:source_high_order}, we have

\begin{align*}
M^3\int_{D'(\tau_1,\tau_2)} \left(1-\frac{3M}{r}\right)^2|\partial G_P|^2+\D^2|\nabla_r G_P|^2 dvol&\lesssim M\int_{\tau_1}^{\tau_2} B[\f^{\leq 1} S,\f^{\leq 3} r^{-1}\psi_Z](\tau) d\tau\\
&\lesssim_{R_{\textup{null}}} \left(1+\frac{\tau_1}{M}\right)^{-2+2\delta} I^{(0)}[\f^{\leq 1} S,\f^{\leq 3} r^{-1}\psi_Z].
\end{align*}

Finally, the integrals along $\Sigma_{\tau_1}''$ is bounded as

\begin{align*}
\int_{\Sigma''_{\tau_1}}M^2|G_P|^2 dvol_3&\lesssim \int_{\Sigma''_{\tau_1}} |(M\partial)^{\leq 1} S|^2+|(M\partial)^{\leq 3} r^{-1}\psi_Z|^2 dvol_3\\
&\lesssim_{R_{\textup{null}}} \left(1+\frac{\tau_1}{M}\right)^{-2+2\delta}\cdot M^2 I^{(0)}[ S,\f^{\leq 2} r^{-1}\psi_Z].
\end{align*}
Putting these together, we have for any $p\in [\delta,2-3\delta]$,

\begin{align*}
\left(1+\frac{\tau_1}{M}\right)^{2-p-3\delta} \textup{E}_{\textup{source}}^p[G_P](\tau_1,\tau_2)\lesssim_{b,R_{\textup{null}}} M^2I^{(0)}[\f^{\leq 1}S,\f^{\leq 3}r^{-1}\psi_Z]+M^2I^{(1)}[S,\K^{\leq 2}r^{-1}\psi_Z].
\end{align*}

Then the result follows form the definition of $\textup{I}_{\textup{source},3\delta}[\cdot]$ in \eqref{source_assumption}.

\end{proof}

Now we can apply Proposition \ref{cor:source_high_order} to show that $\hat{P}_{\even}$ decays.

\begin{proposition}\label{pro:P}

Let $P_{\even}$, $S$ and $r^{-1}\psi_Z$ be solutions of \eqref{equ:P}, \eqref{equ:S} and \eqref{equ:psi} respectively. Let $\hat{P}_{\even}$ be defined as in \eqref{def:hP}. Suppose $\textup{Decay}[\hat{P}_\even]$ defined below is finite, then for any $p\in [\delta,2-3\delta]$ and $\tau\geq 0$ we have

\begin{equation}
\begin{split}
&F[\hat{P}_{\even}](\tau)+E^{p,(0)}_L[\hat{P}_{\even}](\tau)+M^{-1}\int_{\tau}^{\infty} B[\hat{P}_{\even}](\tau)+E^{p-1,(0)}_{L,\nablas}[\hat{P}_{\even}](\tau)d\tau\\
& \lesssim_{b,R_{\textup{null}}} \left(1+\frac{\tau}{M}\right)^{-2+p+3\delta}\cdot \textup{Decay}[\hat{P}_\even].
\end{split}
\end{equation}
Here

\begin{align}
\textup{Decay}[\hat{P}_\even]:=I^{(0)}[\hat{P}_{\textup{even}}]+M^2 I^{(0)}[\f^{\leq 1}\K^{\leq 2}S,\f^{\leq 3}\K^{\leq 2} r^{-1}\psi_Z]+M^2 I^{(1)}[\K^{\leq 2}S,\K^{\leq 4}r^{-1}\psi_Z],
\end{align}
and $I^{(k)}[\cdot]$ is defined in \eqref{initial_norm}.

\end{proposition}

\subsection{proof of Proposition \ref{pro:X0}}

We turn to equation \eqref{equ:hW0} and prove the Proposition \ref{pro:X0}. Still, we begin with the source term $G_0$. Comparing Lemma \ref{lem:GP_bound} and Lemma \ref{lem:G0_bound}, we are able to bound $G_0$ with the estimate of $G_P$ and $\hat{P}_{\even}$.

\begin{lemma}\label{lem:G0_decay}

\begin{align*}
\textup{I}_{\textup{source},3\delta}[G_0] \lesssim_{b,R_{\textup{null}}} &I^{(0)}[\K^{\leq 2}\hat{P}_{\textup{even}}]+M^2 I^{(0)}[\f^{\leq 1}\K^{\leq 4}S,\f^{\leq 3}\K^{\leq 4}r^{-1}\psi_Z]+M^2 I^{(1)}[\K^{\leq 4}S,\K^{\leq 6}r^{-1}\psi_Z].
\end{align*}
Here $\textup{I}_{\textup{source},3\delta}[\cdot]$ and $I^{(k)}[\cdot]$ are define in \eqref{source_assumption} and \eqref{initial_norm}.
\end{lemma}

Next, we verify the potential of \eqref{equ:hW0} is a member of $\mathcal{V}(\delta,R_{\textup{null}},1)$.

\begin{lemma}\label{lem:V0}
For any $R_{\textup{null}}\geq 10M$, $-2M/r^3$ belongs to $\mathcal{V}(\delta,R_{\textup{null}},1)$ defined in \eqref{def:potential1}.
\end{lemma}

\begin{proof}
The requirements other than the Morawetz one \eqref{potential_Mor} can be verified through direct computation. The choice of Moratwez function $f(r)$ is 

\begin{align*}
f(r):=&\left(1-\frac{3M}{r}\right)\left(1+\frac{2M}{r}+\frac{2M^2}{r^2}\right),\\
\omega(r)=&\D\left(\frac{2f}{r}+\frac{df}{dr}\right).
\end{align*}
Note that $\frac{d f}{dr}=\frac{M}{r^2}+\frac{8M^2}{r^3}+\frac{18M^3}{r^4}>0.$ And we calculate for $V(r)=-2M/r^3$,

\begin{align*}
&f\left(1-\frac{3M}{r}\right)\cdot \frac{2}{r^3}  +\left(-\frac{1}{4}\Box\omega{\cblue-\frac{1}{2}f\D\frac{d V}{d r}-\frac{M}{r^2}fV}\right)\\
=&r^{-3} \left(2 - \frac{15}{2} s^{-1} - 4s^{-2} + 7s^{-3}+ 130s^{-4}  - 198s^{-5}\right),
\end{align*}
which is positive in $r\geq 2M$. Hence \eqref{potential_Mor} holds for $\ell=1$.

\end{proof}

\begin{proof}[proof of Proposition \ref{pro:X0}]
From the view of Lemma \ref{lem:G0_decay} and Lemma \ref{lem:V0}, Proposition \ref{cor:source_high_order} yields the desired result.
\end{proof}

Recall that $Q_{\even}$ satisfies the wave equation \eqref{equ:Q}. 

\begin{proposition}\label{pro:Q}

Let $Q_{\even}$ be a solution of \eqref{equ:Q}. Suppose $I^{(0)}[Q_{\even}]$ is finite, then for any $p\in [\delta,2-\delta]$ and $\tau\geq 0$, we have

\begin{equation}
\begin{split}
&F[Q_{\even}](\tau)+E^{p,(0)}_L[Q_{\even}](\tau)+M^{-1}\int_{\tau}^{\infty} B[Q_{\even}](\tau)+E^{p-1,(0)}_{L,\nablas}[Q_{\even}](\tau)d\tau\\
& \lesssim_{R_{\textup{null}}} \left(1+\frac{\tau}{M}\right)^{-2+p+\delta}\cdot I^{(0)}[Q_{\even}].
\end{split}
\end{equation}
Here $I^{(k)}[\cdot]$ is defined in \eqref{initial_norm}.

\end{proposition}

\begin{proof}
This actually a consequence of $-2M/r^3\in\mathcal{V}(\delta,R_{\textup{null}},1)$ proved in Lemma \ref{lem:V0}. Note that $Q_{\even}$ is a section of $\mathcal{L}(-1)$ and is automatically supported on $\ell\geq 1$. This is also the reason that we consider the potential $-2M/r^3$ in stead of $\frac{1}{r^2}\D$. The equation \eqref{equ:Q} is equivalent to \eqref{equ:hW0} without $G_0$ in terms application of Proposition \ref{cor:source_high_order}. This can be seen through the communication relation

\begin{align*}
r\done \cdot {}^1\Box=\left(\Box+\frac{1}{r^2} \right)\cdot r\done,
\end{align*}
which implies $r\done Q_{\even}$ satisfies the equation $(\Box+2M/r^3)\cdot r\done Q_{\even}=0$. Then the result follows from Lemma \ref{lem:V0} and Proposition \ref{cor:source_high_order}

\end{proof}

\section{Analysis of $\hat{W}_1$ and $\hat{W}_2$}\label{sec:W12}

In this section we study the remaining two quantities $\hat{W}_1$ and $\hat{W}_2$ in the vector field $\hat{W}$. Recall that $\hat{W}_1$ and $\hat{W}_2$ satisfy a coupled wave equation \eqref{equ:hW1} and \eqref{equ:hW2}. We denote $\W:=(\hat{W}_1,\hat{W}_2)$ and $\mathsf{G}:=(G_1,G_2)$. The equations \eqref{equ:hW1} and \eqref{equ:hW2} can be rewritten as

\begin{align}\label{equ:W12}
\Box\W=A_{\W}\W+\mathsf{G},
\end{align}
where

\begin{align*}
A_{\W}=\frac{1}{r^2}\left[\begin{array}{cc}
2\left(1-\frac{4M}{r}\right) & 2r\left(1-\frac{3M}{r}\right)\done\\
2r\donest & \D
\end{array}\right].
\end{align*}
The main result of this section is

\begin{proposition}\label{pro:X12}

Let $W_0, W_1, W_2$, $r^{-1}\psi_Z$, and $S$ be solutions of \eqref{equ:W_0}, \eqref{equ:W_1}, \eqref{equ:W_2}, \eqref{equ:psi}, and \eqref{equ:S} respectively and are supported on $\ell\geq 1$. We further assume that $\psi_{Z,\ell=1}=0$ and that $W_2$ is \textbf{even}. Let $\hat{W_0}$, $\hat{W}_1$, $\hat{W}_2$ and $\hat{P}_{\even}$ be defined as in \eqref{def:hW0}, \eqref{def:hW1}, \eqref{def:hW2} and \eqref{def:hP}. Suppose $\textup{Decay}[\W]$ defined below is finite. Then for any $p\in [\delta,p-3\delta]$ and $\tau\geq 0$, we have

\begin{equation}
\begin{split}
&F[\W](\tau)+E^{p,(0)}_L[\tilde{\W}](\tau)+M^{-1}\int_{\tau}^{\infty} \bar{B}[\W](\tau')+E^{p-1,(0)}_{L,\nablas}[\tilde{\W}](\tau')d\tau'\\
\lesssim &\left(1+\frac{\tau}{M}\right)^{-2+p+3\delta} \textup{Decay}[\W].
\end{split}
\end{equation}
Here $\W=(\hat{W}_1,\hat{W}_2)$, $\tilde{\W}=\eta_{\textup{null}}\W$ with $\eta_{\textup{null}}$ defined in \eqref{cutoff_rp},

\begin{align*}
\textup{Decay}[\W]:=& F[\W](0)+E^{2-\delta,(0)}_L[\tilde{\W}](0)+ I^{(0)}[\K^{\leq 1}\hat{W}_0,\K^{\leq 3}\hat{P}_{\even}]\\
+&M^2I^{(0)}[\f^{\leq 1}\K^{\leq 5}S,\f^{\leq 3}\K^{\leq 5}r^{-1}\psi_Z,Q_{\even}]+M^2 I^{(1)}[\K^{\leq 5}S,\K^{\leq 7}r^{-1}\psi_Z],
\end{align*}
and $I^{(k)}[\cdot]$ is defined in \eqref{initial_norm}.

\end{proposition}

From now on we assume that for all $\tau\geq 0$,

\begin{equation}\label{assumption}
\begin{split} 
E^{2-\delta,(0)}[\tilde{\W}](\tau)<&\infty,\\
\int_{0}^\tau E^{2-\delta,(0)}[\tilde{\W}](\tau') d\tau'<&\infty.
\end{split}
\end{equation}
The purpose of the assumption is to add a zeroth order term into the boundary term of $T$-current in subsection \ref{subsection:T} through Corollary \ref{cor:F_base}. This assumption is actually satisfied for any smooth solution with $E^{2-\delta,(0)}[\tilde{\W}](0)<\infty$. In subsections below, we construct for all $p\in [\delta,2-\delta]$ a current $J^p_{\W}$ such that for any $\tau\geq 0$,

\begin{align}\label{boundary}
\int_{\Sigma_\tau} J^p_{\W}\cdot n dvol_3 \approx &F[\W](\tau)+E^{p,(0)}[\tilde{\W}](\tau).
\end{align}
And for any $r\geq R_{\textup{null}}+M$,

\begin{align}
\int_{\mathbb{S}^2(\tau,2M)} J^p_{\W}\cdot L\ dvol_{\mathbb{S}^2}\geq &0, \label{boundary_horizon}\\
\left(\int_{\mathbb{S}^2(\tau,r)} J^p_{\W}\cdot \underline{L}\  r^2dvol_{\mathbb{S}^2}\right)_{-}\lesssim &s^{-(2-\delta)} M^{-1}E^{2-\delta,(0)}_L[\tilde{\W}](\tau) \label{boundary_null},
\end{align}
where $(\cdot)_{-}$ stands for the negative part. For we want to show that

\begin{align}\label{bulk}
\int_{\Sigma_\tau} \Div J^p_{\W}-\textup{Err}[\W,\mathsf{G}]-{Err}[\W,\mathsf{F}] dvol_3 \gtrsim_{b,R_{\textup{null}}}  M^{-1}\big( \bar{B}[\W](\tau)+E^{p-1,(0)}_{L,\nablas}[\tilde{\W}](\tau)\big).
\end{align}
Here $\textup{Err}[\W,\mathsf{G}]$ is of the form $\nabla\W\cdot G$. The term ${Err}[\W,\mathsf{F}]$ is a quadratic form involving $\W$ and its first order derivatives. The reason we group them together is that they can be rewritten in terms of $S$, $r^{-1}\psi_Z$, $\hat{P}_{\even}$, $Q_{\even}$ and $\hat{W}_0$, which we already controlled. We record below the relations we need to do so.\\

By rewriting the definition of $S_W$, $Q_{\even}$ and $P_{\even}$, \eqref{def:SW}, \eqref{def:Q} and \eqref{def:P}, we have

\begin{align}
\nabla_r W_1+\done W_2+\frac{2}{r}W_1=&F_1, \label{X1r_X2o_substitution}\\
\donest W_1+\D\nablas_r W_2+\D\cdot\frac{1}{r}W_2=&F_2, \label{X1o_X2r_substitution}\\
\nabla_t W_1=&F_3.\label{X1t_substitution}
\end{align}
Here

\begin{equation}\label{def:F}
\begin{split}
F_1:=&-\frac{1}{2}S+\frac{1}{2} tr h^{\RW}+\D^{-1}\nabla_t W_0,\\
F_2:=&Q_{\even}-\frac{r}{2}\donest S+\frac{r}{2}\donest trh^{\RW}+\frac{r^3}{3}  \donest  \slashed{\Delta}_Z( r^{-1}\psi_Z),\\
F_3:=&\D\nabla_r W_0+\D\cdot\frac{1}{r}W_0+\D\cdot\frac{1}{r} P_{\even}.
\end{split}
\end{equation}

\subsection{$T$-current}\label{subsection:T}

In this subsection we construct a current $J_1$ which satisfies the following two properties. First, there exists a constant $C_T>0$ such that

\begin{align}\label{T_boundary}
\frac{1}{C_T} F^T[\W](\tau)\leq \int_{\Sigma_\tau} J_1\cdot n dvol_3 \leq C_T F^T[\W](\tau).
\end{align}
Second,

\begin{align}\label{T_bulk}
\Div J_1=\textup{Err}_1[\W,\mathsf{G}]+Err_{1,[2M,4R_T]}[\W]+Err_{1,[R_T,\infty)}[\W].
\end{align}
Here

\begin{align*}
\textup{Err}_{1}[\W,\mathsf{G}]:={\cred \nabla_t\W\cdot\mathsf{G}+\eta_T\cdot(9\nabla_t W_1\cdot G_1+\nabla_t W_2\cdot G_2)},
\end{align*}

\begin{align*}
Err_{1,[R_T,\infty)}[\W]:=\left(-\frac{1}{4}\eta_T'\right)\left( -9{\cpur |\nabla_L \hat{W}_1|^2}-{\cpur |\nabla_L \hat{W}_2|^2} \right)+{\cpur\chi_{[R_T,\infty)}(r)\cdot \frac{6M}{r^2}\nabla_t \hat{W}_2\donest \hat{W}_1},
\end{align*}

\begin{align*}
Err_{1,[2M,4R_T]}[\W,\mathsf{F}]:=& {\cblue\left[ -\chi_{[2M,R_T]}(r)\cdot \frac{6M}{r^2} + \eta_T(r)\cdot\left(\frac{16}{r}-\frac{54M}{r^2}\right)\right]\nabla_t \hat{W}_1\done \hat{W}_2}\\
+&{\cblue \eta_T\cdot \frac{18}{r^2}\left(1-\frac{4M}{r}\right)\nabla_t \hat{W}_1\cdot \hat{W}_1}.
\end{align*}
The cut-off function $\eta_T(r)$ and a large number $R_T$ will be defined below. Here $\textup{Err}_1[\W,\mathsf{G}]$ is of the form $\nabla_t\W\cdot \mathsf{G}$ and will be part of $\textup{Err}[\W,\mathsf{G}]$. The term $Err_{1,[R_T,\infty)}[\W]$ falls off rapidly in $r$ and can be absorbed by using $K^p_4$ in subsection \ref{subsec:rp}. The term $Err_{1,[2M,4R_T]}[\W,\mathsf{F}]$ comes from the feature that $A_{\W}$ is not self-adjoint. We will rely on \eqref{X1t_substitution} and \eqref{X1o_X2r_substitution} to estimate it.

Another difficulty in constructing $J_1$ is to make $J_1\cdot n$ positive definite. Let 

\begin{align*}
A_{\W,\textup{main}}:=\frac{1}{r^2}\left[\begin{array}{cc}
2 & 2r\done\\
2r\donest & 1
\end{array}\right]
\end{align*}
be the leading term of $A_{\W}$. We focus on a fixed mode $\ell\geq 1$ and $|m|\leq \ell$. By using $Y_{m\ell}$ and $r{\Lambda}^{-1/2} \nablas_A Y_{m\ell}$ as basis, we calculate that

\begin{align*}
|\nablas\W|^2+ W\cdot A_{\W,\textup{main}} \cdot W=_s \W\cdot
r^{-2}\left[ \begin{array}{cc}
\Lambda+2 & -2\sqrt{\Lambda}\\
-2\sqrt{\Lambda} & \Lambda
\end{array}\right]\cdot \W.
\end{align*} 

The worst case is as $\ell=1$ and $\Lambda=2$, the above matrix is only semi-positive definite with eigenvalues $\{0,6\}$. Hence we need to be careful about the lower order term even as $s$ is large.\\

Now we start the construction. We use $C_T$ to stand for a constant which may increase from line to line. Define two self-adjoint operators

\begin{equation}
A_{\W,\textup{ex}}:=\frac{1}{r^2}\left[ \begin{array}{cc}
2\left(1-\frac{4M}{r}\right) & 2r\left(1-\frac{3M}{r}\right)\done\\
2r\left(1-\frac{3M}{r}\right)\donest & \D
\end{array} \right],
\end{equation}

\begin{equation}
A_{\W,\textup{in}}:=\frac{1}{r^2}\left[ \begin{array}{cc}
2\left(1-\frac{4M}{r}\right) & 2r\done\\
2r\donest & \D
\end{array} \right].
\end{equation}
Let $R_T\geq 2R_{\textup{null}}$ be a large number. It will be fixed in subsection \ref{subsec:rp} and we remark that $C_T$ will \textbf{not} depend on $R_T$. Let

\begin{equation}
A_{\W,T}:=\left\{ \begin{array}{cc}
A_{\W,\textup{in}} & r\in [2M,R_T],\\
A_{\W,\textup{ex}} & r\in (R_T,\infty).
\end{array}\right.
\end{equation}
Consider 

\begin{equation}
(J_{1,1})_a:=\bigg(T_{ab}[\W]-\frac{1}{2}\W\cdot A_{\W,T}\cdot \W g_{ab}\bigg)T^b.
\end{equation}
Note $J_{1,1}$ is not continuous across $r=R_T$. But the it has no contribution as we apply divergence theorem \eqref{div_thm} since $T\cdot\nabla r=0$. The reason to use $A_{\W,\textup{ex}}$ for large $r$ is to make sure $J_{1,1}\cdot n$ is positive definite for $r\geq R_T$. To ensure the positivity in $r\in [2M,R_T]$ we need an auxiliary current. Let $\eta_T(r)$ be a cut-off function such that

\begin{equation*}
\begin{split}
\eta_T(r):=&\left\{ \begin{array}{cc}
1 & r\in [2M,R_T],\\
0 & r\in [4R_T,\infty),
\end{array} \right.\\
&|d\eta_T/dr|\leq r^{-1}. 
\end{split}
\end{equation*}
Define

\begin{align*}
A_{\W,\textup{aux}}:=\frac{1}{r^2}\left[ \begin{array}{cc}
0 & 2r\done\\
2r\donest & \D
\end{array} \right],
\end{align*}
and

\begin{align}
(J_{1,2})_a=\eta_T(r) \left( 9T_{ab}[W_1] +T_{ab}[W_2] -\frac{1}{2}\W\cdot A_{\W,\textup{aux}}\cdot\W g_{ab} \right)T^b.
\end{align}
Let $J_{1}:=J_{1,1}+J_{1,2}$. It's easy to check that

\begin{align*}
\int_{\Sigma_\tau} J_1\cdot n\ dvol_3\leq C_T F^T[\W](\tau),
\end{align*}
with the constant $C_T$ doesn't depend on $R_T$. We start to verify the positivity of $J_1\cdot n$. In $[R_T,\infty)$, as $n=L$,

\begin{align*}
J_1\cdot n=&\frac{1}{2}|\nabla_L\W|^2+\frac{9\eta_T(r)}{2}|\nabla_L W_1|^2+\frac{\eta_T(r)}{2}|\nabla_L W_2|^2\\
+& \frac{1}{2}\D\bigg( |\nablas \W|^2+{9\eta_T(r)}|\nablas W_1|^2+{\eta_T(r)}|\nablas W_2|^2\bigg).\\
 &+\frac{1}{2}\D \W\cdot( A_{\W,\textup{ex}}+\eta_T(r)A_{\W,\textup{aux}} )\cdot\W.                
\end{align*}
On a fixed mode $\ell$, the positivity of the sumand in the second and the third line (after integrated along $\mathbb{S}^2$) is equivalent to the positivity of the matrix

\begin{align*}
\left[\begin{array}{cc}
\Lambda+2-8s^{-1} & -2\sqrt{\Lambda}(1-3s^{-1})\\
-2\sqrt{\Lambda}(1-3s^{-1}) & \Lambda-2s^{-1}
\end{array}
\right]
+\eta_T
\left[\begin{array}{cc}
9\Lambda & -2\sqrt{\Lambda}\\
-2\sqrt{\Lambda} & \Lambda-2s^{-1}
\end{array}
\right].
\end{align*}
Here $\Lambda=\Lambda(\ell)$ is defined in \eqref{def:lambda_ell}. The first matrix has eigenvalues $\{\Lambda+1-5s^{-1}\pm(1-3s^{-1}) \sqrt{4\Lambda+1}\}$, which is positive for $\ell\geq 1$ ($\Lambda\geq 2$) and $s\geq 10$. Moreover, for $\ell\geq 3$ ($\Lambda\geq 6$), the two eigenvalues are both comparable to $\Lambda$. Together with the positivity of the second matrix, we finish the verification for $r\in [R_T,\infty)$. Similarly, the positivity of $J_1\cdot n$ in $r\in [2M,R_T]$ is equivalent to the positivity of the matrix

\begin{align*}
&\left[\begin{array}{cc}
\Lambda+2-8s^{-1}  & -2\sqrt{\Lambda}\\
-2\sqrt{\Lambda} & \Lambda-2s^{-1}
\end{array}
\right]+\left[\begin{array}{cc}
9\Lambda & -2\sqrt{\Lambda}\\
-2\sqrt{\Lambda} & \Lambda-2s^{-1}
\end{array}
\right]\\
=&\left[\begin{array}{cc}
10\Lambda+2-8s^{-1} & -4\sqrt{\Lambda}\\
-4\sqrt{\Lambda} & 2\Lambda-4s^{-1}
\end{array}
\right]\geq \left[\begin{array}{cc}
10\Lambda-2 & -4\sqrt{\Lambda}\\
-4\sqrt{\Lambda} & 2\Lambda-2
\end{array}
\right].
\end{align*}
In the last inequality we used $s\geq 2$. The last matrix has eigenvalues $6\Lambda-2\pm 4	\sqrt{\Lambda (\Lambda+1)}$, which are positive and comparable to $\Lambda$ for $\ell\geq 1$ and $\Lambda\geq 2$. We now have checked that for $\ell\geq 1$, 

\begin{align*}
\int_{\Sigma_\tau} J_1\cdot n dvol_3\geq &\frac{1}{C_T}\int_{\Sigma_\tau''} (1-2s^{-1})|\nabla_{\underline{L}'} \W|^2+|\nabla_L \W|^2+|\nablas \W|^2+M^{-2}s^{-2}|\W|^2 dvol_3\\
+&\frac{1}{C_T}\int_{\Sigma_\tau'} |\nabla_L \W|^2+|\nablas \W_{\ell\geq 2}|^2 dvol_3.
\end{align*}
By the assumption \eqref{assumption} and Corollary \ref{cor:F_base}, we can add zeroth order term along $\Sigma_{\tau}'$ and finish the proof of \eqref{T_boundary}.\\

We proceed to compute the divergence of $J_1$ and to verify \eqref{T_bulk}. From \eqref{equ:W12}, we calculate

\begin{equation}
\Div J_{1,1}=_s\left\{\begin{array}{cc}
\nabla_t\W\cdot\mathsf{G}-6M/r^2\cdot\nabla_t \hat{W}_1 \done \hat{W}_2 & r\in [2M,R_T],\\
\nabla_t\W\cdot\mathsf{G}+6M/r^2\cdot\nabla_t \hat{W}_2 \donest \hat{W}_1 & r\in (R_T,\infty).
\end{array}\right.
\end{equation}
and

\begin{align*}
\Div J_{1,2}=&\left(-\frac{1}{4}\eta_T'\right)\left( 9(|\nabla_{\underline{L}}\hat{W}_1|^2-|\nabla_L \hat{W}_1|^2)+(|\nabla_{\underline{L}}\hat{W}_2|^2-|\nabla_L \hat{W}_2|^2) \right)\\
            +&\eta_T\cdot\left(9\nabla_t \hat{W}_1\cdot G_1+\nabla_t \hat{W}_2\cdot G_2\right)\\
            +&\eta_T\cdot\left(\frac{16}{r}-\frac{54M}{r^2}\right)\nabla_t \hat{W}_1\cdot\done \hat{W}_2+\eta_T\cdot \frac{18}{r^2}\left(1-\frac{4M}{r}\right)\nabla_t \hat{W}_1\cdot \hat{W}_1.
\end{align*}
Therefore in $r\in[2M,R_T),$

\begin{align*}
\Div J_1=&{\cblue \left[-\frac{6M}{r^2}+\left(\frac{16}{r}-\frac{54M}{r^2}\right)\cdot\eta_T\right]\nabla_t \hat{W}_1\done \hat{W}_2}+{\cblue \eta_T\cdot \frac{18}{r^2}\left(1-\frac{4M}{r}\right)\nabla_t \hat{W}_1\cdot \hat{W}_1}\\
        +&{\cred \nabla_t\W\cdot\mathsf{G}+\eta_T\cdot(9\nabla_t \hat{W}_1\cdot G_1+\nabla_t \hat{W}_2\cdot G_2)}.        
\end{align*}
In $(R_T,\infty)$,

\begin{align*}
\Div J_1\geq &\left(-\frac{1}{4}\eta_T'\right)\left( -9{\cpur |\nabla_L \hat{W}_1|^2}-{\cpur |\nabla_L \hat{W}_2|^2} \right)+{\cpur \frac{6M}{r^2}\nabla_t \hat{W}_2\donest \hat{W}_1}\\
        +&{\cblue \left(\frac{16}{r}-\frac{54M}{r^2}\right)\cdot\eta_T\nabla_t \hat{W}_1\done \hat{W}_2}+{\cblue \eta_T\cdot \frac{18}{r^2}\left(1-\frac{4M}{r}\right)\nabla_t \hat{W}_1\cdot \hat{W}_1}\\
        +&{\cred \nabla_t\W\cdot\mathsf{G}+\eta_T\cdot(9\nabla_t \hat{W}_1\cdot G_1+\nabla_t \hat{W}_2\cdot G_2)}.        
\end{align*}
Then \eqref{T_bulk} follows by collecting terms.

\subsection{red-shift current}\label{subsection:red-shift}

In this subsection, we construct the red-shift current $J_2$ which satisfies the following two properties. First, there exists a constant $C_{rs}>0$ such that for any $0<\epsilon\leq 1$,

\begin{align}\label{redshift_boundary}
\frac{\epsilon}{C_{rs}}F[\W](\tau) \leq \epsilon\int_{\Sigma_\tau} J_2\cdot n\ dvol_3+F^T[\W](\tau) \leq C_{rs} F[\W](\tau).
\end{align}
Second, with the same constant $C_{rs}$, we have

\begin{align}\label{redshift_bulk}
\Div J_2= K_2 +\textup{Err}_2[\W,\mathsf{G}]+Err_{2,[r_{rs,\W},r_{rs}^+]}[\W],
\end{align}
Here

\begin{align*}
K_2&\geq \frac{1}{C_{rs}} M^{-3}|(M\partial)^{\leq 1}\W|^2\cdot\chi_{[2M,r_{rs,\W}]},\\
\textup{Err}_2[\W,\mathsf{G}]&=\nabla_Y\W\cdot \mathsf{G},\\
\big|Err_{2,[r_{rs,\W},r_{rs}^+]}[\W]\big| &\leq C_{rs}M^{-3}|(M\partial)^{\leq 1}\W|^2\cdot\chi_{[r_{rs,\W},r_{rs}^+]}.
\end{align*}
The vector field $Y$ and $r_{rs,\W}\in (2M,r_{rs}^+)$ will be defined below. The term $\textup{Err}_2[\W,\mathsf{G}]$ will be part of $\textup{Err}[\W,\mathsf{G}]$. The term $Err_{2,[r_{rs,\W},r_{rs}^+]}[\W]$ will be absorbed into $K_3$ in subsection \ref{subsection:Morawetz}.\\

We start the construction of $J_2$. We use $C_{rs}$ to stand for a constant that may increase from line to line. Let $\sigma>0$ be a number to be determined and $Y(\sigma)$ be the red-shift vector define in \eqref{vector_rs}. Consider

\begin{align*}
J_{2,a}=\left( T_{ab}[\W]-\frac{1}{r^2}|\W|^2 g_{ab}\right) Y^b.
\end{align*}
Since $Y(\sigma)\big|_{r=2M}=2\underline{L}'$, we deduce \eqref{redshift_boundary} with the constant $C_{rs}$ depending on $\sigma$. \\

We now compute $\Div J_2$. From \eqref{g_redshift}, on the horizon $r=2M$ we have

\begin{align*}
\Div J_{2}\bigg|_{r=2M}=&\frac{\sigma}{2}|\nabla_v\W|^2+\frac{1}{2M}|\nabla_R\W|^2+\frac{2}{M}\nabla_R\W\cdot\nabla_v\W+\frac{\sigma}{2}|\slashed{\nabla}\W|^2+\frac{\sigma}{8M^2}|\W|^2\\
+&\nabla_{Y}\W\cdot (\Box\W-\frac{1}{r^2}\W).
\end{align*}
By choosing $\sigma$ large enough depending on $A_{\W}$, one can make

\begin{align*}
\Div J_{2}\bigg|_{r=2M}- \nabla_Y\W\cdot \mathsf{G}\geq \frac{1}{C_{rs}} M^{-3}|(M\partial)^{\leq 1}\W|^2.
\end{align*}
Through continuity, there exists $r_{rs,\W}^-\in (2M,r_{rs}^+)$ such that

\begin{equation*}
\begin{split}
\Div J_2-{\cred \nabla_Y\W\cdot \mathsf{G}}  \geq &\frac{1}{C_{rs}} \left( M^{-1}|\partial\W|^2+M^{-3}|\W|^2\right)\ \ \textup{in}\ [2M,r_{rs,\W}^-],\\
\big| \Div J_2-{\cred \nabla_Y\W\cdot \mathsf{G}}\big|  \leq& C_{rs} \left( M^{-1}|\partial\W|^2+M^{-3}|\W|^2\right)\ \ \textup{in}\ [r_{rs,\W}^-,r_{rs}^+],
\end{split}
\end{equation*}
which yields \eqref{redshift_bulk}.

\subsection{Morawetz current}\label{subsection:Morawetz}

In this subsection, we construct a current $J_3$ which satisfies the following two properties. First, there exists a constant $C_{Mor}$ such that

\begin{equation}\label{Morawetz_boundary}
\left|\int_{\Sigma_\tau} J_3\cdot n\ dvol_3\right|\leq C_{Mor}F^T[\W]. 
\end{equation}
Second, with the same constant $C_{Mor}$,

\begin{equation}\label{Morawetz_bulk}
\Div J_3=K_3+\textup{Err}_3[\W,\mathsf{G}]+Err_{3,[2M,bM]}[\W,\mathsf{F}]+Err_{3,[bM,\infty)}[\W].
\end{equation}
Here

\begin{align*}
K_3\geq_s &\frac{1}{C_{Mor}} M^{-1}s^{-2}\bigg( (1-2s^{-2})^2|\nabla_r\W|^2+|\nabla_t \W|^2+|\nablas \W|^2+ M^{-2}s^{-2}|\W|^2 \bigg),\\
\textup{Err}_3[\W,\mathsf{G}]:=&\left(\nabla_X \W+\left(\frac{1}{2}\omega-2\bar{f}\right)\W\right)\cdot\mathsf{G},\\
Err_{3,[bM,\infty)}[\W]:= &\left(-\frac{2M}{r^2}\nabla_X\hat{W}_1\cdot\done \hat{W}_2+\frac{4M}{r^2}\nabla_X\hat{W}_2\cdot\donest \hat{W}_1\right)\cdot\chi_{[bM,\infty)}.
\end{align*}
The number $b$ is the same constant we used to subsection \ref{subsec:to_RW} and its value will be determined in subsection \ref{subsec:estimate_W12}. The functions $f(r),\bar{f}(r),\omega(r)$ and the vector field $X$ will be defined below. The term $\textup{Err}_3[\W,\mathsf{G}]$ will be part of $\textup{Err}[\W,\mathsf{G}]$. The term $Err_{3,[bM,\infty)}[\W]$ will be treated as perturbation and will be absorbed into $K^p_4$ in subsection \ref{subsec:rp}. The term $Err_{3,[2M,bM]}[\W,\mathsf{F}]$ is more involved and will be defined later in this subsection. Compared to the scalar case, $K_3$ doesn't degenerate at photon sphere because of \eqref{X1o_X2r_substitution} and \eqref{X1r_X2o_substitution}.\\

We start the construction of $J_3$. We use $C_{Mor}$ to stand for a constant that may increase from line to line. Let


\begin{equation*}
A_{\W, Mor}:=\frac{1}{r^2}\left[ \begin{array}{cc}
2-8s^{-1}  & (2r-4M)\done\\
(2r-4M)\donest & 1-2s^{-1}
\end{array} \right],
\end{equation*}
and

\begin{align*}
f(r):=&\left(1-\frac{3M}{r}\right)\left(1+\frac{2M}{r}+\frac{2M^2}{r^2}\right),\\
\omega(r)=&\D\left(\frac{2f}{r}+\frac{d f}{dr}\right)\\
X:=&f(r)\D\frac{\partial}{\partial r}.
\end{align*}
Note that $\frac{df}{dr}=\frac{M}{r^2}(1+8s^{-1}+18s^{-2})\geq \frac{M}{r^2}>0$. We consider the current

\begin{align}
(J_{3,1})_a=&T_{ab}[\W]X^b-\frac{1}{4}\nabla_a\omega |\W|^2+\frac{1}{4}\omega \nabla_a |\W|^2-\frac{1}{2}(\W\cdot A_{\W, Mor}\cdot \W) X_a.
\end{align}
Through direct computation,

\begin{align*}
\Div J_{3,1}=_s&\left(1-\frac{2M}{r}\right)^2\frac{d f}{d r}|\nabla_{r}\W|^2+\frac{f}{r}\left(1-\frac{3M}{r}\right)|\slashed{\nabla}\W|^2\\
&+\W\cdot \left(-\frac{1}{4}\Box\omega -\frac{1}{2}[X, A_{\W, Mor}]-\frac{M}{r^2}f A_{\W, Mor} \right)\cdot \W\\
&-\frac{2M}{r^2}\nabla_X\hat{W}_1\cdot\done \hat{W}_2+\frac{4M}{r^2}\nabla_X\hat{W}_2\cdot\donest \hat{W}_1+\frac{M\omega}{r^2}\hat{W}_1\cdot\donest \hat{W}_2\\
&+{\cred\left(\nabla_X \W+\frac{1}{2}\omega \W\right)\cdot\mathsf{G}}.
\end{align*}
Here we used

\begin{align*}
A_{\W}-A_{\W, Mor}=\frac{1}{r^2} \left[
\begin{array}{cc}
0 & -2M\done\\
4M\donest & 0
\end{array}
\right].
\end{align*}
The last line comes from the source term $\mathsf{G}$ and will be part of $\textup{Err}_3[\W,\mathsf{G}]$. Let

\begin{align*}
Err_{3,[bM,\infty)}[\W]:=\left(-\frac{2M}{r^2}\nabla_X\hat{W}_1\cdot\done \hat{W}_2+\frac{4M}{r^2}\nabla_X\hat{W}_2\cdot\donest \hat{W}_1\right)\cdot\chi_{[bM,\infty)}.
\end{align*}
We will take $b$ large enough and show that the remaining terms are positive definite. To simplify the computation, we focus on a fixed mode $\ell\geq 1$, $|m|\leq \ell$. Using the basis $Y_{m\ell}$, $r\Lambda^{-1/2}\nablas Y_{m\ell}$ with $\Lambda$ given by \eqref{def:lambda_ell}, we calculate

\begin{align*}
\W\cdot \mathbb{M}_{[bM,\infty)}  \cdot \W =_s&\frac{f}{r}\left(1-\frac{3M}{r}\right)|\slashed{\nabla}\W|^2+\frac{M\omega}{r^2}\hat{W}_1\cdot\donest \hat{W}_2\\
+&\W\cdot \left(-\frac{1}{4}\Box\omega -\frac{1}{2}[X, A_{\W, Mor}]-\frac{M}{r^2}f A_{\W, Mor} \right)\cdot \W,
\end{align*}  
where $\mathbb{M}_{[bM,\infty)}$ is the matrix

\begin{align*}
\mathbb{M}_{[bM,\infty)}:=& \frac{f(1-3s^{-1})}{r^3}\left[ 
\begin{array}{cc}
\Lambda & 0\\
0 & \Lambda-1
\end{array}
\right]+\frac{M\omega}{2r^3}\left[ 
\begin{array}{cc}
0 & -\sqrt{\Lambda}\\
-\sqrt{\Lambda} & 0
\end{array}
\right]\\
+&\left( -\frac{1}{4} \Box \omega\right) \cdot \left[ 
\begin{array}{cc}
1 & 0\\
0 & 1
\end{array}\right]+\frac{f(1-2s^{-1})}{r^3}\left[ 
\begin{array}{cc}
2-12s^{-1} & (-2+6s^{-1})\sqrt{\Lambda}\\
(-2+6s^{-1})\sqrt{\Lambda} & 1-3s^{-1}
\end{array}\right] \\
-&\frac{f}{sr^3}\left[ 
\begin{array}{cc}
2-8s^{-1} & (-2+4s^{-1})\sqrt{\Lambda}\\
(-2+4s^{-1})\sqrt{\Lambda} & 1-2s^{-1}
\end{array}
\right]
\end{align*}
We claim that $\mathbb{M}_{[bM,\infty)}$ is positive definite for $\ell\geq 1$ as $s$ is large enough. The asymptotic of $\mathbb{M}_{[bM,\infty)}$ reads, 

\begin{align*}
r^3 \mathbb{M}_{[bM,\infty)}=\left[\begin{array}{cc}
(\Lambda+2)-(4\Lambda+33/2)s^{-1} & (-2+13s^{-1})\sqrt{\Lambda}\\
(-2+13 s^{-1})\sqrt{\Lambda} & \Lambda-(4\Lambda-1/2)s^{-1}
\end{array}
\right]+\left[\begin{array}{cc}
O(s^{-2})\Lambda & O(s^{-2})\sqrt{\Lambda}\\
O(s^{-2})\sqrt{\Lambda} & O(s^{-2})\Lambda
\end{array}
\right].
\end{align*}
Here the constant in $O(s^{-2})$ term doesn't depend on $\ell$. The determinant has the expansion

\begin{align*}
\textup{det} \big(r^3 \mathbb{M}_{[bM,\infty)}\big)=\Lambda(\Lambda-2)-(8\Lambda^2-28\Lambda-1)s^{-1}+O(\Lambda^2s^{-2}).
\end{align*}
For $\ell\geq 2$ and $\Lambda\geq 6$, the above determinant is positive and comparable to $\Lambda^2$ for $s$ large enough. For $\ell=1$ and $\Lambda=2$, the above determinant has the asymptotic $25s^{-1}+O(s^{-2})$, which is also positive for $s$ large enough. Therefore, there exists a constant $s_{Mor}$ such that for all $s\geq s_{Mor}$ and $\ell\geq 1$,

\begin{align*}
\W\cdot \mathbb{M}_{[bM,\infty)}  \cdot \W\ \gtrsim  M^{-3}s^{-4}\Lambda|\W|^2 \approx_s  M^{-1}s^{-2}|\nablas\W|^2+ M^{-3}s^{-4}|\W|^2.
\end{align*}
We from now on require $b\geq s_{Mor}$ and obtain that up to $Err_{3,[bM,\infty)}[\W]$ and $\textup{Err}_3[\W,\mathsf{G}]$, $\Div J_{3,1}$ is positive definite in $[bM,\infty)$.\\
 
We turn to the region $[2M,bM]$. Since $r\leq bM$, we can use \eqref{X1r_X2o_substitution} to replace $\nabla_r\hat{W}_1=\nabla_r {W}_1$ by $-\done W_2-\frac{2}{r}W_1+F_1$ as

\begin{align*}
&\D^{2}\frac{d f}{d r}|\nabla_r W_1|^2 -\frac{2M}{r^2}\nabla_X W_1\cdot \done W_2\\
=&\D^{2}\frac{d f}{d r}\left| \done W_2+\frac{2}{r}W_1 \right|^2+\D\frac{2M}{r^2} f\left(|\done W_2|^2+ \frac{2}{r}W_1\cdot \done W_2 \right)\\
+&{\cblue \D^2\frac{d f}{d r}\left(|F_1|^2-2\left(\done W_2+\frac{2}{r}W_1\right)\cdot F_1\right)} -{\cblue \D\frac{2M}{r^2} f F_1\cdot\done W_2}.
\end{align*}
We denote the last line by $Err'_{3,[2M,bM]}[\W,\mathsf{F}]$. Similarly, using the substitution \eqref{X1o_X2r_substitution} to replace $\D\nabla_r W_2$ by $-\donest W_1-\D\frac{1}{r}W_2+F_2$, we obtain

\begin{align*}
&\left(1-\frac{2M}{r}\right)^2\frac{d f}{d r}|\nabla_{r} W_2|^2+\frac{4M}{r^2}\nabla_X {W}_2 \cdot \donest \hat{W}_1 \\
=&\frac{d f}{d r}\left| \donest W_1+\D\frac{1}{r}W_2 \right|^2- \frac{4M}{r^2}f \left( |\donest W_1|^2+\D\frac{1}{r}W_2\cdot\donest W_1 \right)\\
+&{\cblue \frac{d f}{d r}\left(|F_2|^2-2\left(\donest W_1+\D\frac{1}{r}W_2\right)\cdot F_2\right)+\frac{4M}{r^2}f F_2\cdot \donest W_1}.
\end{align*}
We denote the last line by $Err''_{3,[2M,bM]}[\W,\mathsf{F}]$. Then up to $Err'_{3,[3M,bM]}[\W,\mathsf{F}]$, $Err''_{3,[3M,bM]}[\W,\mathsf{F}]$ and $\textup{Err}_{3}[\W,\mathsf{G}]$, $\Div J_{3,1}$ becomes a quadratic from involving $W_1,W_2,\donest W_1$ and $\done W_2$:

\begin{align*}
&\frac{f}{r}\left(1-\frac{3M}{r}\right)|\slashed{\nabla}\W|^2+\frac{M\omega}{r^2}\hat{W}_1\cdot\donest \hat{W}_2+\W\cdot \left(-\frac{1}{4}\Box\omega -\frac{1}{2}[X, A_{\W, Mor}]-\frac{M}{r^2}f A_{\W, Mor} \right)\cdot \W\\
+&\D^{2}\frac{d f}{d r}\left| \done W_2+\frac{2}{r}W_1 \right|^2+\D\frac{2M}{r^2} f\left(|\done W_2|^2+ \frac{2}{r}W_1\cdot \done W_2 \right)\\
+&\frac{d f}{d r}\left| \donest W_1+\D\frac{1}{r}W_2 \right|^2- \frac{4M}{r^2}f \left( |\donest W_1|^2+\D\frac{1}{r}W_2\cdot\donest W_1 \right).
\end{align*}
On a fixed mode $\ell\geq 1$ and $|m|\leq \ell$, after integrated along $\mathbb{S}^2$, it equals $
\W\cdot\mathbb{M}_{[3M,bM]} \cdot\W$, where $\mathbb{M}_{[3M,bM]}$ is a two by two matrix

\begin{align*}
\mathbb{M}_{[3M,bM]}=&\mathbb{M}_{[bM,\infty)}+\frac{df}{dr}\frac{1}{r^2} \left[ 
\begin{array}{cc}
\Lambda & -(1-2s^{-1})\sqrt{\Lambda}\\
-(1-2s^{-1})\sqrt{\Lambda} & (1-2s^{-1})^2 
\end{array} \right]\\
+&(1-2s^{-1})^2\frac{df}{dr} \frac{1}{r^2}\left[ 
\begin{array}{cc}
4 & -2\sqrt{\Lambda}\\
-2\sqrt{\Lambda} & \Lambda 
\end{array} \right]+ \frac{2M}{r^4}f\left[ 
\begin{array}{cc}
-2\Lambda & 0\\
0 & (1-2s^{-1})\Lambda 
\end{array} \right].
\end{align*}
Through computation, $\textup{det}[r^3\mathbb{M}_{[2M,bM]}]=p_1(s)\Lambda^2+p_2(s)\Lambda+p_3(s)$, where

\begin{align*}
p_1(s)=&1-8s^{-1}+15s^{-2}+42s^{-3}+7s^{-4}-28s^{-5}-1316s^{-6}+3084s^{-8}+4032s^{-9},\\
p_2(s)=&-2+27s^{-1}-183s^{-2}+561s^{-3}-\frac{1573}{4}s^{-4}-2845s^{-5}+6834s^{-6}+3564s^{-7}\\
       &-14337s^{-8}-2412s^{-9}-6948s^{-10},\\
p_3(s)=&3s^{-1}-\frac{59}{4}s^{-2}+\frac{25}{2}s^{-3}+354s^{-4}-1752s^{-5}+4095s^{-6}-5832s^{-7}\\
&-432s^{-8}+6156s^{-9}+6804s^{-10}.
\end{align*}
We check that for $s\geq 2$, $4p_1(s)+2p_2(s)+p_3>0$, hence $\mathbb{M}_{[2M,bM]}$ is positive definite for $\ell=1$. We further check that for $s\geq 2$, $p_1(s)$, $4p_1(s)+p_2(s)$ and $4p_1(s)+p_3(s)$ are positive, hence $\mathbb{M}_{[2M,bM]}$ is positive definite for $\ell\geq 2$ as well. As $\mathbb{M}_{[2M,bM]}$ is positive definite, we obtain for all $r\geq 2M$,

\begin{align*}
\Div J_{3,1} \geq & \frac{1}{C_{Mor}} M^{-1} \bigg( s^{-2} |\nabla_r\W|^2+ (1-3s^{-1})^2 s^{-2}|\nablas\W|^2+M^{-2}s^{-4}|\W|^2 \bigg)\\
                          +&\textup{Err}_3[\W,\mathsf{G}]+Err_{3,[bM,\infty)}[\W]+Err'_{3,[2M,bM)}[\W,\mathsf{F}]+Err''_{3,[2M,bM)}[\W,\mathsf{F}].
\end{align*}
Moreover, from the view of \eqref{X1o_X2r_substitution} and \eqref{X1r_X2o_substitution}, we can remove the degeneracy of $|\nablas\W|^2$ at $r=3M$ with the help of $F_1$ and $F_2$. Thus we manage to show that for all $r\geq 2M$,

\begin{align*}
\Div J_{3,1} \geq & \frac{1}{C_{Mor}} M^{-1} \bigg( s^{-2} (1-2s^{-1})^2 |\nabla_r\W|^2+ s^{-2}|\nablas\W|^2+M^{-2}s^{-4}|\W|^2 \bigg)\\
                          +&\textup{Err}_3[\W,\mathsf{G}]+Err_{3,[bM,\infty)}[\W]+Err_{3,[2M,bM)}[\W,\mathsf{F}].
\end{align*}
with

\begin{align*}
Err_{3,[2M,bM)}[\W,\mathsf{F}]:=Err'_{3,[2M,bM)}[\W,\mathsf{F}]+Err''_{3,[2M,bM)}[\W,\mathsf{F}]+C_{Mor} M^{-1}(|F_1|^2+|F_2|^2)\chi_{[2M,bM)}.
\end{align*}

We proceed to add $|\nabla_t\W|^2$ into the bulk. Let $\bar{f}(r)\geq  0$ be a non-negative function in $r$. We consider

\begin{align}
(J_{3,2})_a=-\bar{f}\nabla_a|\W|^2+\nabla_a \bar{f}|\W|^2.
\end{align}
We compute

\begin{align*}
\Div J_{3,2}=&-2\bar{f} \nabla^a\W\cdot\nabla_a \W-2\bar{f} \W\cdot\mathsf{G}+\frac{1}{4}\Box \bar{f}\cdot |\W|^2-2\bar{f} \W\cdot A_{\W}\cdot W\\
=&2\bar{f}\big( (1-2s^{-1})^{-1}|\nabla_t\W|^2-(1-2s^{-1})|\nabla_r\W|^2-|\nablas\W|^2  \big)\\
-&2\bar{f} \W\cdot\mathsf{G}+\frac{1}{4}\Box \bar{f}\cdot |\W|^2-2\bar{f} \W\cdot A_{\W}\cdot W
\end{align*}
By taking 

\begin{align*}
\bar{f}(r)=\epsilon_{3,2}\frac{M}{r^2}\D,
\end{align*}
we can add $|\nabla_t\W|^2$ into the divergence by defining 

\begin{align}
J_3:=J_{3,1}+J_{3,2}.
\end{align}
Now \eqref{Morawetz_bulk} follows from the above discussion. The requirement \eqref{Morawetz_boundary} can be verified by the asymptotics of $f(r),\omega(r),\bar{f}(r)$ and $A_{\W, Mor}$. We finish this subsection by recording $Err_{3,[2M,bM)}[\W,\mathsf{F}]$, which is supported in $r\in [2M,bM)$.

\begin{equation}\label{def:err3}
\begin{split}
Err_{3,[2M,bM)}[\W,\mathsf{F}]=\bigg\{ &{\cblue \D^2\frac{d f}{d r}\left(|F_1|^2-2\left(\done W_2+\frac{2}{r}W_1\right)\cdot F_1\right)} -{\cblue \D\frac{2M}{r^2} f F_1\cdot\done W_2}\\
+&{\cblue \frac{d f}{d r}\left(|F_2|^2-2\left(\donest W_1+\D\frac{1}{r}W_2\right)\cdot F_2\right)+\frac{4M}{r^2}f F_2\cdot \donest W_1}\\
+&C_{Mor} M^{-1}(|F_1|^2+|F_2|^2)\bigg\}\cdot\chi_{[2M,bM)}.
\end{split}
\end{equation}

\subsection{$r^p$-current}\label{subsec:rp}

In this subsection, we construct a current $J^p_4$ for $p\in [\delta,2-\delta]$, which satisfies the following two properties. First, for any $p\in [\delta,2-\delta]$

\begin{align}\label{rp_boundary}
E^{p,(0)}_L[\tilde{\W}](\tau)\leq  \int_{\Sigma_\tau} J_4\cdot n\ dvol_3 \leq  E^{p,(0)}_L[\tilde{\W}](\tau)+F[\W](\tau).
\end{align}
Here $\tilde{\W}=\eta_{\textup{null}}\W$ and $\eta_{\textup{null}}$ is the cut-off function defined in \eqref{cutoff_rp}. Second, there exists a constant $C_{rp}>0$ such that for any $\tau_2\geq \tau_1$ and $p\in [\delta,2-\delta]$,

\begin{align}\label{rp_bulk}
\int_{D(\tau_1,\tau_2)} \Div J^p_4\ dvol = &\int_{D(\tau_1,\tau_2)} K^p_4+ \textup{Err}_4[\W,\mathsf{G}]+Err_{4,[R_{\textup{null}},R_{\textup{null}}+M]} \ dvol.
\end{align}
Here

\begin{align*}
K^p_4\geq &\frac{1}{C_{rp}} M^{-1}\bigg(s^{p-3}|\nabla_L (s\tilde{\W})|^2+s^{p-3}|\nablas (s\tilde{\W})|^2+s^{-1-\delta}|\nabla_{\underline{L}} \tilde{\W}|^2+M^{-2}s^{p-5}|s\tilde{\W}|^2\bigg)\\
     \geq &\frac{ 1}{4C_{rp}} M^{-1}\bigg(s^{p-1}|\nabla_L \tilde{\W}|^2+s^{p-1}|\nablas \tilde{\W}|^2+s^{-1-\delta}|\nabla_{\underline{L}} \tilde{\W}|^2+M^{-2}s^{p-3}|\tilde{\W}|^2\bigg),\\
\textup{Err}^p_4[\W,\mathsf{G}]=&{\cred \left(s^{p-1}\eta_{\textup{null}} \D^{-1}\nabla_L(s\tilde{\W})+\epsilon_{4,2} \eta_{\textup{null}}  s^{-\delta}\nabla_t \tilde{\W} \right)\cdot \mathsf{G}},\\
Err^p_{4,[R_{\textup{null}},R_{\textup{null}}+M]}=&{\cpur \left(s^{p-1}\D^{-1}\nabla_L(s \tilde{\W})+\epsilon_{4,2}  s^{-\delta}\nabla_t\tilde{\W}  \right)\cdot (\Box\eta_{\textup{null}} \cdot\W+2\nabla\eta_{\textup{null}} \cdot\nabla\W)}.\\
\end{align*}
The term $\textup{Err}^p_4[\W,\mathsf{G}]$ will be part of $\textup{Err}^p[\W,\mathsf{G}]$. The term $Err^p_{4,[R_{\textup{null}},R_{\textup{null}}+M]}$ comes from the derivative of the cut-off function $\eta_{\textup{null}}$ and will be absorbed into $K_3$.\\

Now we start the construction. We use $C_{rp}$ to stand for a constant which may increase from line to line. Recall that 

\begin{align*}
A_{\W,\textup{main}}=\frac{1}{r^2}\left[ \begin{array}{cc}
2 & 2r\done\\
2\donest & 1
\end{array} \right]
\end{align*}
is the leading part of $A_{\W}$. Let $A_{\W,\textup{sub}}=s(A_{\W}-A_{\W,\textup{main}})$ be the subleading part and $\tilde{\W}=\eta_{\textup{null}}(r)\W$ with $\eta_{\textup{null}}$ defined in \eqref{cutoff_rp}. Define

\begin{align}
(J^p_{4,1})_a:=s^{p-2}\D^{-1} \left( T_{ab}[ s\tilde{\W} ]L^b- \frac{1}{2}(s\tilde{\W})\cdot (A_{\W,\textup{main}}+2M/r^3) \cdot (s\tilde{\W})|^2 L_a\right).
\end{align}
We compute

\begin{align*}
\Div J^p_{4,1}=_s&M^{-1}s^{p-3}\left(\frac{p}{2}\D -\frac{M}{r}\right)\D^{-2} |\nabla_L(s\tilde{\W})|^2+M^{-1}s^{p-3}\left(1-\frac{p}{2}\right)|\nablas (s\tilde{\W})|^2\\
- &M^{-1}\frac{1}{2}s^{p-3} (s\tilde{\W})\cdot \left(r\left[\partial_r, (A_{\W,\textup{main}}+2M/r^3)\right]+p((A_{\W,\textup{main}}+2M/r^3)) \right)\cdot (s\tilde{\W})\\
+&s^{p-1}\D^{-1}\nabla_L(s\tilde{\W}) \cdot (\Box \tilde{\W}-A_{\W,\textup{main}}\tilde{\W}).
\end{align*}
Since 
 
\begin{align*}
|\nablas \W|^2+\W\cdot A_{\W,\textup{main}}\cdot \W
\end{align*}
is at least non-negative for $\ell\geq 1$ and is comparable to $|\nablas\W|^2$ for $\ell\geq 2$, there exists a constant $C_{rp}$ such that

\begin{align*}
&M^{-1}s^{p-3}\left(\frac{p}{2}\D -\frac{M}{r}\right)\D^{-2} |\nabla_L(s\tilde{\W})|^2+M^{-1}s^{p-3}\left(1-\frac{p}{2}\right)|\nablas (s\tilde{\W})|^2\\
- &M^{-1}\frac{1}{2}s^{p-3} (s\tilde{\W})\cdot \left(r\left[\nabla_r, (A_{\W,\textup{main}}+2M/r^3)\right]+p((A_{\W,\textup{main}}+2M/r^3)) \right)\cdot (s\tilde{\W})\\
\geq_s &\frac{1}{C_{rp}} M^{-1}s^{p-3}\left( |\nabla_L (s\tilde{\W})|^2+|\nablas (s\tilde{\W}_{\ell\geq 2})|^2+M^{-2}|\tilde{\W}|^2\right).
\end{align*}
By Corollary \ref{cor:E^p_base}, we can recover $s^{p-3}|\nablas (s\tilde{\W}_{\ell= 1})|^2\approx_s s^{p-3} M^{-2}|\tilde{\W}_{\ell=1}|^2 $ through Hardy inequality. Therefore, after integrated in $D(\tau_1,\tau_2)$, the above is bounded from below as

\begin{align*}
\frac{1}{C_{rp}} M^{-1}s^{p-3}\left( |\nabla_L (s\tilde{\W})|^2+|\nablas (s\tilde{\W})|^2+M^{-2}|\tilde{\W}|^2\right),
\end{align*}
with a larger constant $C_{rp}$ depending on $\delta$. For the remaining terms we compute 

\begin{align*}
&s^{p-1}\D^{-1}\nabla_L(s\tilde{\W}) \cdot (\Box \tilde{\W}-A_{\W,\textup{main}}\tilde{\W})\\
=&s^{p-1}\D^{-1}\nabla_L(s\tilde{\W}) \cdot ({\cpur s^{-1}A_{\W,\textup{sub}}\tilde{\W}}+{\cred\eta_{\textup{null}} \mathsf{G}})\\
+&s^{p-1}\D^{-1}\nabla_L(s\tilde{\W}) \cdot ({\cpur \Box\eta_{\textup{null}}(r)\cdot \W+2\nabla\eta_{\textup{null}}(r)\cdot\W}).
\end{align*}
The contribution of $s^{p-2} (1-2s^{-1})^{-1}\nabla_L(s\tilde{\W}) \cdot {\cpur  A_{\W,\textup{sub}}\tilde{\W}}$ has higher power in $s^{-1}$ and can be absorbed by 

\begin{align*}
\epsilon\cdot M^{-1}s^{p-3}|\nabla_L (s\tilde{\W})|^2+\frac{1}{\epsilon} \cdot M^{-1}s^{p-5}|\nablas (s\tilde{\W})|^2+\frac{1}{\epsilon}\cdot M^{-3} s^{p-7}|  s\tilde{\W}|^2.
\end{align*}
Hence there exists $R_{\W}$ such that for $r\geq R_{\W}$, we can absorb $s^{p-2} (1-2s^{-1})^{-1}\nabla_L(s\tilde{\W}) \cdot {\cpur  A_{\W,\textup{sub}}\tilde{\W}}$ into the positive terms above. We require $R_{\textup{null}}\geq R_{\W}$. From now on we fix $R_{\textup{null}}=\max \{ R_{\h}, R_{Z} ,R_{\W} \}$ and drop the dependence of $R_{\textup{null}}$ in estimates. After integrated along $D(\tau_1,\tau_2)$, $\Div J^p_{4,1}$ is bounded from below by

\begin{align*}
\geq &\frac{1}{C_{rp}} M^{-1}s^{p-3}\left( |\nabla_L (s\tilde{\W})|^2+|\nablas (s\tilde{\W})|^2+M^{-2}|\tilde{\W}|^2\right)\\
+&{\cpur s^{p-1}\D^{-1}\nabla_L(s\tilde{\W}) \cdot ( \Box\eta_{\textup{null}}(r)\cdot \W+2\nabla\eta_{\textup{null}}(r)\cdot\W})\\
                  +&{\cred s^{p-1}\eta_{\textup{null}} \D^{-1}\nabla_L(s\tilde{\W}) \cdot \mathsf{G}}.
\end{align*}
The last line will be part of $\textup{Err}_4[\W,\mathsf{G}]$ and the second last line will be part of $Err_{4,[R_{\textup{null}},R_{\textup{null}+M}]}[\W]$.

To add the $|\nabla_{\underline{L}}\W|^2$ term into the divergence, we consider

\begin{align}
(J_{4,2})_a:=s^{-\delta}\left( T_{ab}[\tilde{\W}]T^b-\frac{1}{2}\tilde{\W}\cdot A_{\W,\textup{main}}\cdot 
\tilde{\W}T_a  \right),
\end{align}
and compute

\begin{align*}
\Div J_{4,2}=&\frac{\delta}{4}M^{-1}s^{-1-\delta}|\nabla_{\underline{L}}\tilde{\W}|^2-{\cpur \frac{\delta}{4}M^{-1}s^{-1-\delta}|\nabla_{L}\tilde{\W}|^2}\\
+&{\cpur \eta_{\textup{null}}(r)  s^{-\delta}\nabla_t\psi\cdot (\Box\eta_{\textup{null}}(r)\cdot\W+2\nabla\eta_{\textup{null}}(r)\cdot\nabla\W)}\\
+&{\cred \eta_{\textup{null}}(r)^2s^{-\delta}\nabla_t\W\cdot \mathsf{G}}+ s^{-1-\delta}\nabla_t\tilde{\W} \cdot A_{\W,\textup{sub}}\tilde{\W}.
\end{align*}
The term involving $A_{\W,\textup{sub}}$ has higher power of $s^{-1}$ and can be estimated as

\begin{align*}
 s^{-1-\delta}\nabla_t\tilde{\W} \cdot A_{\W,\textup{sub}}\tilde{\W}\geq & -\frac{\delta}{4} M^{-1} s^{-1-\delta}|\nabla_t\tilde{\W}|^2-\delta^{-1} M s^{-1-\delta}|A_{\W,\textup{sub}}\tilde{\W}|^2\\
 \geq& -\frac{\delta}{8}s^{-1-\delta} M^{-1} (|\nabla_L\tilde{\W}|^2+|\nabla_{\underline{L}}\tilde{\W}|^2)-C_{rp}M^{-1}s^{-3-\delta}|\nablas\tilde{\W}|^2.
\end{align*}
Hence

\begin{align*}
\Div J_{4,2}\geq &\frac{\delta}{8}M^{-1}s^{-1-\delta}|\nabla_{\underline{L}}\tilde{\W}|^2-{\cpur \frac{3\delta}{8}M^{-1}s^{-1-\delta}|\nabla_{L}\tilde{\W}|^2}\\
+&{\cpur \eta_{\textup{null}}(r)  s^{-\delta}\nabla_t\psi\cdot (\Box\eta_{\textup{null}}(r)\cdot\W+2\nabla\eta_{\textup{null}}(r)\cdot\nabla \W)}\\
-&{\cpur C_{rp}M^{-1}s^{-\delta-3}|\nablas\tilde{\W}|^2}+{\cred \eta_{\textup{null}}(r)^2s^{-\delta}\nabla_t\W\cdot \mathsf{G}}.
\end{align*}
Take $\epsilon_{4,2}$ small enough such that

\begin{align*}
&\frac{1}{C_{rp}} M^{-1}s^{p-3}\left( |\nabla_L (s\tilde{\W})|^2+|\nablas (s\tilde{\W})|^2+M^{-2}|\tilde{\W}|^2\right)\\
\geq &2\epsilon_{4,2}\cdot\left(\frac{3\delta}{8}M^{-1}s^{-1-\delta}|\nabla_{L}\tilde{\W}|^2+C_{rp}M^{-1}s^{-\delta-3}|\nablas\tilde{\W}|^2  \right).
\end{align*}
Then define

\begin{align}
J^p_4=J^p_{4,1}+\epsilon_{4,2}J_{4,2}.
\end{align}
We have, after integrated along $D(\tau_1,\tau_2)$, $\Div J^{p}_4$ is bounded from below by

\begin{align*}
&\frac{1}{C_{rp}} M^{-1}s^{p-3}\left( |\nabla_L (s\tilde{\W})|^2+|\nablas (s\tilde{\W})|^2+M^{-2}|\tilde{\W}|^2\right)+\frac{1}{C_{rp}} M^{-1}  s^{-1-\delta}|\nabla_{\underline{L}}\tilde{\W}|^2\\
+&{\cpur \left(s^{p-1}\D^{-1}\nabla_L(s\tilde{\W})+\epsilon_{4,2}   s^{-\delta}\nabla_t \tilde{\W}  \right)\cdot (\Box\eta_{\textup{null}}(r)\cdot\W+2\nabla\eta_{\textup{null}}(r)\cdot\nabla\W)}\\
+&{\cred \left(s^{p-1}\eta_{\textup{null}}(r)\D^{-1}\nabla_L(s\tilde{\W})+\epsilon_{4,2} \eta_{\textup{null}}(r) s^{-\delta}\nabla_t \tilde{\W}\right)\cdot \mathsf{G}}.
\end{align*}
We denote the last line by $\textup{Err}_4[\W,\mathsf{G}]$ and the second last one by $Err_{4,[R_{\textup{null}},R_{\textup{null}}+M]}[\W]$ and obtain \eqref{rp_bulk}. The estimate \eqref{redshift_boundary} holds easily from the form of $J^p_{4}$.

\subsection{Proof of Proposition \ref{pro:X12}}\label{subsec:estimate_W12}

In this subsection we combine the currents above to construct $J^p_{\W}$. We define

\begin{align}
J^p_{\W}:=&C_1 J_1+\epsilon_2 J_2+J_3+\epsilon_4 J^p_4,\\
\textup{Err}^p[\W,\mathsf{G}]:=&C_1 \textup{Err}_1^p[\W,\mathsf{G}]+\epsilon_2 \textup{Err}_2[\W,\mathsf{G}]+\textup{Err}_3[\W,\mathsf{G}]+\epsilon_4 \textup{Err}^p_4[\W,\mathsf{G}],
\end{align}
with constants $C_1,\epsilon_2,\epsilon_4>0$ determined below. We fix $\epsilon_2$ and $\epsilon_4$ small enough such that

\begin{align*}
K_3+\epsilon_2 Err_{2,[r_{rs,\W},r_{rs}^+]}[\W]+\epsilon_4 Err_{4,[R_{\textup{null}},R_{\textup{null}}+M]}[\W]\geq \frac{1}{2}K_3.
\end{align*}
We take $C_1$ large to ensure \eqref{boundary}, \eqref{boundary_horizon} and \eqref{boundary_null} as follows. We ask $C_1\geq 2C_{Mor}C_T$ to get \eqref{boundary} from \eqref{T_boundary}, \eqref{redshift_boundary}, \eqref{Morawetz_boundary} and \eqref{rp_boundary}. To obtain \eqref{boundary_horizon}, we calculate at $r=2M$,

\begin{align*}
J_1\cdot L=&\frac{1}{2}|\nabla_L \W|^2+\frac{9}{2}|\nabla_L \hat{W}_1|^2+\frac{1}{2}|\nabla_L \hat{W}_2|^2,\\
J_3\cdot L=&\frac{f}{2}|\nabla_L \W|^2.
\end{align*}
Here we used $T=(1-2s^{-1})\frac{\partial}{\partial r}=\frac{1}{2}L$ and $\omega, \bar{f}, \nabla_L\omega,\nabla_L\bar{f}$ all vanish at $r=2M$. Together with  $J_2\cdot L\geq 0$ and $J_4=0$ at $r=2M$, \eqref{boundary_horizon} follows by requiring $C_1\geq |f(2M)|$. To guarantee \eqref{boundary_null}, we calculate for any $r\geq 2M$,

\begin{align*}
\int_{\mathbb{S}^2(\tau,r)} J_1\cdot \underline{L}\ r^2 dvol_{\mathbb{S}^2} \gtrsim  &\int_{\mathbb{S}^2(\tau,r)}\big(|\nabla_{\underline{L}}\W|^2+|\nablas\W_{\ell\geq 2}|\big) r^2 dvol_{\mathbb{S}^2},\\
\left|\int_{\mathbb{S}^2(\tau,r)} J_3\cdot \underline{L}\ r^2 dvol_{\mathbb{S}^2}\right| \lesssim & \int_{\mathbb{S}^2(\tau,r)}\big(|\nabla_{\underline{L}}\W|^2+|\nablas\W|\big) r^2 dvol_{\mathbb{S}^2}.
\end{align*}
The only obstruction to make whole thing positive is $|\nablas\W_{\ell=1}|^2$, which can be controlled as

\begin{align*}
\int_{\mathbb{S}^2(\tau,r)} |\nablas\W_{\ell=1}|^2\ r^2 dvol_{\mathbb{S}^2}\approx &\int_{\mathbb{S}^2(\tau,r)} |\W_{\ell=1}|^2\ dvol_{\mathbb{S}^2}
\lesssim  s^{-(2-\delta)} M^{-1} E^{1-\delta,(0)}_L[ \tilde{\W}](\tau).
\end{align*}
for any $r\geq R_{\textup{null}}+M$. Here we used Corollary \ref{cor:E^p_base} with $p=1-\delta$. Therefore this term goes to zero as $r$ goes to infinity by assumption \ref{assumption}. Together with $J^p_4\cdot\underline{L}\geq 0$ and $J_2=0$ for $r\geq R_{\textup{null}}+M$, \eqref{boundary_null} follows for $C_1$ large enough and the value of $C_1$ is then determined.\\

We can now fix the value of $R_T$, which will be large enough such that the $Err_{1,[R_T,\infty)}[\W]$ term in $\Div J_1$ can be absorbed into $K^p_4$. To begin, we require

\begin{align*}
\frac{R_T}{M}\geq \left(\frac{36 C_1C_{rp}}{\epsilon_4}\right)^{1/\delta}.
\end{align*}
Then for $r\geq R_T$ and $p\in [\delta,2-\delta]$,

\begin{align*}
C_1\left| \frac{9}{4}\eta_T'|\nabla_L \W|^2 \right|\leq  \frac{9C_1}{4}M^{-1}s^{-1}|\nabla_L \W|^2\leq  \frac{1}{16 C_{rp}} M^{-1}s^{\delta-1}|\nabla_L \W|^2\leq \frac{\epsilon_4}{4}K^p_4.
\end{align*}
Next, we absorb the term $\nabla_t \hat{W}_1\done \hat{W}_2$ into $K^p_4$ for $r$ large.

\begin{align*}
\left| \frac{6M}{r^2} \nabla_t\hat{W}_2\cdot\donest \hat{W}_1\right|\lesssim & M^{-1}s^{-2}|\nabla_t\hat{W}_2|^2+M^{-1}s^{-2}|\nablas\hat{W}_1|^2,
\end{align*}
which falls off faster than $K^p_4$ for any $p\geq \delta$. We further require $R_T$ to be large enough such that for $r\geq R_T$,

\begin{align*}
C_1\left|\frac{6M}{r^2}\nabla_t\hat{W}_2\cdot\donest \hat{W}_1\right|\leq \frac{\epsilon_4}{4} K^p_4.
\end{align*}
The value of $R_T$ is then determined. We proceed to fix the value of $b$. Besides $b\geq \max\{s_{Mor},R_T/4\}$, We further require that for all $r\geq bM$ and $p\in [\delta,2-\delta]$,

\begin{align*}
|Err_{3,[bM,\infty)}[\W]|\leq \frac{\epsilon_4}{4}K^p_4.
\end{align*}
This can be done since $|Err_{3,[bM,\infty)}[\W]|$ is bounded by 

\begin{align*}
M^{-1}\bigg(s^{-2}|\nabla_r \W|^2+s^{-2}|\nablas \W|^2+M^{-2}s^{-2}|\W|^2\bigg),
\end{align*}
which decays faster than $K^p_4$. The value of $b$ is now determined and from now on we drop the dependence of $b$ in estimates. We have managed to show that $\Div J^p_{\W}$, after integrated along $D(\tau_1,\tau_2)$, is bounded from below by 

\begin{align*}
\int_{D(\tau_1,\tau_2)} \epsilon_2 K_2+\frac{1}{2}K_3+\frac{\epsilon_4}{4}K^p_4 dvol+\int_{D(\tau_1,\tau_2)}              \textup{Err}^p [\W,\mathsf{G}]+Err_{1,[2M,4R_T]}[\W,\mathsf{F}]+Err_{3,[2M,bM]}[\W,\mathsf{F}] dvol.
\end{align*}

We turn to estimating $Err_{1,[2M,4R_T]}[\W,\mathsf{F}]$ and $Err_{3,[2M,bM]}[\W,\mathsf{F}]$ define in subsection \ref{subsection:T} and subsection \ref{subsection:Morawetz}. Note that $\hat{W}_1=W_1$ for $r\leq bM$. From \eqref{X1t_substitution}, in $[2M,4R_T]\subset [2M,bM]$,  

\begin{align*}
|\nabla_t\hat{W}_1|^2=|\nabla_t W_1|^2= |F_3|^2.
\end{align*}
Similarly, from \eqref{X1r_X2o_substitution}, 

\begin{align*}
|\donest \hat{W}_2|^2=|\donest {W}_2|^2\lesssim |\nabla_rW_1|^2+M^{-2}s^{-2}|W_1|^2+|F_1|^2.
\end{align*}
Thus through Cauchy-Schwarz, one has for any $\epsilon>0$,

\begin{align*}
\left| \int_{\Sigma_\tau} Err_{1,[2M,4R_T]}[\W,\mathsf{F}] dvol_3 \right| \lesssim \epsilon B[\W](\tau)+\frac{1}{\epsilon} \int_{\Sigma_\tau} M^{-1}(|F_1|^2+|F_3|^2)\cdot\chi_{[2M,4R_T]} dvol_3 .
\end{align*}
From \eqref{def:err3} and Cauchy-Schwarz, we have

\begin{align*}
\left| \int_{\Sigma_\tau} Err_{3,[2M,bM]}[\W] dvol_3 \right| \lesssim \epsilon B[\W]+\frac{1}{\epsilon} \int_{\Sigma_\tau} M^{-1}(|F_1|^2+|F_2|)\cdot\chi_{[2M,bM]} dvol_3.
\end{align*}
Thus, by choosing $\epsilon>0$ small enough, we bound the integral of $\Div J^p_{\W}$ in $D(\tau_1,\tau_2)$ by

\begin{align*}
\int_{D(\tau_1,\tau_2)} \frac{1}{2}\bigg(\epsilon_2 K_2+\frac{1}{2}K_3+\frac{\epsilon_4}{4}K^p_4\bigg) dvol+\int_{D(\tau_1,\tau_2)}              \textup{Err}^p [\W,\mathsf{G}] dvol\\
-C \int_{D(\tau_1,\tau_2)} M^{-1}(|F_1|^2+|F_2|^2+|F_3|^2)\cdot\chi_{[2M,bM]}  dvol.
\end{align*}

Next we deal with the term $\textup{Err}^p [\W,\mathsf{G}]$.

\begin{align*}
\textup{Err}^p[\W,\mathsf{G}]=&C_1\cdot({\cred \nabla_t\W\cdot\mathsf{G}+\eta_T\cdot(9\nabla_t W_1\cdot G_1+\nabla_t W_2\cdot G_2)})+\epsilon_2\cdot({\cred \nabla_Y\W\cdot \mathsf{G}})\\
+&{\cred\left(\nabla_X \W+\frac{1}{2}\omega \W-2\bar{f}\W\right)\cdot\mathsf{G}}\\
+&\epsilon_4\cdot{\cred \left(s^{p-1}\eta_{\textup{null}}(r)\D^{-1}\nabla_L(\tilde{s\W})+\epsilon_{4,2} \eta_{\textup{null}}(r)s^{-\delta}\nabla_t \tilde{\W} \right)\cdot \mathsf{G}}.
\end{align*}
From Cauchy-Schwarz, we have for any $p\in [\delta,2-\delta]$ and $\epsilon>0$,

\begin{align*}
\textup{Err}^p[\W,\mathsf{G}]\lesssim \epsilon \bigg(\epsilon_2 K_2+\frac{1}{2}K_3+\frac{\epsilon_4}{4}K^p_4\bigg)+\frac{1}{\epsilon} Ms^{p+1}|G|^2. 
\end{align*}
Again through picking $\epsilon>0$ small enough, we have the integral of $\Div J^p_{\W}$ in $D(\tau_1,\tau_2)$ is bounded from below by

\begin{align*}
\int_{D(\tau_1,\tau_2)} \frac{1}{2}\bigg(\epsilon_2 K_2+\frac{1}{2}K_3+\frac{\epsilon_4}{4}K^p_4\bigg) dvol-C \int_{D(\tau_1,\tau_2)}              Ms^{p+1}|\mathsf{G}|^2 dvol\\
-C \int_{D(\tau_1,\tau_2)} M^{-1}(|F_1|^2+|F_2|^2+|F_3|^2)\cdot\chi_{[2M,bM]}  dvol.
\end{align*}

\begin{proposition}

Let $\W$ be a solution of \eqref{equ:W12} which also satisfies \eqref{X1r_X2o_substitution}, \eqref{X1o_X2r_substitution} and \eqref{X1t_substitution}. Further assume that

\begin{align*}
F[\W](0)+E^{2-\delta,(0)}_L[\W](0) <\infty,
\end{align*}
and 

\begin{align*}
\int_{D(\tau_1,\tau_2)} Ms^{3-\delta}|\mathsf{G}|^2 dvol<\infty,
\end{align*}
for any $\tau_2\geq\tau_1\geq 0$. Then we have for any $p\in [\delta,2-\delta]$ 

\begin{equation}\label{X12_rp_int}
\begin{split}
&F[\W](\tau)+E^{p,(0)}_L[\W](\tau_2)+M^{-1}\int_{\tau_1}^{\tau_2} \bar{B}[\W](\tau)+E^{p-1,(0)}_{L,\nablas}[\W](\tau)d\tau\\
\lesssim &F[\W](\tau)+E^{p,(0)}_L[\W](\tau_1)+\int_{D(\tau_1,\tau_2)} Ms^{p+1}|\mathsf{G}|^2 dvol\\
+& \int_{D(\tau_1,\tau_2) }M^{-1} (|F_1|^2+|F_2|^2+|F_3|^2)\cdot\chi_{[2M,bM]} dvol. 
\end{split}
\end{equation}

\end{proposition}

\begin{proof}
Clearly the integrand of $\bar{B}[\W]$ and $E^{p-1,(0)}[\W]$ is bounded by $\epsilon_{2}K_2+\frac{1}{2}K_3+\frac{\epsilon_4}{4}K^p_4$. By applying divergence theorem \eqref{div_thm} to $J^p_{\W}$, it's sufficient to show that the assumption \eqref{assumption} holds. By applying the divergence theorem \eqref{div_thm} to $J^{2-\delta}_4$ in $D(0,\tau)$, we have

\begin{align*}
&\int_{\Sigma_{\tau}} J^{2-\delta}_4\cdot n\ dvol_3+\int_{D(0,\tau)} K^{2-\delta}_4 dvol\\
\leq &\int_{\Sigma_{0}} J^{2-\delta}_4\cdot n\ dvol_3+\int_{D(0,\tau)}-\textup{Err}^{2-\delta}_4[\W,\mathsf{G}]-Err^{2-\delta}_{4,[R_{\textup{null}},R_{\textup{null}}+M]}[\W] dvol.
\end{align*}
Here we drop the boundary term along null infinity, which is non-negative. Through Cauchy-Schwarz, we have for any $\epsilon>0$,

\begin{align*}
\textup{Err}^{2-\delta}_4[\W,\mathsf{G}]=&{\cred \left(s^{1-\delta}\eta_{\textup{null}}(r)\D^{-1}\nabla_L(s\tilde{\W})+\epsilon_{4,2} \eta_{\textup{null}}(r)^2s^{-\delta}\nabla_t\W\right)\cdot \mathsf{G}}\\
\lesssim &\epsilon M^{-1}\big( s^{-1-\delta}|\nabla_L(s\tilde{\W})|^2+ s^{-1-\delta}|\nabla_t \tilde{\W}|^2\big)+ \frac{1}{\epsilon} M(s^{3-\delta}+s^{-\delta+1})|\mathsf{G}^2|\\
\lesssim &\epsilon K^{2-\delta}_4 +\frac{1}{\epsilon} Ms^{3-\delta}|\mathsf{G}^2|.
\end{align*}
By taking $\epsilon>0$, we get

\begin{align*}
&\int_{\Sigma_{\tau}} J^{2-\delta}_4\cdot n\ dvol_3+\int_{D(0,\tau)} K^{2-\delta}_4 dvol\\
\lesssim &\int_{\Sigma_{0}} J^{2-\delta}_4\cdot n\ dvol_3+\int_{D(0,\tau)}Ms^{3-\delta}|G|-Err^{2-\delta}_{4,[R_{\textup{null}},R_{\textup{null}}+M]}[\W] dvol.
\end{align*}
The term $Err^{2-\delta}_{4,[R_{\textup{null}},R_{\textup{null}}+M]}[\W]$ is supported in $[R_{\textup{null}},R_{\textup{null}}+M]$ and can be shown to have bounded integral by using $e^{-Ct}T$ as multiplier with $C>0$ large enough. Hence from \eqref{rp_boundary}, we conclude that $E^{2-\delta,(0)}[\tilde{\W}](\tau)$ and its integral is finite for any $\tau\geq 0$, which is the assumption \eqref{assumption}.\\

\end{proof}

\begin{proof}[proof of Proposition \ref{pro:X12}]
From the view of \eqref{X12_rp_int}, it suffices to estimate 

\begin{align*}
\int_{D(\tau_1,\tau_2)} Ms^{p+1}|\mathsf{G}|^2 dvol+ \int_{D(\tau_1,\tau_2) }M^{-1} (|F_1|^2+|F_2|^2+|F_3|^2)\cdot\chi_{[2M,bM]} dvol. 
\end{align*}
From the view of Lemma \ref{lem:G12_bound}, the term $Ms^{p+1}|\mathsf{G}|^2$ can be bounded by the integrand of $B[S,\f^{\leq 2}r^{-1}\psi_Z]$ and $E^{p-1,(0)}_{L,\nablas}[S,\f^{\leq 2}r^{-1}\psi_Z]$ in the regions $r\in [2M,2bM]$ and $r\in [2bM,\infty)$ respectively. From Proposition \ref{cor:source_high_order}, 

\begin{align*}
\int_{D(\tau_1,\tau_2)} Ms^{p+1}|\mathsf{G}|^2 dvol \lesssim \left(1+\frac{\tau_1}{M}\right)^{-2+p+\delta}M^2 I^{(0)}[S,\f^{\leq 2} r^{-1}\psi_Z].
\end{align*}
We don't keep the dependence on $b$ and $R_{\textup{null}}$ since their values are already fixed. Now we turn to $|F_i|^2$ terms. Through the definition of $F_1,F_2$ and $F_3$ in \eqref{def:F} , 

\begin{align*}
M^{-1} (|F_1|^2+|F_2|^2+|F_3|^2)\cdot\chi_{[2M,bM]}\lesssim &\textup{the intrgrand of}\ \bigg(M^{-1}B[\K^{\leq 1}\hat{W}_0,\hat{P}_{\even}]\\
&+M B[\f S, \f^{\leq 3}r^{-1}\psi_Z,Q_{\even}]\bigg).
\end{align*}
Thus from Propositions \ref{pro:X0}, \ref{pro:P}, \ref{pro:Q} and Proposition \ref{cor:source_high_order}, we have

\begin{align*}
\int_{D(\tau_1,\tau_2) }M^{-1} (|F_1|^2+|F_2|^2+|F_3|^2)\cdot\chi_{[2M,bM]} dvol\lesssim \left(1+\frac{\tau_1}{M}\right)^{-2+p+3\delta}\cdot \textup{Decay}'[\W].
\end{align*}
Here

\begin{align*}
\textup{Decay}'[\W]= &I^{(0)}[\K^{\leq 1}\hat{W}_0,\K^{\leq 3}\hat{P}_{\even}]+M^2I^{(0)}[\f^{\leq 1}\K^{\leq 5}S,\f^{\leq 3}\K^{\leq 5}r^{-1}\psi_Z,Q_{\even}]\\
+&M^2 I^{(1)}[\K^{\leq 5}S,\K^{\leq 7}r^{-1}\psi_Z].
\end{align*}
This together with \eqref{X12_rp_int} and \eqref{energy_by_rp} implies that for any $p\in [1,2-3\delta]$ and $\tau_2\geq \tau_1\geq 0$, we have

\begin{equation}\label{X12_rp_int_f}
\begin{split}
&F[\W](\tau_2)+E^{p,(0)}_L[\W](\tau_2)+M^{-1}\int_{\tau_1}^{\tau_2} F[\W](\tau)+E^{p-1,(0)}_{L,\nablas}[\W](\tau)d\tau\\
\lesssim &F[\W](\tau_1)+E^{p,(0)}_L[\W](\tau_1)+ \left(1+\frac{\tau_1}{M}\right)^{-2+p+3\delta}\cdot \textup{Decay}'[\W].
\end{split}
\end{equation}
Furthermore, as $p\in [\delta,1)$, 
\begin{equation}\label{X12_rp_int_g}
\begin{split}
&F[\W](\tau_2)+E^{p,(0)}_L[\W](\tau_2)+M^{-1}\int_{\tau_1}^{\tau_2} E^{p-1,(0)}_{L,\nablas}[\W](\tau)d\tau\\
\lesssim &F[\W](\tau_1)+E^{p,(0)}_L[\W](\tau_1)+ \left(1+\frac{\tau_1}{M}\right)^{-2+p+3\delta}\cdot \textup{Decay}'[\W].
\end{split}
\end{equation}
The result then follows through an elementary argument we sketch below. By taking $\bar{\tau}_k:=2^k M$ and applying \eqref{X12_rp_int_f} with $\tau_2=\bar{\tau}_{k+1}  , \tau_1=\bar{\tau}_k$ and $p=2-3\delta$ together with the mean value theorem, we get a sequence $\tilde{\tau}_k\approx 2^k M$ such that

\begin{align*}
F[\W](\tilde{\tau}_k)+E^{1-3\delta,(0)}_L[\W](\tilde{\tau}_k)\lesssim &\left(1+\frac{\tilde{\tau}_k}{M}\right)^{-1} \left(F[\W](\bar{\tau}_k)+E^{2-3\delta,(0)}_L[\W](\bar{\tau}_k)+ \textup{Decay}'[\W]\right)\\
\lesssim &\left(1+\frac{\tilde{\tau}_k}{M}\right)^{-1} \textup{Decay}[\W].
\end{align*}
Here we used \eqref{X12_rp_int_f} with $\tau_2=\bar{\tau}_k, \tau_1=0$ and $p=2-3\delta$ to estimate $F[\W](\bar{\tau}_k)+E^{2-3\delta,(0)}_L[\W](\bar{\tau}_k)$ and $\textup{Decay}[\W]$ is defined in Proposition \ref{pro:X12}. Through \eqref{X12_rp_int_g} with $p=1-3\delta$ and $\tau_2=\tau \geq \tilde{\tau}_k=\tau_1$,

\begin{align*}
F[\W]({\tau})+E^{1-3\delta,(0)}_L[\W]({\tau})\lesssim &\left(1+\frac{\tilde{\tau}_k}{M}\right)^{-1} \textup{Decay}[\W]. 
\end{align*}
As each $\tau\geq  0$ is in some interval $[\tilde{\tau}_k,4\tilde{\tau}_k]$ or $[0,4M]$, we conclude for all $\tau\geq 0$,

\begin{align*}
F[\W]({\tau})+E^{1-3\delta,(0)}_L[\W]({\tau})\lesssim &\left(1+\frac{{\tau}}{M}\right)^{-1} \textup{Decay}[\W]. 
\end{align*}
Through interpolation we have for $p\in [1-3\delta,2-3\delta]$,

\begin{align*}
F[\W]({\tau})+E^{p,(0)}_L[\W]({\tau})\lesssim &\left(1+\frac{{\tau}}{M}\right)^{-2+p+3\delta} \textup{Decay}[\W]. 
\end{align*}
Repeating the argument one more time gives the result.

\end{proof}

\section{Proof of Theorem \ref{thm:ell_geq_2} and \ref{thm:ell_geq_1}}\label{sec:proof}

In this section we prove Theorem \ref{thm:ell_geq_2} and \ref{thm:ell_geq_1}. We first use $\psi_Z$ and $W$ to bound $h$ in the region $r\in [r_{rs,\h},R_{\textup{null}}+M]$, which leads to the decay of $F[h](\tau)$ from the view of Propositions \ref{pro:metric_far}, \ref{pro:psi}, \ref{pro:X0} and \ref{pro:X12}. Then the main results follow by expressing the initial norms of $r^{-1}\psi_Z, \hat{P}_{\even}$ and so on in terms of $h$ or $W$.

\begin{proposition}
Let $h_{ab}$ be an \textbf{even} solution of \eqref{HG} and \eqref{Main_equation} supported on $\ell\geq 1$. Further assume that $\textup{Decay}'[h]$ defined below is finite. Then for any $p\in [\delta,2-3\delta]$ and $\tau_2\geq \tau_1\geq 0$ we have

\begin{align*}
&F[h](\tau_2 )+E^{p,(0)}_L[\tilde{h}](\tau_2)+\int_{\tau_1}^{\tau_2} \bar{B}[h](\tau)+E^{p-1,(0)}_{L,\nablas}[\tilde{h}](\tau)d\tau\\
\lesssim &F[h](\tau_1 )+E^{p,(0)}_L[\tilde{h}](\tau_1)+\left(1+\frac{\tau_1}{M}\right)^{-2+3\delta} \textup{Decay}'[h].
\end{align*}
Here 

\begin{align*}
\textup{Decay}'[h]:=& M^{-2}F[\K^{\leq 1}\W](0)+M^{-2}E^{2-\delta,(0)}_L[\K^{\leq 1}\W](0)+ M^{-2}I^{(0)}[\K^{\leq 2}\hat{W}_0,\K^{\leq 4}\hat{P}_{\even}]\\
+&I^{(0)}[\f^{\leq 1}\K^{\leq 6}S,\f^{\leq 3}\K^{\leq 6}r^{-1}\psi_Z,\K^{\leq 1}Q_{\even}]+ I^{(1)}[\K^{\leq 6}S,\K^{\leq 8}r^{-1}\psi_Z],
\end{align*}
and $\tilde{h}=\eta_{\textup{null}}h$ with $\eta_{\textup{null}}(r)$ defined in \eqref{cutoff_rp} and the initial norm $I^{(k)}[\cdot]$ is defined in \eqref{initial_norm}.

\end{proposition}

\begin{proof}
From Proposition \ref{pro:metric_far}, it suffices to show that

\begin{align*}
M^{-3} \int_{D(\tau_1,\tau_2)\cap \{r_{rs,\h}\leq r\leq R_{\textup{nul}}+M\}} |(M\partial)^{\leq 1}h|^2  dvol \lesssim \left(1+\frac{\tau_1}{M}\right)^{-2+3\delta} \textup{Decay}'[h].
\end{align*}
From $h=h^{\RW}-{}^W\pi$, \eqref{pi} and \eqref{RW_by_psi}, we have in $r\in [r_{rs,\h},R_{\textup{nul}}+M]$,

\begin{align*}
M^{-3}|(M\partial)^{\leq 1} h|^2\lesssim & M^{-3}|\f^{\leq 3}r^{-1}\psi_Z|^2+M^{-5}|(M\partial)^{\leq 2}W|^2.              
\end{align*}
From \eqref{killing_to_partial}, \eqref{equ:W_0}, \eqref{equ:W_1} and \eqref{equ:W_2}, the above is bounded by

\begin{align*}
M^{-3}|\f^{\leq 3}r^{-1}\psi_Z|^2+M^{-5}|(M\partial)^{\leq 1}\K^{\leq 1}\W|^2+ M^{-5}|(M\partial_r) \K^{\leq 1}{W}_0|^2+M^{-5}|\K^{\leq 2}{W}_0|^2,
\end{align*}
which is bounded by the integrand of $M^{-1}B[\f^{\leq 3}r^{-1}\psi_Z]+M^{-3} \bar{B}[\K^{\leq 1}\W]+M^{-3}{B}[\K^{\leq 2}\hat{W}_0]$. Then the result follows from Proposition \ref{pro:psi}, \ref{pro:X0} and \ref{pro:X12}.

\end{proof}

\begin{proof}[proof of Theorem \ref{thm:ell_geq_2} (1)]

From the previous proposition and the dyadic argument at the end of section \ref{sec:W12}, the result holds with $\textup{Decay}_{\ell\geq 2}[h]$ replaced by $\textup{Decay}'[h]+F[h](0)+E^{2-\delta,(0)}_L[\tilde{h}](0)$. As $F[h](0)+E^{2-\delta,(0)}_L[\tilde{h}](0)$ is clearly bounded by $\textup{Decay}_{\ell\geq 2}[h]$, it remains to show that $\textup{Decay}'[h]$ is also bounded by $\textup{Decay}_{\ell\geq 2}[h]$. Fix an integer $m\geq 0$. From the definition of $W_0$, $W_1$ and $W_2$, \eqref{def:W0}, \eqref{def:W1} and \eqref{def:W2}, we have

\begin{align*}
|(M\partial)^{\leq m}W_0|^2,\ |(M\partial)^{\leq m}W_1|^2,\ |(M\partial)^{\leq m}W_2|^2\lesssim_{s,m} |Ms^2(M\partial)^{\leq m+1}h|^2.
\end{align*}
From the definition of $\psi_Z$, $P_{\even}$ and $Q_{\even}$, \eqref{def:psi}, \eqref{def:P} and \eqref{def:Q}, we have

\begin{align*}
|(M\partial)^{\leq m} r^{-1}\psi_Z|^2\lesssim_{s,m} |s (M\partial)^{\leq m+1}h|^2,\\
|(M\partial)^{\leq m} P_{\even}|^2\lesssim_{s,m} |Ms^2(M\partial)^{\leq m+1}h|^2,\\
|(M\partial)^{\leq m} Q_{\even}|^2\lesssim_{s,m} |s^2(M\partial)^{\leq m+3}h|^2.
\end{align*}
Note that $P_{\even}$ can be bounded by one instead of two more derivative of $h$ since there's a cancellation in the definition of $P_{\even}$ \eqref{def:P} from \eqref{def:W0} and \eqref{def:W1}. Similarly, even though $Q_{\even}$ in \eqref{def:Q} involves three more angular derivative of $\psi_Z$, from \eqref{def:psi} the angular derivative will actually cancels and $Q_{\even}$ can be bounded by three more derivatives of $h$. Then along $\Sigma_0$,

\begin{align*}
|(M\partial)^{\leq m}\hat{W}_0|^2,\ |(M\partial)^{\leq m}\hat{W}_1|^2,\ |(M\partial)^{\leq m}\hat{W}_2|^2,\ |(M\partial)^{\leq m}\hat{P}_\even|^2\lesssim_{s,m} |Ms^3 (M\partial)^{\leq m+2}h|^2.
\end{align*}
Thus

\begin{align*}
\textup{Decay}'[h]=& M^{-2}F[\K^{\leq 1}\W](0)+M^{-2}E^{2-\delta,(0)}_L[\K^{\leq 1}\W](0)+ M^{-2}I^{(0)}[\K^{\leq 2}\hat{W}_0,\K^{\leq 4}\hat{P}_{\even}]\\
+&I^{(0)}[\f^{\leq 1}\K^{\leq 6}S,\f^{\leq 3}\K^{\leq 6}r^{-1}\psi_Z,\K^{\leq 1}Q_{\even}]+ I^{(1)}[\K^{\leq 6}S,\K^{\leq 8}r^{-1}\psi_Z]\\
\lesssim &I^{(0)}[s^3 (M\partial)^{\leq 2}\K^{\leq 4}h, s^2(M\partial)^{\leq 3}\K^{\leq 1}h ,s\f^{\leq 3}(M\partial)^{\leq 1}\K^{\leq 6}h]\\
+&I^{(1)}[s(M\partial)^{\leq 1}\K^{\leq 8}h]\\
=&\textup{Decay}_{\ell\geq 2}[h].
\end{align*}
And the proof is finished.

\end{proof}

\begin{proof}[proof of Theorem \ref{thm:ell_geq_2} (2)]
In this proof we assume all quantities are supported on the mode $\ell=1$. Similar to the proof of Theorem \ref{thm:ell_geq_2} (1), it suffices to show that $\textup{Decay}'[h]$ is bounded by $\textup{Decay}_{\ell=1}[h]$. From the definition of $W_0$, $W_1$ and $W_2$, \eqref{def:W0_ell=1}, \eqref{def:W1_ell=1} and \eqref{def:W2_ell=1} and noting the the $\done$, $\donest$ behaves like $r^{-1}$ on fixed mode, we have

\begin{align*}
|(M\partial)^{\leq m}W_0|^2,\ |(M\partial)^{\leq m} W_1|^2,\ |(M\partial)^{\leq m} W_2|^2\lesssim_{s,m} |Ms^3(M\partial)^{\leq m+1}h|^2.
\end{align*}
From the definition of $P_{\even}$ and $Q_{\even}$, \eqref{def:P} and \eqref{def:Q}, we have

\begin{align*}
|(M\partial)^{\leq m} P_{\even}|^2\lesssim_{s,m} &|Ms^4(M\partial)^{\leq m+2}h|^2,\\
|(M\partial)^{\leq m} Q_{\even}|^2\lesssim_{s,m} &|s^3(M\partial)^{\leq m+2}h|^2.
\end{align*}
Hence with $\psi_Z\equiv 0$, along $\Sigma_0$ we have

\begin{align*}
|(M\partial)^{\leq m} \hat{P}_\even|^2\lesssim_{s,m} |Ms^4 (M\partial)^{\leq m+2}h|^2.
\end{align*}
Thus,

\begin{align*}
\textup{Decay}'[h]=& M^{-2}F[\K^{\leq 1}\W](0)+M^{-2}E^{2-\delta,(0)}_L[\K^{\leq 1}\W](0)+ M^{-2}I^{(0)}[\K^{\leq 2} {W}_0,\K^{\leq 4}\hat{P}_{\even}]\\
+&I^{(0)}[\f^{\leq 1}\K^{\leq 6}S,\K^{\leq 1}Q_{\even}]+ I^{(1)}[\K^{\leq 6}S]\\
\lesssim &I^{(0)}[s^4(M\partial)^{\leq 2}\K^{\leq 4}h,\f^{\leq 1}\K^{\leq 6} h ]+I^{(1)}[\K^{\leq 6} h]\\
=&\textup{Decay}_{\ell=1}[h].
\end{align*}

\end{proof}

\begin{proof}[proof of Theorem \ref{thm:ell_geq_1}]

We recall that in this case $r^{-1}\psi_Z$ is zero, $\hat{W}_i=W_i$ and $S=-S_W$. From Proposition \ref{pro:X0} and Proposition \ref{pro:X12}, it suffices to show that $\textup{Decay}[\hat{W}_0]$ and $\textup{Decay}[\W]$ are bounded by $\textup{Decay}_{\ell\geq 1}[W]$. In particular, $\textup{Decay}[\hat{W}_0]\leq \textup{Decay}[\W]$ as we used $W_0$ in controlling $W_1$ and $W_2$. Hence it is enough to deal with $\textup{Decay}[\hat{W}_0]$. From the definition of $S_W$, $P_{\even},\ \hat{P}_{\even}$ and $Q_{\even}$, we have along $\Sigma_0$,

\begin{align*}
|(M\partial)^{\leq m} S_W|^2\lesssim_m &| M^{-1}(M\partial)^{\leq m+1}W|^2,\\
|(M\partial)^{\leq m}P_{\even}|^2,|(M\partial)^{\leq m}\hat{P}_{\even}|^2\lesssim_m &|s(M\partial)^{\leq m+1}W|^2,\\
|(M\partial)^{\leq m}Q_{\even}|^2\lesssim_m & |M^{-1}s(M\partial)^{\leq m+2}W|^2.
\end{align*}

Hence

\begin{align*}
\textup{Decay}[\W]=&F[\W](0)+E^{2-\delta,(0)}_L[\W](0)+ I^{(0)}[\K^{\leq 1}\hat{W}_0,\K^{\leq 3}\hat{P}_{\even}]\\
+&M^2I^{(0)}[\K^{\leq 5}\f^{\leq 1}S,Q_{\even}]+M^2 I^{(1)}[\K^{\leq 5}S]\\
\lesssim &I^{(0)}[s(M\partial)^{\leq 1}\K^{\leq 3}W,\f^{\leq 1}(M\partial)^{\leq 1}\K^{\leq 5}W]+I^{(1)}[(M\partial)^{\leq 1}\K^{\leq 5}W]\\
=&\textup{Decay}_{\ell\geq 1}[W].
\end{align*}

\end{proof}

\section{The $\ell=0$ mode}\label{sec:ell=0}

In this section we turn to the mode $\ell=0$. We start with equation \eqref{wave_equation_vector} with spherically symmetric solutions. Then we discuss the $\ell=0$ mode of equations \eqref{HG} and \eqref{Main_equation} in subsection \ref{subsec:sad}. To begin, we estimate the trace of the deformation tensor, $S_W$, and \ $P_{\even}$ as before using vector field method in Appendix \ref{sec:wave equation}. However, it's impossible to show $W_a$, a solution of \eqref{wave_equation_vector} supported on $\ell=0$, decays as there exists a stationary solution $W^*$ with finite initial energy. The explicit form of $W^*$ is given by
\begin{align}\label{def:W*}
W^*_adx^a:=\frac{1}{r}dt+\frac{2M}{r^2}\D^{-1}dr.
\end{align} 
Therefore we adapt a different approach. \\

Through equation \eqref{wave_equation_vector}, one can show that provided $P_{\even}$ and $S_W$ both vanish, $W_a$ is actually stationary (\eqref{W0t} and \eqref{W-1t}). Then the wave equation, from the view of \eqref{Box decom}, becomes a second order ODE in $\rho$. This ODE has two explicit linear independent solutions but only one is smooth upto the horizon $\rho=2M$, which produces the stationary solution we discussed.\\
With the estimate of $P_{\even}$ and $S_W$, we can follow the argument above by introducing error terms, denoted by $O_j$, that would vanish if $P_{\even}$ and $S_W$ do. Instead of showing there's no singular part, we will argue that the coefficient of the singular solution decays through the red-shift estimate.\\

Let $W_a$ be an spherically symmetric solution of \eqref{wave_equation_vector}. It has only two components $W_adx^a=W_0 dt+\D^{-1} W_1 dr$, and $W_0$ and $W_1$ are spherical symmetric functions. Similar to \eqref{equ:W_0} and \eqref{equ:W_1} The equation \eqref{wave_equation_vector} can be rewritten as

\begin{align}
\Box W_0+\frac{2M}{r^3}W_0=-\frac{2M}{r^3}P_{\even},\label{equ:W_0,l=0}\\
\Box W_1-\frac{2}{r^2}\left(1-\frac{4M}{r}\right)W_1=\frac{M}{r^2}S_W.\label{equ:W_1,l=0}
\end{align}
Here $P_{\even}$ and $S_W$ are defined as in \eqref{def:P} and in \eqref{def:SW} with $W_2\equiv 0$. Similar to \eqref{equ:P} and \eqref{equ:S}, $P_{\even}$ and $S_W$ satisfies the wave equations

\begin{equation}\label{equ:P,l=0}
\Box P_{\even}=\nabla_t S_W,
\end{equation}

\begin{equation}\label{equ:S,l=0}
\Box S_W=0.
\end{equation}
Recall that $b\geq 2R_{\textup{null}}/M$ is a large number determined in subsection \ref{subsec:estimate_W12} and $\eta_b(r)$ is a cut-off function which equals zero in $[2M,bM]$ and equals one in $[2bM,\infty)$. As in section \ref{sec:4}, define
 
\begin{align}
\hat{P}_{\even}:=&P_{\even}+\eta_b \cdot\left(\frac{u+r}{2}S_W\right).
\end{align}
Now we state the main result for the spherically symmetric solutions of \eqref{wave_equation_vector}.

\begin{theorem}\label{thm:vector,l=0}
Let $W_a$ be an \textbf{even} solution of \eqref{wave_equation_vector} and is supported on $\ell=0$. Recall that $W^*$ is a stationary solution defined in \eqref{def:W*}. Let $S_W$, $P_{\even}$ and $\hat{P}_{even}$ be defined as in \eqref{def:SW}, \eqref{def:P} and \eqref{def:hP} with $\psi_Z\equiv 0$. Then

\begin{align*}
c_1(\infty):=-2M(M S_W-  P_{\even}-{W}_0)\bigg|_{\mathcal{H}^+}
\end{align*}
is a constant. Further suppose that $\textup{Decay}_{\ell=0}[W]$ defined below is finite. Then

\begin{align}
\hat{W} :=W -c_1(\infty)W^*
\end{align}
converges to zero with the estimate

\begin{equation*}
M^{-2}\int_{\Sigma''_\tau}|(M\partial)^{\leq 1} \hat{W}|^2 dvol_3+M^{-3} \int_{D''(\tau,\infty)} |(M\partial)^{\leq 1} \hat{W}|^2 dvol \lesssim \left(1+\frac{\tau}{M}\right)^{-2+2\delta} \textup{Decay}_{\ell=0}[W].
\end{equation*}
Here

\begin{equation*}
\begin{split}
\textup{Decay}_{\ell=0}[W]:=&F[\f^{\leq 2} W](0)+\bigg(F[\f^{\leq 7}\hat{P}_{\even}]+E^{2-\delta,(0)}_L[\f^{\leq 7}\hat{P}_{\even}]\bigg)(0)\\
+&\bigg(F[\f^{\leq 7}S_W]+E^{4-\delta,(0)}_L[\f^{\leq 7}S_W]\bigg)(0).
\end{split}
\end{equation*}

\end{theorem}

\subsection{Analysis of $S_W$ and $P_{\even}$}

By applying Proposition \ref{cor:free_ell=0} (1) to \eqref{equ:S,l=0}, we obtain

\begin{proposition}\label{pro:S,l=0}

Let $m\geq 0$ be a fixed integer. For any $p\in [\delta,4-\delta]$ and $\tau\geq 0$, we have

\begin{equation*}
\begin{split}
&F[\f^{\leq m}S_W](\tau)+E^{p,(0)}_L[\f^{\leq m}\tilde{S}_W](\tau)+ M^{-1} \int_{\tau}^{\infty} \bar{B}[\f^{\leq m} S_W](\tau')+E^{p-1}_L[\f^{\leq m} \tilde{S}_W](\tau')d\tau'\\
\lesssim_m & \left(1+\frac{\tau}{M}\right)^{-4+\delta+p} \left( F[\f^{\leq m} S_W]+E^{4-\delta,(0)}_L[\f^{\leq m} \tilde{S}_W]\right)(0),
\end{split}
\end{equation*}
provided the right hand side is finite.
\end{proposition}

Let $G_P:=\Box \hat{P}_{\even}.$ To apply Proposition \ref{cor:free_ell=0} (2) to $\hat{P}_{\even}$, we need a bound $\textup{I}_{\textup{source},\ell=0,\delta}[\f^{\leq m}G_P]$ defined in \eqref{def:I0}. Similar to lemma \ref{lem:GP_bound}, we have

\begin{lemma}\label{lem:GP_bound,l=0}
Fix an integer $m\geq 0$. In $r\in [2M,bM]$, we have

\begin{align*}
|\f^{\leq m}G_P|\lesssim_m M^{-1}(|(M\partial)^{\leq 1} \f^{\leq m} S_W|.
\end{align*}
As $r\in [bM,2bM]$, we have

\begin{align*}
|\f^{\leq m}G_P|\lesssim_m M^{-1}\left( 1+\frac{\tau}{M} \right) |(M\partial)^{\leq 1}\f^{\leq m}S_W|.
\end{align*}
And as $r\in [2bM,\infty)$,

\begin{align*}
|\f^{\leq m} G_P|\lesssim &s^{-1}|\nabla_L \f^{\leq m} S_W|+M^{-1}s^{-2}|\f^{\leq m} S_W|.
\end{align*}

\end{lemma}

The direct consequence is

\begin{align*}
\textup{I}_{\textup{source},\ell=0,\delta}[\f^{\leq m}G_P] \lesssim_m & M^2\left( F[\f^{\leq m}S_W]+E^{4-\delta,(0)}_L[\f^{\leq m}\tilde{S}_W]\right)(0).
\end{align*}
Therefore, from Proposition \ref{cor:free_ell=0} (2), we have

\begin{proposition}\label{pro:P,l=0}

Let $m\geq 0$ be a fixed integer. Further assume that $\textup{Decay}_{\ell=0}[\f^{\leq m} \hat{P}_{\even}]$ defined below is finite. Then for any $p\in [\delta,2-\delta]$ and $\tau\geq 0$,

\begin{equation*}
\begin{split}
&F[\f^{\leq m}\hat{P}_{\even}](\tau)+E^{p,(0)}_L[\f^{\leq m}\hat{P}_{\even}](\tau)+\int_\tau^{\infty} \bar{B}[\f^{\leq m}\hat{P}_{\even}](\tau')+E^{p,(0)}_L[\f^{\leq m}\hat{P}_{\even}](\tau') d\tau'\\
 \lesssim_m &\left(1+\frac{\tau}{M}\right)^{-2+\delta+p}\cdot \textup{Decay}_{\ell=0}[\f^{\leq m} \hat{P}_{\even}].
\end{split}
\end{equation*}
Here
\begin{align*}
\textup{Decay}_{\ell=0}[\f^{\leq m} \hat{P}_{\even}]=\left(F[ \f^{\leq m}\hat{P}_{\even}]+E^{2-\delta,(0)}_L[\f^{\leq m}\hat{P}_{\even}]\right)(0)+ M^2\left( F[\f^{\leq m} S_W]+E^{4-\delta,(0)}_L[\f^{\leq m} S_W]\right)(0).
\end{align*}
\end{proposition}

\subsection{Analysis of $W$}

Define $W_{3}$ as

\begin{align}
W_{3}:=\D^{-1}(W_1-W_0).
\end{align}
Since on the horizon $dt=\D^{-1} dr$, $W_{3}$ is smooth upto the horizon. From direct computation, $P_{\even}$ and $S_W$ can be rewritten as

\begin{align}
\label{def:P0} P_{\even}=&r\nabla_t W_{3}+r \nabla_{\underline{L}'}W_0-W_0,\\
\label{def:SW0} S_W=&-2 \nabla_{\underline{L}'}W_0+2\D\nabla_rW_{3}+\frac{4}{r}W_0+\frac{4}{r}\left(1-\frac{M}{r}\right)W_{3}.
\end{align}
and $W_{3}$ satisfies the equation

\begin{equation}\label{equ:W-1}
\Box W_{3}-\frac{2M}{r^2} \nabla_{\underline{L}'}W_{3}-\frac{2}{r^2}\D W_{3}-\frac{2}{r^2}W_0=0.
\end{equation}
Note that the time derivative of $W_0$ can be controlled by $P_{\even}$ and $S_W$ through

\begin{align}\label{W0t}
\nabla_t W_0=\frac{r}{2}\D\nabla_r S_W-\nabla_t P_{\even}.
\end{align}
Here we used the equation for $W_1$ as well as spherical symmetry of $W$. We denote

\begin{equation}\label{def:O1}
O_1:=\frac{r}{2}\D\nabla_r S_W-\nabla_t P_{\even}.
\end{equation}
Recall $g(r)$ is the function such that $\Sigma_{\tau}=\{t+r+2M\log(r/2M-1)-g(r)=\tau\}$ and that $\rho=r$ together with $\tau, \theta$ and $\phi$ forms a coordinate. From \eqref{Box decom}, for spherical symmetric functions, 

\begin{align*}
\Box=\D\frac{\partial^2}{\partial \rho^2}+\frac{2}{r}\left(1-\frac{M}{r}\right)\frac{\partial}{\partial \rho}+Z\frac{\partial}{\partial \tau},
\end{align*}
with

\begin{equation*}
\begin{split}
Z=&\left( -2 \frac{dg}{dr}+\D \left(\frac{dg}{dr}\right)^2 \right)\frac{\partial}{\partial \tau}+\left( 2-2\D \frac{dg}{dr} \right)\frac{\partial}{\partial \rho}\\
&+\left(-\D \frac{d^2g}{dr^2}-\frac{2}{r}\left(1-\frac{M}{r}\right)\frac{dg}{dr}+\frac{2}{r}\right).
\end{split}
\end{equation*}
Then \eqref{equ:W_0,l=0} can be rewritten as

\begin{equation}\label{equ:ODE_W0}
\D\partial^2_\rho M_0+\frac{2}{r}\left(1-\frac{M}{r}\right)\partial_\rho M_0+\frac{2M}{r^3}M_0=O_2,
\end{equation}
where

\begin{align}\label{def:O2}
O_2:=-Z\cdot O_1 -\frac{2M}{r^3}P_{\even}.
\end{align}
We view \eqref{equ:ODE_W0} as a second order ODE with source term $O_1$. The homogeneous solutions are

\begin{equation}
\begin{split}
W_{0,1}(\rho)&=\frac{1}{r},\\
W_{0,2}(\rho)&=1+\frac{2M}{r}\log \left(\frac{r}{2M}-1\right).
\end{split}
\end{equation}
The Wronskian of $W_{0,1}$ and $W_{0,2}$ is

\begin{align*}
\left|\begin{array}{cc}
W_{0,1} & W_{0,2}\\
W'_{0,1} & W'_{0,2}
\end{array}\right|=\frac{1}{r^2}\D^{-1}.
\end{align*} 
Then $W_0$ can be expressed as

\begin{equation}
W_0(\tau,\rho)=c_1(\tau)W_{0,1}(\rho)+c_2(\tau)W_{0,2}(\rho)+O_3,
\end{equation}
where 

\begin{equation}
\begin{split}
c_1(\tau)=&r^2\D \left|\begin{array}{cc}
W_{0}(\tau,4M) & W_{0,2}(4M)\\
\partial_\rho W_{0}(\tau,4M) & \partial_\rho W_{0,2}(4M)
\end{array}\right|\\
=&-8M^2 \partial_\rho W_{0}(\tau,4M)+2M W_{0}(\tau,4M),
\end{split}
\end{equation}

\begin{equation}\label{W0:ode}
\begin{split}
c_2(\tau)=&r^2\D \left|\begin{array}{cc}
W_{0,1}(\tau,4M) & W_{0}(4M)\\
\partial_\rho W_{0,1}(\tau,4M) & \partial_\rho W_{0}(4M)
\end{array}\right|\\
=&2M\cdot\partial_\rho  W_{0}(\tau,4M)+\frac{1}{2} W_{0}(\tau,4M).
\end{split}
\end{equation}
And

\begin{equation}\label{def:O3}
\begin{split}
O_3(\tau,\rho)=&\int_{4M}^\rho -(\rho')^2W_{0,2}(\rho')O_2(\tau,\rho') d\rho'\cdot W_{0,1}(\rho)+\int_{4M}^\rho (\rho')^2W_{0,1}(\rho')O_2(\tau,\rho') d\rho'\cdot W_{0,2}(\rho).
\end{split}               
\end{equation}

We will show that $c_2(\tau)$ converges to zero. The time derivative of $c_1(\tau)$ and $c_2(\tau)$ are already controlled as

\begin{align}\label{def:Oc1}
c_1'(\tau)=&-8M^2\partial_\rho O_1 (\tau,4M)+2MO_1 (\tau,4M)=:O_{c1},
\end{align}
and

\begin{align}\label{def:Oc2}
c_2'(\tau)=& 2M\cdot\partial_\rho  O_1 (\tau,4M)+\frac{1}{2} O_1 (\tau,4M)=:O_{c2}.
\end{align}
By \eqref{def:P0} and $\partial_\rho=-\underline{L}'+\frac{d g}{dr}\partial_\tau$, we obtain

\begin{align*}
P_{\even}=r\partial_\tau W_{3}+r\frac{d g}{dr} \partial_\tau W_0-r\partial_\rho W_0-W_0.
\end{align*}
Equivalently,

\begin{align*}
\partial_\tau W_{3}=&\frac{1}{r}P_{\even}-\frac{dg}{dr} \partial_\tau W_0+\partial_\rho W_0+\frac{1}{r}W_0\\
                 =&\left( \frac{1}{r}P_{\even}-\frac{dg}{dr} O_1+\left(\partial_\rho+\frac{1}{r}\right)O_3 \right)+c_2(\tau)\cdot\frac{1}{r}\D^{-1}\\
                 =&\frac{c_2(\tau)}{r}\D^{-1}+O_4,
\end{align*}
where $O_4$ is defined by the last equality and we used \eqref{W0:ode} in the second equality. By taking one more time derivative, we have

\begin{align*}
\partial^2_\tau W_{3}=&O_{c_2} \cdot\frac{1}{r}\D^{-1}+\partial_\tau O_4=:O_5.
\end{align*}
By $\underline{L}'=\frac{dg}{dr}\partial_\tau-\partial_\rho$ and \eqref{Box decom}, \eqref{equ:W-1} can be written as

\begin{align*}
0=&\D \partial_\rho^2 W_{3}+\frac{2}{r}\partial_\rho W_{3}+\left( Z-\frac{2M}{r^2}g' \right)\partial_\tau W_{3}\\
&-\frac{2}{r^2}\D  W_{3}-\frac{2}{r^2}W_0.
\end{align*}
Taking one more time derivative, we obtain

\begin{align*}
0=&\left[ \D\partial_\rho^2+\frac{2}{r}\partial_\rho-\frac{2}{r^2}\D \right]\partial_\tau W_{3}+\left( Z-\frac{2M}{r^2}g' \right)\partial^2_\tau W_{3}-\frac{2}{r^2}\partial_\tau W_0\\
=&\left( \D\partial_\rho^2+\frac{2}{r}\partial_\rho-\frac{2}{r^2}\D \right)\left(\frac{1}{r}\D^{-1}c_2(\tau)+O_4\right)\\
&+\left( Z-\frac{2M}{r^2}g' \right) O_5-\frac{2}{r^2}O_1\\
=&-\frac{2}{r^3} c_2(\tau)+O_6.
\end{align*}
Here $O_6$ is defined by the last equality as

\begin{align*}
O_6:=\left( \D\partial_\rho^2+\frac{2}{r}\partial_\rho-\frac{2}{r^2}\D \right) O_4+\left( Z-\frac{2M}{r^2}g' \right) O_5-\frac{2}{r^2}O_1.
\end{align*}
In particular, $c_2(\tau)\equiv 0$ if $S_W$ and $P_{\even}$ vanish. In general, $c_2(\tau)$ decays to zero if $O_6$ do. Moreover, by using $c_2(\tau)=r^3/2*O_6$, we have

\begin{align*}
W_0(\tau,\rho)=&\frac{c_1(\tau)}{\rho}+O_7,\ \ \  O_7:=2r^3 O_6 W_{0,2}(\rho)+O_3,
\end{align*}
and

\begin{align}\label{W-1t}
\partial_\tau W_{3}=\left(\frac{r^3}{2}O_6\right) \frac{1}{r}\D^{-1}+O_4=:O_8.
\end{align}
From \eqref{def:SW0} and $\underline{L}'=\frac{dg}{dr}\partial_\tau-\partial_\rho$, we have

\begin{align*}
&2\D\partial_\rho W_{3}+\frac{4}{r}\left(1-\frac{M}{r}\right)W_{3}\\
=&S_W+2\frac{dg}{dr}  \partial_\tau W_0-2\partial_\rho W_0+\left(2-\D \frac{dg}{dr}\right)\partial_\tau W_{3}-\frac{4}{r}W_0\\
=&S_W+2 \frac{dg}{dr} O_1-2\partial_\rho O_7+\left(2-\D \frac{dg}{dr}\right)O_8-\frac{4}{r}O_7-\frac{2c_1(\tau)}{r^2}\\
=:&-\frac{2c_1(\tau)}{r^2}+O_9.
\end{align*}
We view 

\begin{equation}\label{equ:ODE_W-1}
2\D\partial_\rho W_{3}+\frac{4}{r}\left(1-\frac{M}{r}\right)W_{3}=-\frac{2c_1(\tau)}{r^2}+O_9
\end{equation}
as an ODE. A homogeneous solution is $\frac{1}{r^2}\D^{-1}$ and a particular solution for the $-\frac{2c_1(\tau)}{r^2}$ source term is $-\frac{c_1(\tau)}{r}$. Therefore,

\begin{align*}
W_{3}(\tau,\rho)=-\frac{c_1(\tau)}{r}+ \frac{d(\tau)}{r^2}\D^{-1}+O_{10}.
\end{align*}
Here

\begin{align*}
O_{10}(\tau,\rho)=\int_{4M}^\rho \frac{1}{2}(\rho')^2  O_9(\tau,\rho')d\rho' \cdot \rho^{-2}\left(1-\frac{2M}{\rho}\right)^{-1},
\end{align*}
and

\begin{align*}
d(\tau)=8M^2(W_{3}(\tau,4M)+\frac{c_1(\tau)}{4M}).
\end{align*}

We would show that $d(\tau)$ decays to zero and $c_1(\tau)$ converges to a constant as $\tau$ goes to infinity. Let

\begin{align}
\bar{W}_0:=&W_0-\frac{c_1(\tau)}{r},\\
\bar{W}_{3}:=&W_{3}+\frac{c_1(\tau)}{r},\\
\W_{\ell=0}:=&(\bar{W}_0,\bar{W}_{3}).
\end{align}
We remark that

\begin{align}
\bar{W}_0=&O_7,\label{bW0_substitute}\\
\bar{W}_{3}=&\frac{d(\tau)}{r^2}\D^{-1}+O_{10}.\label{bW-1_substitute}
\end{align}
As $d(\tau)$ is the coefficient of a singular function, it must be zero if $O_{10}=0$. In general, we use red-shift to obtain integrated decay estimate of $d(\tau)$ and $d'(\tau)$. \\

By direct computation, we have

\begin{align}
&\Box \bar{W}_0+\frac{2M}{r^3}\bar{W}_0=O_{11}, \\
&\Box \bar{W}_{3}-\frac{2M}{r^2} \nabla_{\underline{L}'}\bar{W}_{3}-\frac{2}{r^2}\D\bar{W}_{3}-\frac{2}{r^2}\bar{W}_0=O_{12}.
\end{align}
Here

\begin{align}
O_{11}=&-Z( r^{-1}\cdot O_{c1})-\frac{2M}{r^3}P_{\even}, \label{def:O11}\\
O_{12}=&\left(Z-\frac{2M}{r^2}\frac{dg}{dr} \right) (r^{-1}\cdot O_{c1}).\label{def:O12}
\end{align}

Let $\sigma>0$ be determined soon, $\eta_{rs}(r)$ be the cur-off function defined in \eqref{cutoff_rs} and $Y(\sigma)$ be the red-shit vector field defined in \eqref{vector_rs}. We consider

\begin{equation}
({J}_{\ell=0})_a:=\bigg(T_{ab}[\W_{\ell=0}]-\frac{1}{r^2}|\W_{\ell=0}|^2g_{ab} \bigg)(Y^b+\eta_{rs}T^b).
\end{equation}
Then from \eqref{g_redshift}, on the horizon $r=2M$, 

\begin{align*}
\Div {J}_{\ell=0}\bigg|_{r=2M}=&\frac{\sigma}{2}|\nabla_v\W_{\ell=0}|^2+\frac{1}{2M}|\nabla_R\W_{\ell=0}|^2+\frac{2}{M}\nabla_v\W_{\ell=0}\cdot\nabla_{R}\W_{\ell=0}+\frac{\sigma}{8M^2}|\W_{\ell=0}|^2\\
            &+ \nabla_{T+Y} \W_{\ell=0}\cdot \left(\Box\W_{\ell=0}-\frac{1}{r^2}\W_{\ell=0} \right).
\end{align*}
Note that the coefficient of $\nabla_{\underline{L}'}\W_{\ell=0}$ in $\Box \W_{\ell=0}$ is non-negative. Therefore by choosing $\sigma>>M^{-1}$, there exists a constant $C_{rs}>0$, which may increase from line to line, such that on the horizon

\begin{align*}
\Div {J}_{\ell=0}\bigg|_{r=2M} \geq &\frac{1}{C_{rs}}M^{-3}|(M\partial)^{\leq 1}\W_{\ell=0}|^2 -C_{rs}M(|O_{11}|^2+|O_{12}|^2).
\end{align*}
By continuity, there exists $r_{rs,\ell=0}\in (2M,r_{rs}^+)$ such that in $[2M,r_{rs,\ell=0}]$ we still have

\begin{align*}
\Div {J}_{\ell=0}+C_{rs} M(|O_{11}|^2+|O_{12}|^2)\geq  &\frac{1}{C_{rs}} M^{-3}|(M\partial)^{\leq 1}\W_{\ell=0}|^2.
\end{align*}
Together with \eqref{bW-1_substitute}, we have in $[2M,r_{rs,\ell=0}]$

\begin{align*}
&\Div {J}_{\ell=0}+C_{rs} M(|O_{11}|^2+|O_{12}|^2)+C_{rs} M^{-3}|(M\partial)^{\leq 1}O_{10}|^2\\
\geq & \frac{1}{C_{rs}} \left( M^{-5}\D^{-2}|d'(\tau)|^2+M^{-7}\D^{-4}|d(\tau)|^2\right).
\end{align*}
In $r\in[r_{rs,\ell=0},r_{rs}^+]$, we have

\begin{align*}
|\Div {J}_{\ell=0}|\leq C_{rs}\big( M^{-3}|(M\partial)^{\leq 1}\W_{\ell=0}|^2+M|O_{11}|^2+M|O_{12}|^2 \big).
\end{align*}
Through \eqref{bW0_substitute} and \eqref{bW-1_substitute}, in $r\in [r_{rs,\ell=0},r_{rs}^+]$, $|\Div {J}_{\ell=0}|$ can be further bounded by

\begin{align*}
C_{rs} &\bigg( M^{-5}|d'(\tau)|^2+M^{-7}|d(\tau)|^2+M^{-3}|(M\partial)^{\leq 1}O_7|^2+M^{-3}|(M\partial)^{\leq 1}O_{10}|^2+M|O_{11}|^2+M|O_{12}|^2 \bigg).
\end{align*}
Now we fix a number $\underline{r}\in (2M,r_{rs,\ell=0})$ close to $2M$ such that for any $\tau_2\geq \tau_1$

\begin{align*}
&\frac{1}{2}\int_{D(\tau_1,\tau_2)}\left[ M^{-5}\D^{-2}|d'(\tau)|^2+M^{-7}\D^{-4}|d(\tau)|^2 \right]\cdot\chi_{[\underline{r},  r_{rs,\ell=0}]} dvol\\
\geq &C_{rs}^{2}\ \int_{D(\tau_1,\tau_2)} \left[ M^{-5}|d'(\tau)|^2+M^{-7}|d(\tau)|^2\right]\cdot \chi_{[r_{rs,\ell=0}, r_{rs}^+]}dvol.
\end{align*}
This can be done as the left hand side diverges when $\underline{r}$ goes to $2M$. Then in applying the divergence theorem \ref{div_thm} to  ${J}_{\ell=0}$, we obtain positive $M^{-3}|(M\partial)\W_{\ell=0}|^2$ term in $r\in [2M,\underline{r}]$, positive $|d(\tau)|^2$ $|d'(\tau)|^2$ terms in $r\in [\underline{r},r_{rs,\ell=0}]$ which overcomes the negative ones in $r\in [r_{rs,\ell=0}, r_{rs}^+]$ and error $O_j$ terms. Throwing away the positive boundary term along $\mathcal{H}^+(\tau_1,\tau_2)$, we obtain

\begin{align*}
&\int_{\Sigma_{\tau_2}} J_{\ell=0}\cdot n dvol_3+\int_{\tau_1}^{\tau_2}M^{-2}|d'(\tau)|^2+M^{-4}|d(\tau)|^2 d\tau+M^{-3}\int_{D(\tau_1,\tau_2)} |(M\partial)^{\leq 1}\W_{\ell=0}|^2\cdot\chi_{[ 2M, \underline{r}]} dvol  \\
\lesssim &\int_{\Sigma_{\tau_1}} J_{\ell=0}\cdot n dvol_3+M\int_{D(\tau_1,\tau_2)} \left(|O_{11}|^2+|O_{12}|^2\right)\cdot\chi_{[2M, r_{rs}^+]} dvol\\
&+M^{-3} \int_{D(\tau_1,\tau_2)}\left( |(M\partial)^{\leq 1}O_{10}|^2+|(M\partial)^{\leq 1}O_{7}|^2\right)\cdot\chi_{[ \underline{r},r_{rs}^+]} dvol.
\end{align*}
The boundary term can be estimated from above as

\begin{align*}
\int_{\Sigma_\tau}J_{\ell=0}\cdot n\ dvol_3\lesssim &M^{-2}\int_{\Sigma_\tau}|(M\partial)^{\leq 1} \W_{\ell=0}|^2\cdot\chi_{[2M, \underline{r}]}\ dvol_3+M^{-1}|d'(\tau)|^2+M^{-3}|d(\tau)|^2\\
+&M^{-2}\int_{\Sigma_\tau}\left( |(M\partial)^{\leq 1}O_{10}|^2+|(M\partial)^{\leq 1}O_{7}|^2\right)\cdot\chi_{[ \underline{r}, r_{rs}^+]} dvol_3.
\end{align*}
Here we used \eqref{bW0_substitute}, \eqref{bW-1_substitute} in $r\in [\underline{r},r^+_{rs}]$. Similarly,

\begin{align*}
&\int_{\Sigma_\tau}J_{\ell=0}\cdot n\ dvol_3+C_{rs}M^{-2}\int_{\Sigma_\tau}  |(M\partial)^{\leq 1}O_{10}|^2\cdot\chi_{[ \underline{r}, r_{rs}^+]} dvol_3.\\
\gtrsim &M^{-2}\int_{\Sigma_\tau}|(M\partial)^{\leq 1} \W_{\ell=0}|^2\cdot\chi_{[2M, \underline{r}]}\ dvol_3+M^{-1}|d'(\tau)|^2+M^{-3}|d(\tau)|^2.
\end{align*}
Denote

\begin{equation}
\mathbf{E}(\tau):=M^{-2}\int_{\Sigma_\tau}|(M\partial)^{\leq 1} \W_{\ell=0}|^2\cdot\chi_{[2M, \underline{r}]}\ dvol_3+M^{-1}|d'(\tau)|^2+M^{-3}|d(\tau)|^2.
\end{equation}
Then we obtain for any $\tau_2\geq \tau_1\geq 0$,
\begin{equation}\label{mathbb{E}}
\mathbf{E}(\tau_2)+M^{-1}\int_{\tau_1}^{\tau_2} \mathbf{E}(\tau)d\tau \lesssim \mathbf{E}(\tau_1)+Err_{\ell=0}(\tau_1,\tau_2),
\end{equation}
where

\begin{align*}
Err_{\ell=0}(\tau_1,\tau_2)=&M\int_{D(\tau_1,\tau_2)} \left(|O_{11}|^2+|O_{12}|^2\right)\cdot\chi_{[2M, r_{rs}^+]} dvol\\
+&M^{-3} \int_{D(\tau_1,\tau_2)\}}\left( |(M\partial)^{\leq 1}O_{10}|^2+|(M\partial)^{\leq 1}O_{7}|^2\right)\cdot\chi_{[ \underline{r},r_{rs}^+]} dvol\\
+&M^{-2}\int_{\Sigma_{\tau_1}}\left( |(M\partial)^{\leq 1}O_{10}|^2+|(M\partial)^{\leq 1}O_{7}|^2\right)\cdot\chi_{[ \underline{r}, r_{rs}^+]} dvol_3\\
+&M^{-2}\int_{\Sigma_{\tau_2}} |(M\partial)^{\leq 1}O_{10}|^2 \cdot\chi_{[ \underline{r}, r_{rs}^+]} dvol_3.
\end{align*}
We will show in Appendix \ref{sec:O} that $Err_{\ell=0}(\tau_1,\tau_2)$ is bounded as

\begin{lemma}\label{lem:Err_ell=0}
\begin{align*}
Err_{\ell=0}(\tau_1,\tau_2)\lesssim \left(1+\frac{\tau_1}{M}\right)^{-2+2\delta} \textup{Decay}_{\ell=0}[\hat{P}_{\even}].
\end{align*}
\end{lemma}

Thus, 

\begin{align*}
\mathbf{E}(\tau)+\int_{\tau}^\infty \mathbf{E}(\tau')d\tau'\lesssim \left(1+\frac{\tau}{M}\right)^{-2+2\delta}\cdot\textup{Decay}[\mathbf{E}],
\end{align*}
where $\textup{Decay}[\mathbf{E}]=\mathbf{E}(0)+\textup{Decay}_{\ell=0}[\hat{P}_{\even}]$. In particular, $\W_{\ell=0}$ decays to zero near the horizon. Furthermore, we will show that the coefficient $c_1(\tau)$ converges as $\tau$ goes to infinity. From \eqref{W0t} and $\D\partial_r=\partial_t$ at $r=2M$, we have 

\begin{align*}
0=M\partial_{\tau} S_W-\partial_{\tau} P_{\even}-\partial_{\tau} W_0\bigg|_{r=2M}.
\end{align*}
Also, from the definition of $\bar{W}_0$,

\begin{align*}
0=\partial_{\tau} W_0-\partial_{\tau}\bar{W}_0-\frac{c_1'(\tau)}{2M}\bigg|_{r=2M}.
\end{align*}
Therefore

\begin{align*}
c_1'(\tau)=2M\partial_{\tau}\big(M S_W- P_{\even}-\bar{W}_0\big)(\tau,2M).
\end{align*}

From $\mathbf{E}(\tau)\to 0$, Proposition \eqref{pro:S,l=0} and \eqref{pro:P,l=0}, $S_W,\ P_{\even}$ and $\bar{W}_0$ all decay to zero at $\tau$ goes to infinity. In particular, $c_1(\infty):=c_1(0)+\int_{0}^\infty c_1'(\tau) d\tau $ exists and 

\begin{align*}
c_1(\infty)=&c_1(0)-2M(M S_W-  P_{\even}- \bar{W}_0)(0,2M)\\
           =&-2M(M S_W-  P_{\even}-{W}_0)(0,2M).
\end{align*}

\begin{proof}[proof of theorem \ref{thm:vector,l=0}]
We first show that $\textup{Decay}[\mathbf{E}]\lesssim \textup{Decay}_{\ell=0}[W]$ and start by analyzing $d(\tau)$. From the definition of $d(\tau)$ and $c_1(\tau)$,

\begin{align*}
d(\tau)=-16M^3 \partial_\rho W_{3}(\tau,4M)+8M^2W_{-1}(\tau,4M)+4M^2W_{0}(\tau,4M).
\end{align*}
Hence

\begin{align*}
M^{-3}|d(0)|^2\lesssim M|(M\partial)^{\leq 1} W(0,4M)|^2 \lesssim M^{-2} \int_{\Sigma_0''} |(M\partial)^{\leq 2}W |^2\ dvol_3,
\end{align*}
which is bounded by $\textup{Decay}_{\ell=0}[W]$. We used the spherical symmetry of $W$ to control its sup norm by its $H^1$ norm. For the term $d'(\tau)$, we have

\begin{align*}
d'(\tau)=8M^2\cdot O_8+2M\cdot O_{c1}\lesssim \sup_{[\underline{r},R_{\textup{null}}]} \bigg( M |(M\partial)^{\leq 5}P_{\even}|+M^2|(M\partial)^{\leq 5}S_W| \bigg).
\end{align*}
Therefore

\begin{align*}
M^{-1}|d'(0)|^2\lesssim  F[\f^{\leq 5}P_{\even}](0)+M^2F[\f^{\leq 5}S_W](0),
\end{align*}
which is also bounded by $\textup{Decay}_{\ell=0}[W]$. The last term to bound in $\mathbf{E}(0)$ is the $H^1$ norm of $\W_{\ell=0}=\bar{W}_0,\bar{W}_3$ on $\Sigma_0$. Clearly $\textup{Decay}_{\ell=0}[W]$ includes the $H^1$ norm of $\W_{\ell=0}={W}_0,{W}_3$ on $\Sigma_0$ so if suffices to bound their difference. The difference between $\bar{W}_{3}, \bar{W}_{0}$ and ${W}_{3}, {W}_{0}$ is $c_1(\tau)/r$. Hence

\begin{align*}
M^{-2}\int_{\Sigma_\tau } |(M\partial)^{\leq 1} (\bar{W}_0-W_0)|^2\cdot \chi_{[2M, \underline{r}]}\ dvol_3\approx M |c_1'(\tau)|^2+M^{-1}|c_1(\tau)|^2,
\end{align*}
which can be estimated by the same way as $d(\tau)$. Thus we conclude

\begin{align*}
\textup{Decay}[\mathbf{E}]\lesssim \textup{Decay}_{\ell=0}[W].
\end{align*}
Therefore

\begin{align*}
\mathbf{E}(\tau)+\int_\tau^\infty \mathbf{E}(\tau')d\tau'\lesssim \left(1+\frac{\tau}{M}\right)^{-2+2\delta}\cdot \textup{Decay}_{\ell=0}[W].
\end{align*}

Second, we want to show that Theorem \ref{thm:vector,l=0} holds with $\hat{W}$ replaced by $\bar{W}$. Note that $\mathbf{E}(\tau)$ controls the $H^1$ norm of $\bar{W}$ in $r\in [2M,\underline{r}]$. For the region $r\in [\underline{r},R_{\textup{null}}]$, we use \eqref{bW0_substitute} and \eqref{bW-1_substitute}. The $d(\tau)$ term is also controlled in $\mathbf{E}(\tau)$. To bound $O_7$ and $O_{10}$, we need

\begin{lemma}\label{lem:Err_ell=02}
\begin{align*}
&M^{-3} \int_{D''(\tau,\infty)} \left(|(M\partial)^{\leq 1}O_{10}|^2+|(M\partial)^{\leq 1}O_{7}|^2\right)\cdot\chi_{[\underline{r},R_{\textup{null}}]}\ dvol\\
&+M^{-2}\int_{\Sigma''_{\tau}} \left( |(M\partial)^{\leq 1}O_{10}|^2+|(M\partial)^{\leq 1}O_{7}|^2\right) \cdot\chi_{[\underline{r},R_{\textup{null}}]}\ dvol_3\\
\lesssim &\left(1+\frac{\tau}{M}\right)^{-2+2\delta}\textup{Decay}_{\ell=0}[\f^{\leq 7} \hat{P}_{\even}].
  \end{align*}
\end{lemma}

This will be proved in Appendix \ref{sec:O}. Putting these together, we have

\begin{equation*}
M^{-2}\int_{\Sigma''_\tau}|(M\partial)^{\leq 1} \bar{W} |^2 dvol_3+M^{-3} \int_{D''(\tau,\infty)} |(M\partial)^{\leq 1} \bar{W}|^2 dvol \lesssim \left(1+\frac{\tau}{M}\right)^{-2+2\delta}\cdot\textup{Decay}_{\ell=0}[W].
\end{equation*}

To control $\hat{W}$ instead of $\bar{W}$, we turn to their difference
\begin{align*}
\hat{W}_0-\bar{W}_0=-\hat{W}_{3}-\bar{W}_{3}=\frac{c_1(\tau)-c_1(\infty)}{r}=:\frac{\bar{c}_1(\tau)}{r}.
\end{align*}
Note that

\begin{align*}
\bar{c}_1(\tau)=2M(MS_W-P_{\even}-\bar{W}_0)(\tau,2M).
\end{align*}
And the time derivative of $\bar{c}_1$ can be estimated as

\begin{align*}
\bar{c}_1'(\tau)=c_1'(\tau)=O_{c1}.
\end{align*}
Thus

\begin{align*}
&M^{-2}\int_{\Sigma''_\tau}  |(M\partial)^{\leq 1}(\hat{W}-\bar{W})|^2 dvol_3\approx M \bar{c}_1'(\tau)^2+M^{-1}\bar{c}_1(\tau)^2\\
\lesssim & \sup_{[\underline{r},R_{\textup{null}}]} \left( M|(M\partial)^{\leq 3} P_{\even}|^2+M^3|(M\partial)^{\leq 3} S_W|^2\right)+\bigg(M|\bar{W}_0|^2+M|P_{\even}|^2+M^3 |S_W|^2\bigg)(\tau,2M)\\
\lesssim &\mathbf{E}(\tau)+F[\f^{\leq 3}P_{\even}](\tau)+M^2F[\f^{\leq 3}S_W](\tau).
\end{align*}
Similarly,

\begin{align*}
&M^{-3}\int_{D''(\tau,\infty)}  |(M\partial)^{\leq 1}(\hat{W}_0-\bar{W}_0)|^2 dvol\\
\lesssim &M^{-1}\int_\tau^{\infty} \mathbf{E}(\tau')+F[\f^{\leq 3}P_{\even}](\tau')+M^2F[\f^{\leq 3}S_W](\tau').
\end{align*}
Thus from the estimate for $\mathbf{E}(\tau)$, Proposition \ref{pro:S,l=0} and Proposition \ref{pro:P,l=0},

\begin{align*}
&M^{-2}\int_{\Sigma''_\tau}|(M\partial)^{\leq 1}(\hat{W}-\bar{W}) |^2 dvol_3+M^{-3} \int_{D''(\tau,\infty)} |(M\partial)^{\leq 1} (\hat{W}-\bar{W})|^2 dvol\\
 \lesssim &\left(1+\frac{\tau}{M}\right)^{-2+2\delta}\cdot\textup{Decay}_{\ell=0}[W].
\end{align*}
And the proof is finished.

\end{proof}

\subsection{The $\ell=0$ mode of \eqref{HG} and \eqref{Main_equation}}\label{subsec:sad}
In this subsection we discuss the solution of \eqref{HG} and \eqref{Main_equation} supported on $\ell=0$. Form \cite{Zerilli_2, Martel-Poisson}, one can decompose $h=h_{\ell=0}$ into the mass change $K$, defined in \eqref{def:K}, and a deformation tensor. There exists a unique constant $c$ and a one form $W=W_{\ell=0}$ unique upto the time translation $T_{\flat}$ such that

\begin{align*}
h_{\ell=0}=cK+{}^W\pi,
\end{align*}
From \eqref{pi}, $W_1=\frac{r}{4}\s{tr}h$, where $W_1=W\cdot \D\partial_r$. Hence from \eqref{pi}, the constant $c$ is determined by

\begin{align}\label{def:csad}
\frac{c}{r}=\D^2 h_{rr}-\frac{r}{2}\D\nabla_r \s{tr}h-\frac{1}{2}\left(1-\frac{3M}{r}\right)\s{tr}h.
\end{align}
Suppose $c=0$, then $W_a$ is a solution of \eqref{wave_equation_vector} supported on $\ell=0$ and Theorem \ref{thm:vector,l=0} applies. Nevertheless, $K$ doesn't satisfy the harmonic gauge \eqref{HG} as

\begin{align*}
\Gamma_a[K]dx^a=\frac{1}{R^2}dv.
\end{align*}
To resolve this issue, one can modify $K$ through a deformation tensor ${}^{W^{**}}\pi$. An explicit form of $W'$ is obtained in \cite{Hafner-Hintz-Vasy}:

\begin{align}
W^{**}_a dx^a:=v\cdot W^*_a dx^a-\log \left(\frac{r}{2M}\right)dt+\frac{1}{2}\left(1+\frac{2M}{r}\right)dr.
\end{align}
Define

\begin{align}\label{def:Kstar}
K^*:=K+{}^{W^{**}}\pi.
\end{align}
From direct calculation $K^*$ satisfies \eqref{HG} and hence \eqref{Main_equation}. Note that $K^*$ grows linearly in $\tau$ with leading term being $\tau\cdot {}^{W^*}\pi$. The remaining one form $W'=W-cW^{**}$ satisfies \eqref{HG} and decays to zero if the assumption in Theorem \ref{thm:vector,l=0} holds. In that case we can conclude that $h$ converges to $cK^*+c_1(\infty){}^{W^*}\pi$ with the constant $c_1(\infty)$ given in Theorem \ref{thm:vector,l=0}. Since $tr K^*=0$, the trace of ${}^{W''}\pi$ is the same as $tr h$. However, the value of $P_{\even}$ , defined in \eqref{def:P}, is only determined upto a constant as $P_{\even}[T_\flat]=1$. 

\begin{theorem}\label{thm:sad}
Let $h=h_{\ell=0}$ be a solution of \eqref{HG} and \eqref{Main_equation} and $c$ be the constant defined in \eqref{def:csad}. Let $W=W_{\ell=0}$ be an even one form such that

\begin{align*}
h=cK^*+{}^W\pi.
\end{align*} 

Suppose that $S=tr h$ satisfies $ \big(F[\f^{\leq 7}S]+E^{4-\delta,(0)}_L[\f^{\leq 7}S]\big)(0)<\infty$. Further assume $W$ can be chosen such that $\big(F[\f^{\leq 7} \hat{P}_{\even}]+E^{2-\delta,(0)}_L[\f^{\leq 7} \hat{P}_{\even}]\big)(0)<\infty$. Then $h$ converges to $cK^*+c_1(\infty){}^{W^*}\pi$.
\end{theorem}

\begin{appendix}

\section{wave equation in Schwarzschild spacetime}\label{sec:wave equation}

In this section we record the estimates for the wave equation equation

\begin{align}\label{wave_equation_source}
\Box\psi-{\cblue V(r)}\psi={\cred G},
\end{align}

or

\begin{equation}\label{wave_equation_free}
\Box\psi-V\psi=0,
\end{equation}

where $V(r)$ is a smooth function depending only on $r$.  \\

The study of wave equation in black hole spacetimes using vector field method was initiated by Dafermos and Rodnianski. In \cite{Dafermos-Rodnianski_redshift, Dafermos-Rodnianski_rp} the authors proposed the red-shfit, Morawetz and $r^p$ vector fields as multiplier and proved Proposition \ref{cor:source_high_order} for $m=0$. To get the higher order estimates, one can commute the equation with various vector fields. It's clear that the equation $\eqref{wave_equation_source}$ commutes with Killing vector fields in $\K$. Moreover, the authors discovered that one can also use the red-shift vector as a commutator. Even though the red-shift vector doesn't commute with $\Box$, the lower order terms has special structure that allows the red-shift estimate. Schlue \cite{Schlue} further observed that one can use $L$ and $rL$ as commutator in the $r^p$ estimate. Later Moschidis \cite{Moschidis_rp} applied this technique to a very general class of asymptotically flat spacetimes. These insights yields the higher order estimate ($m\geq 1$) in Proposition \ref{cor:source_high_order}. See \cite{Dafermos-Rodnianski_note} for more detail.

Angelopoulos-Aretakis-Gajic \cite{Angelopoulos-Aretakis-Gajic1} exploited the vector field method by using $r^2L$ as the commutator. There is actually a bad term in the commutation relation and one key observation they made is that the bad term can be absorbed using the Poincar\'e inequality.  This leads to Proposition \ref{cor:free_price}. Furthermore, the authors also show that for the case $\psi=\psi_{\ell=0}$ is supported on $\ell=0$, one can extend the $r^p$ estimate to $p<4$ provide the first Newman-Penrose constant vanishes and leads to faster decay of $\psi$, which is the content of Proposition \ref{cor:free_ell=0}.

The above mentioned results focus on the case without the potential term $V(r)$. Through straight forward modifications, the potential term can be included with suitable assumptions that we list in the following. 

\subsection{Assumptions}



Fix an integer $\ell_0\geq 0$. The assumptions in below are designed for a solution of \eqref{wave_equation_source} supported on $\ell\geq \ell_0$. First, we assume that

\begin{equation}\label{potential_T}
\frac{\ell_0(\ell_0+1)}{r^2}+{\cblue V(r)} \geq 0.
\end{equation}

This implies for all $\psi=\psi_{\ell\geq \ell_0}$, one has $|\nablas \psi|^2+{\cblue V}|\psi|^2\geq_s 0$. Hence the boundary term of the $T$-current is non-negative as

\begin{align*}
\left(T_{ab}[\psi]-\frac{1}{2}V|\psi|^2\right)T^an^b\geq_s 0.
\end{align*}

Second, we assume that there exists a constant $C_{Mor}>0$ and a function $f(r)$ with the following properties:

\begin{itemize}
\item $\frac{d f}{d r}\geq \frac{1}{C_{Mor}}\cdot\frac{M}{r^2}$
\item $f(r)(1-3M/r)\geq 0$
\end{itemize}

And with $\omega=\left(1-\frac{2M}{r}\right)\left( \frac{2f}{r}+\frac{d f}{d r} \right)$,
 
\begin{equation}\label{potential_Mor}
\begin{split}
&f\left(1-\frac{3M}{r}\right)\frac{\ell_0 (\ell_0+1)}{r^3}  +\left(-\frac{1}{4}\Box\omega{\cblue-\frac{1}{2}f\D\frac{d V}{d r}-\frac{M}{r^2}fV}\right)\geq  \frac{1}{C_{Mor}\cdot r^3} \left(  \left(1-\frac{3M}{r}\right)^2 \ell_0 (\ell_0+1)+1   \right).
\end{split}
\end{equation}
The left hand side of \eqref{potential_Mor} appears in the divergence term in the Morawetz current, see subsection \ref{subsection:Morawetz}. This assumption ensures one can construct a current like $J_3$ in subsection \ref{subsection:Morawetz}.\\

Third, we assume that for all $r\geq R_{\textup{null}}$ and $p\in [\delta, 2-\delta]$, 

\begin{equation}\label{potential_rp_1}
\left(1-\frac{p}{2}\right) \cdot \ell_0(\ell_0+1)  +\left({\cblue-\frac{1}{2}r^3\left(\frac{d V}{d r}+\frac{p}{r} V\right)}+(3-p)Mr^{-1}\right)\geq 0.
\end{equation}
Here $R_{\textup{null}}$ is the value of where $\Sigma_\tau$ changes from spacelike to null and $0<\delta<1/10$ is a fixed constant. The left hand side of \eqref{potential_rp_1} appears in the divergence term in the $r^p$ current, see subsection \ref{subsec:rp}. The assumption ensures one can construct a current like $J^p_4$ in subsection \ref{subsec:rp}.\\

Once the potential $V(r)$ satisfies \eqref{potential_T}, \eqref{potential_Mor} and \eqref{potential_rp_1}, Proposition \ref{cor:source_high_order} applies with $m=0$. To get the higher order estimate, we assume that there exists constants $C_m>0$ for all $m\geq 1$ such that

\begin{align}\label{potential_high}
\left| r^{m+2} \frac{d^m V}{dr^m} \right|\leq C_{m}.
\end{align}

\begin{definition}\label{def:potential1}

We denote by $\mathcal{V}(\delta,R_{\textup{null}},\ell_0)$ be the collection of functions $V(r)$ which satisfies \eqref{potential_T}, \eqref{potential_Mor}, \eqref{potential_rp_1} and \eqref{potential_high}.

\end{definition}

We continue to list the requirements for the potential. Fix $k\geq 1$ and denote for $0\leq j\leq k$,

\begin{align*}
{\cblue V^{(j)}(r):=V(r)+\frac{2M}{r^3}-\frac{j(j+1)}{r^2}+\frac{6Mj(j+1)}{r^3}}.
\end{align*}

We assume that for all $1\leq j\leq k$, $r\geq R_{\textup{null}}$  and $p\in [\delta,2-\delta]$, 

\begin{equation}\label{potential_pk_1}
\begin{split}
\left(1-\frac{p}{2}\right)\cdot  \ell(\ell+1){\cblue -\frac{1}{2}r^3\left(\frac{d V^{(j)}}{d r}+\frac{p}{r} V^{(j)}\right)}\geq 0.
\end{split}
\end{equation}

And we require that

\begin{equation}\label{potential_pk_3}
\liminf_{r\to\infty} r^2{\cblue V^{(k)}} \geq -\ell_0 (\ell_0+1),
\end{equation}

Then we assume that there exists a constant $C_{(k)}>0$ such that for all $0\leq j\leq k$

\begin{equation}\label{potential_pk_4}
{\cblue  \left|\left(\frac{d }{d r}\cdot r^2 \right)^jV(r)\right|}\leq C_{(k)}\frac{M^j}{r^2}.
\end{equation}

\begin{definition}\label{def:potential2}
The set $\mathcal{V}(\delta, R_{\textup{null}},\ell_0,k)$ is the subset of $\mathcal{V}(\delta, R_{\textup{null}},\ell_0)$ whose member further satisfies \eqref{potential_pk_1}, \eqref{potential_pk_3} and \eqref{potential_pk_4}.
\end{definition}

The three extra assumptions are needed in using $r^2L$ as commutator $k$ times. Let $\psi$ be a solution of \eqref{wave_equation_free}. Then $r^{-1}(r^2L)^k(r \psi)$ behaves like a solution of \eqref{wave_equation_free} with potential $V^{(k)}(r)$ in the $r^p$ estimate. These assumptions guarantee the boundary and the bulk terms are non-negative. Applying divergence theorem, one deduces that $E_L^{p,(j)}[\psi](\tau)$ decays like $\tau^{-2+p}\cdot E_L^{2,(j)}[\psi](0)$ for all $0\leq j\leq k$. Then using $E^{2,(j)}_L\lesssim E^{0,(j+1)}_L$, which can be proved by Hardy inequality Lemma \ref{lem:Hardy}, $E_L^{p,(0)}[\psi](\tau)$ decays like $\tau^{-2(k+1)}$.

\begin{remark}\label{rem:0}

It's easy to verify that the zero function $0$ satisfies all the assumptions for $\mathcal{V}(\delta,R_{\textup{null}},\ell_0)$ except \eqref{potential_Mor}, the Morawetz assumption. For such potential ($V=0$), the Morawetz estimate was proved in \cite{Dafermos-Rodnianski_redshift}. Moreover, $0$ clearly satisfies \eqref{potential_pk_1}, \eqref{potential_pk_3} and \eqref{potential_pk_4} provided $\ell_0\geq k$. This is the case studied in \cite{Angelopoulos-Aretakis-Gajic1}.

\end{remark}

\subsection{decay estimate}

Recall that the initial norms of $\psi$ are defined by

\begin{equation*}
I^{(k)}[\psi]:=\left( F[\mathcal{K}^{\leq 2k+2}\psi]+\sum_{\substack{i_1+i_2\leq 2k+1 \\ i_2\leq k}} E_L^{2-\delta,(i_2)}[\mathcal{K}^{i_1}\tilde{\psi}]\right)(0).
\end{equation*}

The spacetime norm of the source term $G$, for any $\bar{\delta}\geq \delta$, is defined by 

\begin{equation*}
\textup{I}_{\textup{source},\bar{\delta}}[G]:=\sup_{p\in [\delta,2-\delta]}\sup_{\tau_2\geq \tau_1\geq 0} \left(1+\frac{\tau_1}{M}\right)^{2-p-\bar{\delta}} \textup{E}^p_{\textup{source}}[\K^{\leq 2} G](\tau_1,\tau_2).
\end{equation*}

\begin{proposition}\label{cor:source_high_order}
Let $\psi$ be a solution of wave equation \eqref{wave_equation_source} with source term $G$ and $m\geq 0$ be a fixed integer. Suppose $\psi=\psi_{\ell\geq \ell_0}$ and $V\in\mathcal{V}(\delta, R_{\textup{null}},\ell_0)$. Suppose that $I^{(0)}[\f^{\leq m}\psi]$ and $\textup{I}_{\textup{source},\bar{\delta}}[\f^{\leq m}G]$ are finite, then for any $p\in [\delta,2-\bar{\delta}]$ and $\tau\geq 0$ we have

\begin{equation}\label{source_rp_decay}
\begin{split}
&F[\f^{\leq m}\psi](\tau)+E^{p,(0)}_L[\f^{\leq m}\tilde{\psi}](\tau)+\int_\tau^\infty B[\f^{\leq m}\psi](\tau')+E^{p-1,(0)}_{L,\nablas}[\f^{\leq m}\tilde{\psi}](\tau') d\tau'\\
\lesssim_{m,V,R_{\textup{null}}} &\left(1+\frac{\tau}{M}\right)^{-2+p+\bar{\delta}} \left(I^{(0)}[\f^{\leq m}\psi]+\textup{I}_{\textup{source},\bar{\delta}}[\f^{\leq m}G]\right).
\end{split}
\end{equation}

\end{proposition}

\begin{proposition}\label{cor:free_price}

Let $\psi=\psi_{\ell\geq \ell_0}$ be a solution of \eqref{wave_equation_free}. Assume that $V\in\mathcal{V}(\delta, R_{\textup{null}},\ell_0,k)$. We have for any $p\in [\delta,2-\delta]$ and any $\tau\geq 0$, 

\begin{equation}\label{free_price_decay}
\begin{split}
&F[\psi](\tau)+E^{p,(0)}_L[\tilde{\psi}](\tau)+\int_\tau^\infty B[\psi](\tau')+E^{p-1,(0)}_{L,\nablas}[\tilde{\psi}](\tau') d\tau'\\
\lesssim_{V,R_{\textup{null}}} &\left(1+\frac{\tau}{M}\right)^{-2(k+1)+p+(1+2k)\delta}I^{(k)}[\psi].
\end{split}
\end{equation}

\end{proposition}

As $k=1$ or $2$, this is part of Proposition 7.5 in \cite{Angelopoulos-Aretakis-Gajic1} with simple modification to include the potential. The authors in \cite{Angelopoulos-Aretakis-Gajic1} work with asymptotically flat space time with suitable behavior near null infinity. In the Schwarzschild spacetiem, one can continue commuting with $r^2L$ as suggested in \cite{Angelopoulos-Aretakis-Gajic1}.


Since the trapped geodesic can only be approached by a solutions of wave equation with higher and higher frequency. Once we restrict to a fixed mode $\ell=0$ (or any other modes), the loss of derivative can be avoided. In particular, we define the norm below.

\begin{align}\label{def:I0}
\textup{I}_{\textup{source},\ell=0,\bar{\delta}}[G]:=\sup_{\tau\geq 0}\sup_{p\in [\delta,2-\bar{\delta}]} \left(1+\frac{\tau}{M}\right)^{2-p-\bar{\delta}}\ M\int_{D(\tau,\infty)} s^{p+1}|G|^2 dvol.
\end{align}

\begin{proposition}\label{cor:free_ell=0}
Fix $m\geq 0$.\\

(1) Let $\psi=\psi_{\ell=0}$ be a solution of wave equation $\Box\psi=0$ supported on the mode $\ell=0$. Suppose that $F[\f^{\leq m}\psi](0)+E^{4-\delta,(0)}[\f^{\leq m}\psi](0)$ is finite. Then for any $p\in [\delta,4-\delta]$ and $\tau\geq 0$,

\begin{equation}\label{free_ell=0_decay}
\begin{split}
&F[\f^{\leq m}\psi](\tau)+E^{4-\delta,(0)}_L[\f^{\leq m}\tilde{\psi}](\tau)+\int_\tau^\infty \bar{B}[\f^{\leq m}\psi](\tau')+E^{p-1,(0)}_{L,\nablas}[\f^{\leq m}\tilde{\psi}](\tau') d\tau'\\
\lesssim_{m,V,R_{\textup{null}}} &\left(1+\frac{\tau}{M}\right)^{-4+p+{\delta}} \left(F[\f^{\leq m}\psi](0)+E^{4-\delta,(0)}[\f^{\leq m}\psi](0)\right).
\end{split}
\end{equation}

(2) Let $\psi=\psi_{\ell=0}$ be a solution of wave equation $\Box\psi=G$ supported on the mode $\ell=0$. Suppose that $F[\f^{\leq m}\psi](0)+E^{2-\delta,(0)}[\f^{\leq m}\psi](0)$ and $\textup{I}_{\textup{source},\ell=0,\bar{\delta}}[\f^{\leq m}G]$ are finite. Then for any $p\in [\delta,2-\bar{\delta}]$ and $\tau\geq 0$,

\begin{equation}\label{source_ell=0_decay}
\begin{split}
&F[\f^{\leq m}\psi](\tau)+E^{4-\delta,(0)}_L[\f^{\leq m}\tilde{\psi}](\tau)+\int_\tau^\infty \bar{B}[\f^{\leq m}\psi](\tau')+E^{p-1,(0)}_{L,\nablas}[\f^{\leq m}\tilde{\psi}](\tau') d\tau'\\
\lesssim_{m,V,R_{\textup{null}}} &\left(1+\frac{\tau}{M}\right)^{-2+p+\bar{\delta}} \left(F[\f^{\leq m}\psi](0)+E^{2-\delta,(0)}[\f^{\leq m}\psi](0)+\textup{I}_{\textup{source},\ell=0,\bar{\delta}}[\f^{\leq m}G]\right).
\end{split}
\end{equation}

\end{proposition}

 Proposition \ref{cor:free_ell=0} (1), as $m=0$, is part of Theorem 1.5 in \cite{Angelopoulos-Aretakis-Gajic1}. Note that our assumption $E^{4-\delta,(0)}_L<\infty$ implies the Newman-Penrose constant vanishes. The case $m\geq 1$ can be proved by using $\f$ as commutator and is closely related to Theorem 1.7 in \cite{Angelopoulos-Aretakis-Gajic1}. Proposition \ref{cor:free_ell=0} (2) can be almost derived from Proposition \eqref{cor:source_high_order} except that we have the non-degenerate norm $\bar{B}$ and $\textup{I}_{\textup{source},\ell=0,\bar{\delta}}[\f^{\leq m}G]$ because we focus on a fixed mode $\ell=0$. 

\section{$O_j$ terms in section \ref{sec:ell=0}}\label{sec:O}

In this appendix, we give estimates on the bounds of the $O_j$ terms in in section \ref{sec:ell=0}. Let $m\geq 1$ be a fixed integer. From the definition of $O_1$ \eqref{def:O1}, we have

\begin{align*}
|(M\partial)^{\leq m} O_1|\lesssim_m &  M^{-1}|(M\partial)^{\leq m+1}P_{\even}|+|(M\partial)^{\leq m+1}S_W|.
\end{align*}

From the definition of $O_2$ \eqref{def:O2}, we have

\begin{align*}
|(M\partial)^{\leq m} O_2|\lesssim_m  M^{-2}|(M\partial)^{\leq m+2}P_{\even}|+M^{-1}|(M\partial)^{\leq m+2}S_W|.
\end{align*}

The definition of $O_3$ involves integration and is singular at $r=2M$ because of $W_{0,2}(\rho)$. Therefore we from now on restrict ourselve in the region $r\in [\underline{r},R_{\textup{null}}]$. Then as $r\in [\underline{r},R_{\textup{null}}]$,

\begin{align*}
|(M\partial)^{\leq m}  O_3|\lesssim_m \sup_{[\underline{r},R_{\textup{null}}]} \bigg( |(M\partial)^{\leq m+2}P_{\even}|+M |(M\partial)^{\leq m+2}S_W| \bigg).
\end{align*}

The constant in above inequality depends on $\underline{r}$ and $R_{\textup{null}}$ and blows up as $\underline{r}$ goes to $2M$ or $R_{\textup{null}}$ goes to infinity. However , $\underline{r}$ and $R_{\textup{null}}$ are already fixed.\\

From the definition of $O_{c1}$ and $O_{c2}$, \eqref{def:Oc1} and \eqref{def:Oc2}, we have 

\begin{align*}
|(M\partial)^{\leq m}O_{c1}|\lesssim_m & \bigg( |(M\partial)^{\leq m+2}P_{\even}|+M|(M\partial)^{\leq m+ 2}S_W| \bigg)\bigg|_{r=4M}\\
|(M\partial)^{\leq m}O_{c2}|\lesssim_m & \bigg( M^{-1}|(M\partial)^{\leq m+2}P_{\even}|+|(M\partial)^{\leq m+ 2}S_W| \bigg)\bigg|_{r=4M}.
\end{align*}

Repeating the same argument, we have

\begin{align*}
\sup_{[\underline{r},R_{\textup{null}}]}|(M\partial)^{\leq m}  O_4|\lesssim_m &\sup_{[\underline{r},R_{\textup{null}}]} \bigg( M^{-1}|(M\partial)^{\leq m+3}P_{\even}|+|(M\partial)^{\leq m+3}S_W| \bigg),\\
\sup_{[\underline{r},R_{\textup{null}}]} |(M\partial)^{\leq m} O_5|\lesssim_m & \sup_{[\underline{r},R_{\textup{null}}]} \bigg( M^{-2}|(M\partial)^{\leq m+4}P_{\even}|+M^{-1}|(M\partial)^{\leq m+4}S_W| \bigg),\\
\sup_{[\underline{r},R_{\textup{null}}]} |(M\partial)^{\leq m} O_6|\lesssim_m & \sup_{[\underline{r},R_{\textup{null}}]} \bigg( M^{-3}|(M\partial)^{\leq m+5}P_{\even}|+M^{-2}|(M\partial)^{\leq m+5}S_W| \bigg).
\end{align*}

and

\begin{align*}
\sup_{[\underline{r},R_{\textup{null}}]}|(M\partial)^{\leq m}O_7|\lesssim_m &\sup_{[\underline{r},R_{\textup{null}}]} \bigg( |(M\partial)^{\leq m+5}P_{\even}|+M |(M\partial)^{\leq m+5}S_W| \bigg),\\
\sup_{[\underline{r},R_{\textup{null}}]}|(M\partial)^{\leq m}O_8|\lesssim_m &\sup_{[\underline{r},R_{\textup{null}}]} \bigg( M^{-1}|(M\partial)^{\leq m+5}P_{\even}|+|(M\partial)^{\leq m+5}S_W| \bigg),\\
\sup_{[\underline{r},R_{\textup{null}}]}|(M\partial)^{\leq m}O_9|\lesssim_m &\sup_{[\underline{r},R_{\textup{null}}]} \bigg( M^{-1}|(M\partial)^{\leq m+6}P_{\even}|+|(M\partial)^{\leq m+6}S_W| \bigg),\\
\sup_{[\underline{r},R_{\textup{null}}]}|(M\partial)^{\leq m}O_{10}|\lesssim_m &\sup_{[\underline{r},R_{\textup{null}}]} \bigg( |(M\partial)^{\leq m+6}P_{\even}|+M|(M\partial)^{\leq m+6}S_W| \bigg).
\end{align*}

The terms $O_{11}$ and $O_{12}$ are smooth upto the horizon as they depend on $O_{c1}$ and $P_{\even}$ in \eqref{def:O11} and \eqref{def:O12}.

\begin{align*}
|(M\partial)^{\leq m} O_{11}|\lesssim_m & \bigg( M^{-2}|(M\partial)^{\leq 3}P_{\even}|+M^{-1}|(M\partial)^{\leq 3}S_W| \bigg)\bigg|_{r=4M}+M^{-2}|P_{\even}|,\\
|(M\partial)^{\leq m} O_{12}|\lesssim_m & \bigg( M^{-2}|(M\partial)^{\leq 3}P_{\even}|+M^{-1}|(M\partial)^{\leq 3}S_W| \bigg)\bigg|_{r=4M},\\
\end{align*}

Now we are ready to prove Lemma \ref{lem:Err_ell=0} and \ref{lem:Err_ell=02}.

\begin{proof}[proof of Lemma \ref{lem:Err_ell=0} and \ref{lem:Err_ell=02}]
From the above discussion, $Err_{\ell=0}(\tau_1,\tau_2)$ is bounded by

\begin{align*}
&\bigg(F[\f^{\leq 7}\hat{P}_{\even}]+M^2F[\f^{\leq 7}S_W]\bigg)(\tau_1)+\bigg(F[\f^{\leq 7}\hat{P}_{\even}]+M^2F[\f^{\leq 7}S_W]\bigg)(\tau_2)\\
&+M^{-1}\int_{\tau_1}^{\tau_2} \bigg(\bar{B} [\f^{\leq 7}\hat{P}_{\even}]+M^2 \bar{B} [\f^{\leq 7}S_W]\bigg)(\tau)d\tau.
\end{align*}

Therefore the claim follows Proposition \ref{pro:S,l=0} and \ref{equ:P,l=0}. The proof of Lemma \ref{lem:Err_ell=02} is similar.

\end{proof}

\section{Equation for $\h$}\label{sec:equation_metirc}

In this appendix we derive \eqref{equation_far}, the wave equation for $\h$. We follow closely the computation in \cite{Berndtson} except that we adapt a different decomposition which is similar to the one in \cite{Barack-Lousto}.\\

Recall that for a function $\psi(t,r,\theta,\phi)$, we decompose $\psi$ as

\begin{align*}
\psi=\sum_{\ell=0}^{\infty}\sum_{|m|\leq \ell} \psi_{m\ell}(t,r)Y_{m\ell}(\theta,\phi).
\end{align*}

Similarly, for any \textbf{even} spherical one form $\xi_A$ and \textbf{even} spherical symmetric  traceless two tensor $\Xi_{AB}$, we decompose them as

\begin{align*}
\xi_A=&\sum_{\ell=1}^\infty \sum_{|m|\leq \ell} \xi_{m\ell}(t,r)\nablas_A Y_{m\ell}(\theta,\phi),\\
\Xi_{AB}=&\sum_{\ell=2}^\infty\sum_{|m|\leq \ell} \Xi_{m\ell}(t,r)\hat{\nablas}_{AB} Y_{m\ell}(\theta,\phi).
\end{align*}

We also define the operator ${}^0\Box_\ell$ on the scalar functions depending on $\{t,r\}$ as

\begin{align}
{}^0\Box_\ell=-\D^{-1}\frac{\partial^2}{\partial t^2}+\D \frac{\partial^2}{\partial r^2}+\frac{2}{r}\left(1-\frac{M}{r}\right)\frac{\partial}{\partial r}-\frac{\ell(\ell+1)}{r^2}.
\end{align}

Then we have

\begin{align*}
\Box f=\sum_{\ell=0}^\infty\sum_{|m|\leq \ell}\left({}^0\Box_\ell f_{m\ell}(t,r)\right)\cdot Y_{m\ell}(\theta,\phi).
\end{align*}

Similarly, we define

\begin{align}
{}^1\Box_\ell=r\cdot{}^0\Box_\ell\cdot r^{-1}+\frac{1}{r^2},\\
{}^2\Box_\ell=r^2\cdot{}^0\Box_\ell\cdot r^{-2}+\frac{4}{r^2}
\end{align}

in order to have

\begin{align*}
{}^1\Box \xi_A=&\sum_{\ell=1}^\infty\sum_{|m|\leq \ell}\left({}^1\Box_\ell \xi_{m\ell}(t,r)\right)\cdot\nablas_A Y_{m\ell}(\theta,\phi),\\
{}^2\Box \Xi_{AB}=&\sum_{\ell=0}^\infty\sum_{|m|\leq \ell}\left({}^2\Box_\ell \Xi_{m\ell}(t,r)\right)\cdot\hat{\nablas}_{AB} Y_{m\ell}(\theta,\phi).
\end{align*}

For an even perturbation $h_{ab}$, the decomposition \cite{Berndtson} in in the following:

\begin{align}
h_{tt}=&\D\sum_{\ell=0}^\infty \sum_{|m|\leq \ell}  H_{0,m\ell}\cdot Y_{m\ell},\\
h_{tr}=&\sum_{\ell=0}^\infty \sum_{|m|\leq \ell} H_{1,m\ell}\cdot Y_{m\ell},\\
h_{rr}=&\D^{-1}\sum_{\ell=0}^\infty \sum_{|m|\leq \ell} H_{2,m\ell}\cdot Y_{m\ell},\\
\s{tr}h=&2\sum_{\ell=0}^\infty \sum_{|m|\leq \ell} K_{m\ell}\cdot Y_{m\ell}.
\end{align}

\begin{align}
h_{tA}=&\sum_{\ell=1}^\infty\sum_{|m|\leq \ell}  h_{0,m\ell}\cdot \nablas_A Y_{m\ell},\\
h_{rA}=&\sum_{\ell=1}^\infty\sum_{|m|\leq \ell}  h_{1,m\ell}\cdot \nablas_A Y_{m\ell},\\
\hat{h}_{AB}=&2r^2 \sum_{\ell=2}^\infty\sum_{|m|\leq \ell}  G_{m\ell}\cdot \hat{\nablas}_{AB} Y_{m\ell}
\end{align}

From now on we fix $\ell\geq 0$, $|m|\leq \ell$. Let

\begin{align}
\LW_{ab}=&\Box h_{ab}+2R_{acbd}h_{cd},\\
\HG_b=&\nabla^a\left( h_{ab}-\frac{1}{2}(\text{tr}h)g_{ab} \right).
\end{align}

Recall that $\lambda=\lambda(\ell)$ is defined in \eqref{def:lambda_ell}. From \cite{Berndtson} (3.1-3.7) or direct computation,

\begin{equation}\label{MainEquation_1}
\tag{ME1}
\begin{split}
\LW_{tt,m\ell}=\D{}^0\Box_\ell H_{0,m\ell}-\frac{2M^2}{r^4}H_{0,m\ell}+\frac{4M}{r^2}\partial_t H_{1,m\ell}+\frac{2M(3M-2r)}{r^4}H_{2,m\ell}+\frac{4M(r-2M)}{r^4}K_{m\ell}.
\end{split}
\end{equation}

\begin{equation}\label{MainEquation_2}
\tag{ME2}
\begin{split}
\LW_{rr,m\ell}=\D^{-1}{}^0\Box_\ell H_{2,m\ell}+\frac{2M}{r^2}\frac{(3M-2r)}{(r-2M)^2}H_{0,m\ell}+\frac{8\lambda+8}{r^3}h_{1,m\ell}+\frac{4M}{(r-2M)^2}\partial_t H_{1,m\ell}\\
-\frac{(4r^2-16Mr+18M^2)}{r^2(r-2M)^2}H_{2,m\ell}+\frac{4(r-3M)}{r^2(r-2M)}K_{m\ell}.
\end{split}
\end{equation}

\begin{equation}\label{MainEquation_3}
\tag{ME3}
\begin{split}
\LW_{tr,m\ell}={}^0\Box_\ell H_{1,m\ell}+\frac{4(1+\lambda)}{r^3}h_{0,m\ell}+\frac{2M}{r^2-2Mr}\partial_t H_{0,m\ell}+\frac{2(r^2-2Mr+2M^2)}{r^3(2M-r)}H_{1,m\ell}+\frac{2M}{r^2-2Mr}\partial_t H_{2,m\ell}.
\end{split}
\end{equation}

\begin{equation}\label{MainEquation_4}
\tag{ME4}
\begin{split}
\LW_{tA,m\ell}={}^0\Box_\ell h_{0,m\ell}+\frac{4M}{r^3}h_{0,m\ell}+\frac{2(M-r)}{r^2}\partial_r h_{0,m\ell}+\frac{2M}{r^2}\partial_t h_{1,m\ell}+\frac{2(r-2M)}{r^2}H_{1,m\ell}.
\end{split}
\end{equation}

\begin{equation}\label{MainEquation_5}
\tag{ME5}
\begin{split}
\LW_{rA,m\ell}={}^0\Box_\ell h_{1,m\ell}+\frac{4\lambda}{r}G_{m\ell}+\frac{2M}{(r-2M)^2}\partial_t h_{0,m\ell}+\frac{8M-4r}{r^3}h_{1,m\ell}+\frac{2}{r}H_{2,m\ell}-\frac{2}{r}K_{m\ell}+\frac{6M-2r}{r^2}\partial_r h_{1,m\ell}.
\end{split}
\end{equation}

\begin{equation}\label{MainEquation_6}
\tag{ME6}
\begin{split}
r^2\slashed{tr} \LW_{m\ell}=r^2\left(  {}^0\Box_\ell 2K_{m\ell}+\frac{4M}{r^3}H_{0,m\ell}+\frac{(8\lambda+8)(2M-r)}{r^4}h_{1,m\ell}+\frac{4(r-3M)}{r^3}H_{2,m\ell}+\frac{4r-16M}{r^3}K_{m\ell}\right).
\end{split}
\end{equation}

\begin{equation}\label{MainEquation_7}
\tag{ME7}
\begin{split}
\hat{\LW}_{AB,m\ell}=2r^2\cdot {}^0\Box_\ell G_{m\ell}+4G_{m\ell}+\frac{4r-8M}{r^2}h_{1,m\ell}.
\end{split}
\end{equation}

Also from \cite{Berndtson} (3.8-3.10) or direct computation, the Harmonic gauge condition reads

\begin{equation}\label{HarmonicGauge_1}
\tag{HG1}
\begin{split}
{\HG}_{t,m\ell}=\frac{-2\lambda-2}{r^2}h_{0,m\ell}-\frac{1}{2}\partial_t H_{0,m\ell}-\frac{2(M-r)}{r^2}H_{1,m\ell}-\frac{1}{2}\partial_t H_{2,m\ell}-\partial_t K_{m\ell}+\D\partial_r H_{1,m\ell}.
\end{split}
\end{equation}

\begin{equation}\label{HarmonicGauge_2}
\tag{HG2}
\begin{split}
{\HG}_{r,m\ell}=-\frac{M}{2Mr-r^2}H_{0,m\ell}-\frac{2\lambda+2}{r^2}h_{1,m\ell}+\frac{r}{2M-r}\partial_t H_{1,m\ell}+\frac{3M-2r}{2Mr-r^2}H_{2,m\ell}\\
-\frac{2}{r}K_{m\ell}+\frac{1}{2}\partial_r H_{0,m\ell}+\frac{1}{2}\partial_r H_{2,m\ell}-\partial_r K_{m\ell}.
\end{split}
\end{equation}

\begin{equation}\label{HarmonicGauge_3}
\tag{HG3}
\begin{split}
{\HG}_{A,m\ell}=-2\lambda G_{m\ell}+\frac{r}{2M-r}\partial_t h_{0,m\ell}+\frac{1}{2}H_{0,m\ell}-\frac{2M-2r}{r^2}h_{1,m\ell}-\frac{1}{2}H_{2,m\ell}+\D\partial_r h_{1,m\ell}.
\end{split}
\end{equation}

Now we start to derive the equation for $\h$. From the definition of $\h$, we have

\begin{align*}
\h_{1,m\ell}=&\frac{1}{2}\D\left(H_{0,m\ell}+H_{2,m\ell}\right)+\frac{2M}{r}K_{m\ell},\\
\h_{2,m\ell}=&K_{m\ell},\\
\h_{3,m\ell}=&\frac{1}{2}(H_{0,m\ell}-H_{2,m\ell}),\\
\h_{6,m\ell}=&\D H_{1,m\ell}+\frac{2M}{r}K_{m\ell}.
\end{align*}

\begin{align*}
\h_{4,m\ell}=&\D h_{1,m\ell},\\
\h_{7,m\ell}=& h_{0,m\ell}.
\end{align*}

And

\begin{align*}
\h_{5,m\ell}=\sqrt{2}r^2 G_{m\ell}.
\end{align*}

The equation for $\h_1$ can be derived by 
\begin{align*}
&{}^0\Box_\ell \h_{1,m\ell}-\left(\frac{1}{2}\eqref{MainEquation_1}+\frac{M}{r^3}\eqref{MainEquation_6}+\frac{1}{2}\D^2\eqref{MainEquation_2}+\frac{4M(r-2M)}{r^3}\eqref{HarmonicGauge_2}\right)\\
=&\frac{2}{r^2}(1-6s^{-1})\h_{1,m\ell}-\frac{2}{r^2}(1-8s^{-1}+10s^{-2})\h_{2,m\ell	}-\frac{2}{r^2}(1-10s^{-1}+20s^{-2})\h_{3,m\ell}\\
-&\frac{4\lambda+4}{r^3}(1-6s^{-1})\h_{4,m\ell}.
\end{align*}

Together with Lemma \ref{lem:d_mode}, the equation of $\h_1$ follows by summing over all $m\ell$.\\

The equation for $\h_2$ can be derived by 
\begin{align*}
&{}^0\Box_\ell \h_{2,m\ell}-\frac{1}{2r^2}\eqref{MainEquation_6}\\
=&-\frac{2}{r^2}\h_{1,m\ell}+\frac{2}{r^2}(1-2s^{-1})\h_{2,m\ell}+\frac{2}{r^2}(1-4s^{-1})\h_{3,m\ell}+\frac{4\lambda+4}{r^3}\h_{4,m\ell}.
\end{align*}

The equation for $\h_3$ can be derived by

\begin{align*}
&{}^0\Box_\ell \h_{3,m\ell}-\left(\frac{1}{2}\D^{-1}\eqref{MainEquation_1}-\frac{1}{2}\D\eqref{MainEquation_2}\right)\\
=&-\frac{2}{r^2}\h_{1,m\ell}+\frac{2}{r^2}(1-2s^{-1})\h_{2,m\ell}+\frac{2}{r^2}(1-4s^{-1})\h_{3,m\ell}+\frac{4\lambda+4}{r^3}\h_{4,m\ell}.
\end{align*}

The equation for $\h_4$ can be derived by

\begin{align*}
&{}^1\Box_\ell \h_{4,m\ell}-\left( \D\eqref{MainEquation_5}+\frac{2M}{r^2}\eqref{HarmonicGauge_3}\right)\\
=&\frac{2}{r^2}(5/2 - 9 s^{-1})\h_{4,m\ell} - \frac{2}{r} \h_{1,m\ell} + \frac{2}{r}\h_{2,m\ell} + 
 \frac{2}{r}(1 - 3 s^{-1})\h_{3,m\ell} -\frac{2\sqrt{2}\lambda}{r^3}(1-3s^{-1})\h_{5,m\ell}.
\end{align*}

The equation for $\h_5$ can be derived by

\begin{align*}
&{}^2\Box_\ell \h_{5,m\ell}-\frac{1}{\sqrt{2}}\eqref{MainEquation_7}=\frac{2}{r^2} \h_{5,m\ell} -\frac{2\sqrt{2}}{r} \h_{4,m\ell}.
\end{align*}

The equation for $\h_6$ can be derived by

\begin{align*}
&{}^0\Box_\ell \h_{6,m\ell}-\left(\D\eqref{MainEquation_2}+\frac{M}{r^2}\eqref{MainEquation_6}+\frac{4M(r-2M)}{r^2}\eqref{HarmonicGauge_2}\right)\\
=&\frac{2}{r^2}(-4s^{-1})\h_{1,m\ell} + \frac{2}{r^2}(4s^{-1} - 6s^{-2})\h_{2,m\ell} + 
 \frac{2}{r^2} (6 s^{-1} - 16 s^{-2})\h_{3,m\ell} + \frac{2}{r^2} (1 - 2 s^{-1})\h_{6,m\ell} \\
 +&\frac{2}{r^3}(s^{-1})(8 \lambda + 8) \h_{4,m\ell} - \frac{4\lambda+4}{r^3}(1-2s^{-1})\h_{7,m\ell}-\frac{4}{r^2}(1 - 2 s^{-1})^{-1} s^{-1} \nabla_L( r(\h_{1,m\ell}-\h_{6,m\ell}) ).
\end{align*}

Finally, the equation for $\h_{7}$ can be derived by

\begin{align*}
&{}^1\Box_\ell \h_{7,m\ell}-\left( \eqref{MainEquation_4}+\frac{2M}{r^2}\eqref{MainEquation_3}\right)\\
=& \frac{2}{r^2}(-2s^{-1})\h_{4,m\ell} + \frac{2}{r^2} (1/2 - 3 s^{-1})\h_{7,m\ell} + 
 \frac{2}{r} (2 s^{-1})\h_{2,m\ell}  -\frac{2}{r} (s^{-1}) \h_{3,m\ell} + 
 \frac{2\sqrt{2}\lambda}{r^3}s^{-1}\h_{5,m\ell}\\
  -& \frac{2}{r} \h_{6,m\ell} + 
 \frac{2}{r^2}(1-2s^{-1})^{-1} (-s^{-1}) r\cdot\nabla_L(\h_{4,m\ell}-\h_{7,m\ell}).
\end{align*}

\section{Equation for gauge change}\label{sec:equation_vector}

Now we start to derive the equation for $W_0,$ $W_1$ and $W_2$ \eqref{equ:W_0}, \eqref{equ:W_1} and \eqref{equ:W_2}. It is equivalent to \cite{Berndtson} except we view $W_2$ as a section of $\mathcal{L}(-1)$ and $W_1$ is defined differently. We will follow closely the approach in \cite{Berndtson}. In this section we always assume $\ell\geq 2$. In particular, all $\done,\donest,\dtwo$ and $\dtwost$ are invertible. Furthermore

\begin{align}
\slashed{\Delta}^{-1} \sum_{|m|\leq \ell, \ell\geq 1} f_{m\ell}Y_{m\ell}=\sum_{|m|\leq \ell, \ell\geq 1}\left(-\frac{r^2}{2\lambda+2}f_{m\ell}Y_{m\ell}\right),\\
\slashed{\Delta}_Z^{-1} \sum_{|m|\leq \ell, \ell\geq 1} f_{m\ell}Y_{m\ell}=\sum_{|m|\leq \ell, \ell\geq 1}\left(- \frac{r^2}{2\lambda+6M/r}f_{m\ell}Y_{m\ell}\right).
\end{align}
Together with Lemma \ref{lem:d_mode}, the mode of the Zerilli quantity $\psi$ is

\begin{equation}\label{def:psi_mode}
\begin{split}
\psi_{m\ell}=-&r^2\D (\lambda+1)^{-1}(\lambda+3M/r)^{-1}\partial_r K_{m\ell}\\
+&r(\lambda+1)^{-1}K_{m\ell}+r\D(\lambda+1)^{-1}(\lambda+3M/r)^{-1} H_{2,m\ell}\\
-&2\D(\lambda+3M/r)^{-1}h_{1,m\ell}+2rG_{m\ell}.
\end{split}
\end{equation}
This is (3.34) in \cite{Berndtson}. The Zerilli equation for $\psi$ is

\begin{align*}
{}^0\Box_\ell (r^{-1}\psi_{m\ell})+\frac{2M}{r^3}\frac{(2\lambda+3)(2\lambda+3M/r)}{(\lambda+3M/r)^2}\cdot(r^{-1}\psi_{m\ell})=0.
\end{align*}
This equation as well as the definition of $\psi_Z$ was first derived by Zerilli \cite{Zerilli}. It can also be verified directly through 

\begin{align*}
0=&-\frac{r  (6 M+\lambda  r)}{(\lambda +1) (3 M+\lambda  r)^2}\partial_t\eqref{HarmonicGauge_1} -\frac{(r-2 M)^2 (6
   M+\lambda  r)}{(\lambda +1) r (3 M+\lambda  r)^2} \partial_r \eqref{HarmonicGauge_2} -\frac{2  (2 M-r) \left(6 M^2-6 M r-\lambda 
   r^2\right)}{(\lambda +1) r^2 (3 M+\lambda  r)^2}\eqref{HarmonicGauge_2}\\
   &+\frac{2  \left(-6 M^2+6 M r+\lambda  r^2\right)}{r^2 (3
   M+\lambda  r)^2}\eqref{HarmonicGauge_3}-\frac{\left(3 M^2+2 \lambda  M r+\lambda  (\lambda +1) r^2\right)}{2 (\lambda +1) r^2 (3
   M+\lambda  r)^2}\eqref{MainEquation_6}\\
   &+\frac{r  (6 M+\lambda  r)}{2 (\lambda +1) (3 M+\lambda  r)^2}\eqref{MainEquation_1}-\frac{\lambda  (r-2 M)^2
   }{2 (\lambda +1) (3 M+\lambda  r)^2}\eqref{MainEquation_2}\\
   &+\frac{2 (r-2 M) }{r (3 M+\lambda  r)}\eqref{MainEquation_5}+\frac{(r-2 M)
   }{2 (\lambda +1) r (3 M+\lambda  r)}\partial_r\eqref{MainEquation_6}\\
   &-\frac{1}{r^2}\eqref{MainEquation_7}.
\end{align*}
From \eqref{def:W0},\eqref{def:W1} and \eqref{def:W2},

\begin{align}
W_{0,m\ell}&=-h_{0,m\ell}+r^2\partial_t G_{m\ell},\\
W_{1,m\ell}&=-\D h_{1,m\ell}+r^2\D \partial_r G_{m\ell},\\
W_{2,m\ell}&=-r^2G_{m\ell}.
\end{align}
Equivalently,

\begin{align*}
h_{0,m\ell}&=-W_{0,m\ell}-\partial_t W_{2,m\ell},\\
h_{1,m\ell}&=-\D^{-1}W_{1,m\ell}-\partial_r W_{2,m\ell}+\frac{2}{r}W_{2,m\ell},\\
G_{m\ell}&=-r^{-2}W_{2,m\ell}.
\end{align*}
To derive the equation for $W_0$, $W_1$ and $W_2$, we first rewrite $H_{1,m\ell}$, $K_{m\ell}$, and $H_{2,m\ell}$ in terms of $h_{0,m\ell}$, $h_{1,m\ell}$, $G_{m\ell}$ and $\psi_{m\ell}$. Then we substitute them in \eqref{MainEquation_4}, \eqref{MainEquation_5} and \eqref{MainEquation_7} to obtain wave equation for $h_{0,m\ell}$, $h_{1,m\ell}$, $G_{m\ell}$. The next step is to replace $h_{0,m\ell}$, $h_{1,m\ell}$, $G_{m\ell}$ by $W_{0,m\ell},\ W_{1,m\ell}$ and $W_{2,m\ell}$.

\begin{lemma}

\begin{align*}
H_{1,m\ell}=&-\frac{2r (3 M - r)	 }{2 M - r}\partial_t G_{m\ell} + \frac{2 M }{
 2 M r - r^2}h_{0,m\ell} + \partial_t h_{1,m\ell}  \\
  +&\frac{3 M^2 + 3 M r \lambda - \lambda r^2  }{(2 M - r) (3 M + \lambda r )} \partial_t \psi_{m\ell} - 
 2 r^2 \partial_t \partial_r G_{m\ell} \\
 +& \partial_r h_{0,m\ell} + 
 r \partial_t\partial_r \psi_{m\ell}.
\end{align*}

\end{lemma}

\begin{proof}
This is (3.39) in \cite{Berndtson}. Or it can be derive by using \eqref{def:psi_mode} and 

\begin{align*}
0= &\frac{M r}{(2 M - r) (\lambda+ 1 )}\eqref{HarmonicGauge_1} + \frac{r^5}{
    2 (2 M - r) ( \lambda+1) (3 M + \lambda r )} \partial_t^2 \eqref{HarmonicGauge_1}\\
    & + \frac{
 r^2 (7 M + r (\lambda -2 )) }{
 2 (\lambda+1  ) (\lambda r+ 3 M  )} \partial_t \eqref{HarmonicGauge_2} + \frac{r^2}{
 \lambda r+ 3 M }\partial_t \eqref{HarmonicGauge_3}\\
 & - \frac{r^5}{
 4 (2 M - r) (\lambda+1 ) (\lambda r+ 3 M)} \partial_t \eqref{MainEquation_1} + \frac{
 r^3 (-2 M + r) }{
 4 (\lambda+1 ) (\lambda r+ 3 M )} \partial_t \eqref{MainEquation_2}\\
 & -\frac{r^2}{
 2\lambda+2} \eqref{MainEquation_3} + \frac{r^2}{
 4 (\lambda+1 ) (\lambda r+ 3 M )} \partial_t \eqref{MainEquation_6}\\
 & + \frac{
 r^2}{
 2\lambda+2 }\partial_r \eqref{HarmonicGauge_1} + \frac{(2 M - r) r^3 }{
 2 (\lambda+ 1) (\lambda r+ 3 M )}\partial_t \partial_r\eqref{HarmonicGauge_2}.
\end{align*}
\end{proof}

\begin{lemma}
\begin{align*}
K_{m\ell}= &-2 (\lambda+ 1) G_{m\ell} + \frac{2 (r-2 M) }{r^2} h_{1,m\ell}  + \frac{(6 M^2 + 3\lambda M r  + 
  \lambda (\lambda+1 ) r^2  ) }{
 r^2 (\lambda r+3 M)}\psi_{m\ell}\\
 & - (2r-4 M)\partial_r G_{m\ell} + \D \partial_r\psi_{m\ell}.
\end{align*}
\end{lemma}

\begin{proof}
This is (3.40) in \cite{Berndtson}. Or it can be derive by using \eqref{def:psi_mode} and 

\begin{align*}
0=&-\frac{r^3}{
  2 (\lambda+ 1 ) (\lambda r+ 3 M)}\partial_t \eqref{HarmonicGauge_1} - \frac{(r-2 M)^2}{(\lambda+ 1) (\lambda r+3 M )}\eqref{HarmonicGauge_2} \\
  &+ \frac{r-2 M}{
 \lambda r+3 M } \eqref{HarmonicGauge_3} + \frac{r^3}{
 4 (\lambda+ 1 ) (\lambda r+ 3 M )}\eqref{MainEquation_1}\\
  &+\frac{r (r-2 M)^2 }{
 4 (\lambda+1) (\lambda r+ 3 M )}\eqref{MainEquation_2} + \frac{r-2 M }{
 4 (\lambda+ 1 ) (\lambda r+ 3 M )}\eqref{MainEquation_6}\\
  &- \frac{
 r (r-2 M)^2}{
 2 (\lambda+1 ) (\lambda r+ 3 M)} \partial_r\eqref{HarmonicGauge_2}.
\end{align*}

\end{proof}

\begin{lemma}
\begin{align*}
H_{2,m\ell}=& -4\lambda G_{m\ell} - \frac{2 r^3}{r-2 M }\partial_t^2 G_{m\ell} + \frac{4 r-6 M } {r^2} h_{1,m\ell} \\
&+ \frac{
    9 M^3 + 9\lambda M^2 r  + 3\lambda^2 M r^2  + 
     \lambda^2 (\lambda+1) r^3 }{r^2 (\lambda r+ 3 M )^2}\psi_{m\ell} + \frac{
    r^2}{r-2M} \partial_t^2\psi_{m\ell} \\
    &+ 
 2 M \partial_r G_{m\ell} + 2\D \partial_r h_{1,m\ell}+\frac{-3 M^2 - 3 \lambda M r  +\lambda r^2  }{r (\lambda r +3 M)}\partial_r\psi_{m\ell}.
\end{align*}
\end{lemma}

\begin{proof}
This is (3.41) in \cite{Berndtson}. Or it can be derive by using \eqref{def:psi_mode} and 

\begin{align*}
&-\frac{r^4 }{2 (\lambda +1) (3 M+\lambda  r)}\partial_t\partial_r  \eqref{HarmonicGauge_1} -\frac{r^3 \left(3 M^2+3 M (2
   \lambda  r+r)+\lambda ^2 r^2\right)}{2 (\lambda +1) (2 M-r) (3 M+\lambda  r)^2} \partial_t \eqref{HarmonicGauge_1}\\
   &-\frac{r^2 (r-2 M)^2 }{2
   (\lambda +1) (3 M+\lambda  r)}\partial_r^2\eqref{HarmonicGauge_2}-\frac{r (2 M-r)  \left(3 M^2+3 (2 \lambda -1) M r+(\lambda -2) \lambda 
   r^2\right)}{2 (\lambda +1) (3 M+\lambda  r)^2}\partial_r\eqref{HarmonicGauge_2}\\
   &+\frac{ (2 M-r) \left(3 M^2-2 (2 \lambda +3) M r-\lambda 
   (\lambda +1) r^2\right)}{(\lambda +1) (3 M+\lambda  r)^2}\eqref{HarmonicGauge_2}+\frac{r (r-2 M) }{3 M+\lambda 
   r}\partial_r \eqref{HarmonicGauge_3}\\
   &-\frac{\left(-3 M^2+(4 \lambda +9) M r+\lambda  (\lambda +2) r^2\right)}{(3 M+\lambda  r)^2}\eqref{HarmonicGauge_3}+\frac{r^3
   \left(3 M^2+3 M (2 \lambda  r+r)+\lambda ^2 r^2\right)}{4 (\lambda +1) (2 M-r) (3 M+\lambda  r)^2}\eqref{MainEquation_1}\\
   &-\frac{r
   (r-2 M)  \left(15 M^2+(10 \lambda -9) M r+(\lambda -4) \lambda  r^2\right)}{4 (\lambda +1) (3 M+\lambda 
   r)^2}\eqref{MainEquation_2}+\frac{ \left(-3 M^2+3 M r+\lambda  (\lambda +2) r^2\right)}{4 (\lambda +1) (3 M+\lambda 
   r)^2}\eqref{MainEquation_6}\\
   &+\frac{r^4 }{4 (\lambda +1) (3 M+\lambda  r)}\partial_r\eqref{MainEquation_1}+\frac{r^2 (r-2 M)^2 }{4 (\lambda +1) (3
   M+\lambda  r)}\partial_r\eqref{MainEquation_2}\\
   &+\frac{2 r (2 M-r) }{3 M+\lambda  r}\eqref{MainEquation_5}+\frac{r (2 M-r) }{4 (\lambda +1) (3
   M+\lambda  r)}\partial_r\eqref{MainEquation_6}.
\end{align*}
\end{proof}

By substituting $H_{1,m\ell}$ in \eqref{MainEquation_4}, we obtain

\begin{equation}\label{reduced_equation1}
\begin{split}
{}^{0}\Box_\ell h_{0,m\ell}=& -4 (r-2M)\partial_t\partial_r G_{m\ell}  -4 \left(1-\frac{3 M}{r}\right)\partial_t G_{m\ell}\\
&-\frac{2 M }{r^2}\partial_r h_{0,m\ell}+\frac{2 \omega 
   (r-M)}{r^2}h_{1,m\ell}\\
   &+\frac{2  \left(-3 M^2-3 \lambda  M r+\lambda  r^2\right)}{r^2 (3 M+\lambda 
   r)}\partial_r\psi_{m\ell}+\left(2-\frac{4 M}{r}\right)\partial_t\partial_r \psi_{m\ell}.
\end{split}
\end{equation}
By substituting $K_{m\ell}$ and $H_{2,m\ell}$ we have

\begin{equation}\label{reduced_equation2}
\begin{split}
{}^0 \Box_\ell h_{1,m\ell}=&\left(4-\frac{4 M}{r}\right) \partial_r G + \frac{4 r^2 }{2 M-r}\partial_t^2 G_{m\ell}\\
&+\frac{4}{r}G_{m\ell}+\frac{2 M }{(r-2 M)^2}\partial_t h_{0,m\ell}\\
&-\frac{2 (M-r) }{r^2}\partial_r h_{1,m\ell}+\frac{4 M }{r^3}h_{1,m\ell}\\
&-\frac{6 M
   \left(3 M^2+2 \lambda  M r+\lambda  (\lambda +1) r^2\right)}{r^3 (3 M+\lambda  r)^2}\psi_{m\ell}-\frac{2 r }{2
   M-r}\partial^2_t\psi_{m\ell}\\
 &  +\frac{2 M (3 M-(\lambda +3) r) }{r^2 (3 M+\lambda  r)}\partial_r\psi_{m\ell}.
\end{split}
\end{equation}
By replacing $G_{m\ell}$ and $h_{1,m\ell}$ by $W_{2,m\ell}$ and $W_{1,m\ell}$ in \eqref{MainEquation_7}, we obtain

\begin{align*}
^{1}\Box_\ell W_{2,m\ell}=\frac{1}{r^2}\D W_{2,m\ell}+\frac{2}{r}W_{1,\ell}.
\end{align*}
Summing over all $m\ell$ with $\ell\geq 2$ and using Lemma \eqref{lem:d_mode}, we obtain \eqref{equ:W_2}. Similarly, by replacing $h_{0,m\ell}$, $h_{1,m\ell}$ and $G_{m\ell}$ in \eqref{reduced_equation1}, we obtain

\begin{align*}
{}^0\Box_\ell W_{0,m\ell}=&\frac{2M}{r^2}\partial_r W_{0,m\ell}-\frac{2M}{r^2}\D^{-1}\partial_t W_{1,m\ell}\\
&+\frac{2 \left(-3 M^2-3 \lambda  M r+\lambda  r^2\right)}{r^2 (\lambda r+ 3 M)}\partial_t\psi_{m\ell}+2\D\partial_t\partial_r\psi_{m\ell}\\
&=-\frac{2M}{r^2}P_{\even,m\ell}-\frac{2M}{r^2}W_{0,m\ell}\\
&+\frac{2 \left(-3 M^2-3 \lambda  M r+\lambda  r^2\right)}{r^2 (\lambda r+ 3 M)}\partial_t\psi_{m\ell}+2\D\partial_t\partial_r\psi_{m\ell}.
\end{align*}
Here we used the definition of $P_{\even}$, \eqref{def:P}. Then \eqref{equ:W_0} follows from summing over $m\ell$ with $\ell\geq 2$. Similarly, from \eqref{reduced_equation2} we obtain

\begin{align*}
{}^0\Box_\ell W_{1,m\ell}=&-\frac{2M}{r^2}\D^{-1}\partial_t W_{0,m\ell}+\frac{2M}{r^2}\partial_r W_{1,m\ell}-\frac{4\lambda+4}{r^3}\D W_{2,m\ell}+\frac{2}{r^2}\D W_{1,m\ell}\\
&+\frac{6 M (2 M-r) \left(3 M^2+2 \lambda  M r+\lambda  (\lambda +1) r^2\right)}{r^4 (3 M+\lambda  r)^2}\psi_{m\ell} -\frac{2 M (2 M-r) (3 M-(\lambda +3) r)}{r^3 (3 M+\lambda  r)}\partial_r\psi_{m\ell}\\
&+2\partial_t^2\psi_{m\ell}\\
=&\frac{M}{r^2}S_{W,m\ell}-\frac{4\lambda+4}{r^3}\left(1-\frac{3M}{r}\right)W_{2,m\ell}+\frac{2}{r^2}\left(1-\frac{4M}{r}\right)W_{1,m\ell}\\
&+\frac{6 M (2 M-r) \left(3 M^2+2 \lambda  M r+\lambda  (\lambda +1) r^2\right)}{r^4 (3 M+\lambda  r)^2}\psi_{m\ell} -\frac{2 M (2 M-r) (3 M-(\lambda +3) r)}{r^3 (3 M+\lambda  r)}\partial_r\psi_{m\ell}\\
&+2\partial_t^2\psi_{m\ell}.
\end{align*}
Here we used

\begin{align*}
S_{W,m\ell}=-2\D^{-1}\partial_t W_{0,m\ell}+2\partial_r W_{1,m\ell}-\frac{4\lambda+4}{r^2}W_{2,m\ell}+\frac{4}{r}W_{1,m\ell}.
\end{align*}
By summing over all $m\ell$ with $\ell\geq 2$ and using Lemma \eqref{lem:d_mode}, one obtains \eqref{equ:W_1}. \\

Now we start to derive the equation for $S_W$, $P_{\even}$ and $Q_{\even}$. $S_W-\slashed{tr}h^{\RW}$ satisfies free wave equation since $\slashed{tr}h^{\RW}={tr}h^{\RW}$ and ${tr}h^{\RW}-S_W=S$. The equation for $P_{\even,m\ell}$ follows the the one for $W_{0,m\ell}$ and $W_{1,m\ell}$. We compute

\begin{align*}
{}^0\Box P_{\even,m\ell}=&-2\D^{-1}\partial^2_t W_{0,m\ell}+2\partial_t \partial_r W_{1,m\ell}-\frac{4\lambda+4}{r^2}\partial_t W_{2,m\ell}+\frac{4}{r}\partial_t W_{1,m\ell}\\
                         &+\frac{4 \left(3 M^2+\lambda  (\lambda +2) r^2\right)}{r^2 (3 M+\lambda  r)}\partial_t\psi_{m\ell}-8\D\partial_t\partial_r\psi_{m\ell}.
\end{align*}
From the expression of $S_{W,m\ell}$ and

\begin{align*}
(\slashed{tr}h^{\RW})_{m\ell}=\frac{2\lambda+2}{r}\psi_{m\ell}+2\D\partial_{r}\psi_{m\ell}-\frac{6M}{r(\lambda r+3M)}\D\psi_{m\ell}.
\end{align*}
we compute

\begin{align*}
{}^0\Box_{\ell} P_{\even,m\ell}=&\partial_t(S_{W,m\ell}-(\slashed{tr}h^{\RW})_{m\ell})\\
                        -&\frac{2 \lambda  ((\lambda +3) r-3 M)}{r (3 M+\lambda  r)}\partial_t\psi_{m\ell}-6\D\partial_t\partial_r\psi_{m\ell}.
\end{align*}
Then after summing over all $m\ell$ with $\ell\geq 2$, one has \eqref{equ:P}. From the definition of $Q_{\even}$, we have

\begin{align*}
Q_{\even,m\ell}=-W_{1,m\ell}+\D \partial_r W_{2,m\ell}+\frac{r}{2}S_{W,m\ell}-\frac{2r}{3}\left(\lambda+\frac{3M}{r}\right)\cdot\frac{1}{r}\psi_{m\ell}.
\end{align*}
$Q_{\even}$ corresponds to $\psi_1$ defined in  \cite[(3.63)]{Berndtson}. And the equation for $Q_{\even}$ is equivalent to \cite[(3.61)]{Berndtson}. Alternative, we can also derive it through the equation of $W_{1,m\ell},\ W_{2,m\ell},\ S_{W,m\ell}$ and $\psi_{m\ell}$ as

\begin{align*}
{}^1\Box_\ell Q_{\even,m\ell}-\frac{1}{r^2}\D Q_{\even,m\ell}=0.
\end{align*}

Then \eqref{equ:Q} follows by summing over all $m\ell$ with $\ell\geq 1$.

\end{appendix}

\bibliography{stability_paper}{}
\bibliographystyle{plain}

\end{document}